%% file: main.tex
\documentclass[11pt]{article} \usepackage{fullpage} %\usepackage{setspace} \doublespacing %\onehalfspace

\usepackage[cmex10]{amsmath}
\usepackage{amsthm}
\usepackage{amssymb}
\usepackage{thmtools}
\usepackage{subfigure}
\usepackage{balance}

\usepackage{tikz}
\usetikzlibrary{shapes, arrows, positioning, decorations.markings}

\usepackage{wasysym}

\newtheorem{algorithm}{Algorithm}

\newcommand{\N}{ {\rm ne}}

\tikzstyle{int}=[draw, fill=blue!20, minimum size=2em]
\tikzstyle{init} = [pin edge={to-,thin,black}]

\tikzset{->-out/.style={decoration={
  markings,
  mark=at position #1 with {\arrow{triangle 45}}},postaction={decorate}}}
\tikzset{->-in/.style={decoration={
  markings,
  mark=at position #1 with {\arrowreversed{triangle 45}}},postaction={decorate}}}

\tikzset{--o/.style={decoration={
  markings,
  mark=at position .95 with {\arrowreversed[line width=0mm]{*}}},postaction={decorate}}}

\tikzset{o--/.style={decoration={
  markings,
  mark=at position 0 with {\arrowreversed[line width=0mm]{*}}},postaction={decorate}}}

\newcommand{\clbox}[4]{+(#1,#2) rectangle +(#3,#4);}

\newcommand{\C}{\mathbb{C}}
\newcommand{\F}{\mathbb{F}}

\newcommand{\X}{{\cal X}}

\newcommand{\G}{{\cal G}}

\newcommand\independent{\protect\mathpalette{\protect\independenT}{\perp}} 
\def\independenT#1#2{\mathrel{\rlap{$#1#2$}\mkern2mu{#1#2}}}

\newcommand{\indep}{\perp\!\!\!\!\perp}

\newtheorem{Theorem}{{Theorem}}
\newtheorem{Lemma}{{Lemma}}

\newtheorem{Prop}{{Proposition}}

\newtheorem{Example}{Example}

\begin{document}
\renewcommand\thmcontinues[1]{Continued}
\bibliographystyle{IEEEtran}

\title{Normal Factor Graphs as Probabilistic Models}

\author{
Ali Al-Bashabsheh and Yongyi Mao\\
School of Electrical Engineering and Computer Science\\
University of Ottawa, Canada\\
\{aalba059, yymao\}@site.uottawa.ca\\
}

\date{}
\maketitle	

\input{abstract}

\input{Intro}

\input{prob_models}

\input{NFGs}

\input{NFGs_Prob2}

\input{transformedNFG}
\input{coding}
\input{infer}

\input{conc}

\balance

%\bibliography{Holographic_revised}
\bibliography{C:/Dropbox/bibliographys/Holographic_revised}

\end{document}

%% file: abstract.tex
\begin{abstract}
We present a new probabilistic modelling framework based on the recent notion of normal factor graph (NFG). We show that the proposed NFG models and
their transformations unify some existing models such as factor graphs, convolutional factor graphs, and cumulative distribution
networks. The two subclasses of the NFG models, namely the constrained and generative models,  exhibit a duality in their dependence structure.
Transformation of NFG models further extends the power of this modelling framework.
We point out the well-known NFG representations of parity and generator realizations of a linear code as generative and
constrained models, and comment on a more prevailing duality in this context. 
%We reinterpret the well-known NFG representation of parity and generator realizations of linear codes in the framework of
%constrained and generative models and comment on a more prevailing duality in this context.
%A more prevailing duality between 
%
Finally, we address the algorithmic aspect of computing the exterior function of NFGs and the inference problem on NFGs.

%More specifically, we introduce the notions of constrained and generative NFGs, and show
%that while the FGs and constrained NFGs realize exactly the same class of exterior functions,
%CFGs and CDNs realize a subclass of the exterior functions realized by the generative NFGs.
%Further, we rederive the implied statistical independence of such graphical models on the random variables they represent
%using the diagrammatic approach of NFGs. The unifying approach of NFGs
%explains the identical implied independence
%of CFGs and CDNs as part of a larger class of 
%probabilistic models, namely the class of generative NFGs, that share the same independence implications.
%More specifically, we introduce the notions of constrained and generative NFGs, and show
%that while FGs and constrained NFGs are equivalent,
%CFGs and CDNs are special cases of transformed generative NFGs.
%Further, we rederive the implied statistical independence of such graphical models under the unifying approach of NFGs, which
%explains the identical implied independence
%of CFGs and CDNs as part of a larger class that share the same independence implications, namely, the class of transformed generative NFGs.
\end{abstract}

%% file: Intro.tex
\section{Introduction}

In the recent years, probabilistic graphical models have emerged from different disciplines as a powerful 
methodology for statistical inference and machine learning. Traditional such models, such as Bayesian networks 
\cite{Pearl:1988} and Markov random fields \cite{MarkovRF}, primarily aim at representing the joint
probability distribution (i.e., probability mass function or probability density function) of the random variables (RVs)  of interest in terms of their 
multiplicative factorization structure. Such ``multiplicative'' modelling semantics can be translated to the language of factor graphs (FGs) \cite{frank:factor},
a mathematical and graphical framework that is convenient and intuitive for representing the multiplicative factorization of a multivariate function. 
It is arguable that FGs and their variants, such as directed factor graphs \cite{FreyDirectedFG}, unify the various such multiplicative models
\cite{frank:factor, FreyDirectedFG}.  In contrast to the multiplicative models, convolutional factor graphs (CFGs) \cite{Mao:UAI2004, Mao2005:FGFT} are models which 
represent the joint distribution of interest in terms of convolutional factorizations. The CFG modelling framework has recently demonstrated
its power in a derivative of the CFG model, known as linear characteristic models (LCM) \cite{Bickson}, for inference in stable distributions.
Instead of directly representing the distribution of interest, LCM represents the characteristic function of the distribution. 
This is advantageous for stable distributions, which are only explicitly defined in the characteristic function domain.  We argue that the
philosophy of modelling in a ``transform domain'', as manifested in LCM,  should not be overlooked. This is because the unique nature of 
an inference task one is faced with may favour  a representation of some other objects than the probability distribution. Incidentally or 
not, cumulative distribution functions, which may be viewed as transformations of probability distributions, appear more favourable in 
structured ranking problems, and this recognition has led to the development of cumulative distribution networks (CDNs) \cite{Frey:CDN2008}.

In this paper, we present a new graphical model, the normal factor graph model, based on the notion of normal factor graphs  
(NFGs) \cite{Bashabsheh:HolTrans, Forney:MacWilliams2}. 
In the framework of NFGs, a powerful tool, called holographic transformation, has been developed . 
It was shown in \cite{Bashabsheh:HolTrans} that this tool unifies a duality theorem of Forney \cite{Forney2001:Normal} 
in coding theory and the Holant theorem of Valiant \cite{Valiant2004:Holographic} in complexity theory. 

The main objective of this paper is to show that the proposed NFG models, together with the holographic transformation technique, 
essentially unifies all the probabilistic models mentioned above.  We will focus on two subfamilies of NFG models, 
namely constrained and generative NFG models. We will show that constrained NFG models reduce to FGs, however, have a
different interpretation, and that generative NFG models, restricted to a special case, reduce to  CFGs. 
In addition, we reveal an interesting ``duality'' between the constrained and generative NFG models in their independence properties. 
A general model transformation technique is introduced, using which we show a CDN is equivalent to a transformed NFG model.

%% file: prob_models.tex
\section{Probabilistic Graphical Models}
\label{section:GMs}

Here we give a brief summary of the previous graphical models relevant to this work. 
As a notational convention that will be used throughout the paper, a RV is denoted by a capitalized letter, for example, 
by $X$, $Y, \ldots$, and the value it takes will be denoted by the corresponding lower-cased letter, i.e.,  $x, y, \ldots$.

\subsection{Factor Graphs}
\label{section:FGs}
%\noindent{\bf Factor Graphs}
A \emph{factor graph} (FG) \cite{frank:factor} is a bipartite graph $(V \cup U, E)$ with independent vertex sets $V$ and $U$, and edge set $E$, 
where each
vertex $v \in V$ is associated a variable $x_v$ from a finite alphabet $\X_v$, and each vertex $u \in U$ is associated
a complex-valued function $f_u$ on the cartesian product $\X_{{\rm ne}(u)}: = \prod \limits_{v \in {\rm ne}(u)} \X_{v}$, where 
${\rm ne}(u): = \{v \in V : \{u,v\} \in E\}$ is the set of neighbors (adjacent vertices) of  $u$.
%From this it is clear that for any $v \in V$ and $u \in U$, we have $\{u,v\}$ is an edge in $E$ if and only if
%$x_v$ is an argument of $f_u$. 
Each function $f_u$ is referred to as a \emph{local} function and the FG is said to \emph{represent}
a function given by
$
f(x_{V}): = \prod_{u \in U} f_{u}(x_{{\rm ne}(u)}), 
$
where we use the ``variable set" notation, defined for any $A \subseteq V$, as $x_{A} := \{x_a : a \in A\}$. In the context of FGs, the function represented by the FG is often called the {\em global function}.
Fig.~\ref{fig:GM}~(a) is  an example FG.

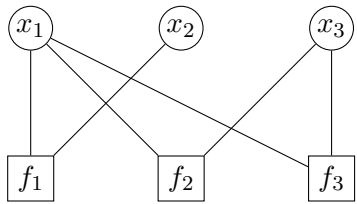
\begin{figure}[ht]
\centering
\begin{tikzpicture}[var/.style={node distance=2cm, draw, circle, inner sep=.5mm, minimum size=4mm}, fun/.style={node distance=2cm, draw, rectangle, minimum size=4mm}]
  \node(x1)[var]{$x_1$}; \node(x2)[var, right of=x1]{$x_2$}; \node(x3)[var, right of=x2]{$x_3$};
  \node(f1)[fun, below of=x1]{$f_1$}; \node(f2)[fun, right of=f1]{$f_2$}; \node(f3)[fun, right of=f2]{$f_3$};
  \path[every node/.style={}]
  (x1)edge(f1)edge(f2)edge(f3)
  (x2)edge(f1)
  (x3)edge(f2)edge(f3);
\end{tikzpicture}
\caption{An example of an FG, a CFG, and a CDN: (a) When viewed as an FG, the graph represents the
global function $f_{1}(x_1,x_2) f_{2}(x_1,x_3) f_{3}(x_1,x_3)$, (b) as
a CFG, the graph represents the global function
$f_{1}(x_1,x_2)*f_{2}(x_1,x_3)*f_{3}(x_1,x_3)$, and (c) as a CDN, the graph is understood as
an FG where each local function is a cumulative distribution, in which case,
the global function $f_{1}(x_1,x_2) f_{2}(x_1,x_3) f_{3}(x_1,x_3)$ satisfies the properties of a cumulative function, and is taken as the
joint cumulative distribution of the RVs $X_1,X_2$ and $X_3$.}
\label{fig:GM}
\end{figure}

Since independence (or conditional independence) relationships among RVs are often captured via the multiplicative factorization
of their joint probability distribution, FGs, when used to represent the joint distribution of RVs, form a convenient probabilistic model.  

The relationship between FG probabilistic model and other classical 
probabilistic models, such as Bayesian networks and Markov random fields, is well-known, see, e.g. \cite{frank:factor}.  In these models, all
featuring the "multiplicative semantics'' and aiming at representing the joint distributions, efficient inference algorithms, such as the belief
propagation or the sum-product algorithm, have been developed and demonstrated great power in various applications.

%In this work we 

\subsection{Convolutional Factor Graphs}

Let $\X_1, \X_2$ and $\X_3$ be 
three (possibly distinct) finite sets. In general, we require the sets to have an ``abelian group'' structure, so that a notion of addition ``$+$''
and its inverse ``$-$'' are well defined. The requirement that the sets be finite is not particularly critical but only for the
convenience of argument. 
%In fact, extending the results of this paper to more general settings (such as for continuous-valued random variables) is mainly a matter of technicality.

Let $f_1$ and $f_2$ be two function on $\X_1 \times \X_2$ and $\X_2 \times \X_3$, respectively. 
The \emph{convolution} of $f_1$ and $f_2$, denoted $f_1*f_2$, is defined as the function on $\X_1 \times \X_2 \times \X_3$ given by, %\cite{Mao2005:FGFT}
$
(f_1*f_2)(x_1,x_2,x_3) := \sum_{x\in \X_2} f_1(x_1,x_2-x)f_2(x,x_3).
$
Following the convention in \cite{Mao:UAI2004}, we may write $(f_1*f_2)(x_1,x_2,x_3)$ as $f_1(x_1,x_2)*f_2(x_2,x_3)$ to emphasize the domains of the original functions. It is not hard to show that the convolution as defined above is both associative and commutative. 
%Hence, it is possible to write the convolution of any number of functions in any order, which justifies the subsequent use of $\prod\limits_{i=1,\ldots,n}^{*}f_i(x_{A_i})$ to denote
%$
%f_1(x_{A_1})*f_2(x_{A_2})*\ldots*f_n(x_{A_n}),
%$
%where $x_{A_{i}}$ is the set of arguments of the function $f_i$.

A \emph{convolutional factor graph} (CFG) \cite{Mao2005:FGFT} is a bipartite graph that represents a global function that 
factors as the convolution of local functions. In fact, the representation semantics in a CFG is identical to that in an FG (often referred to as a ``multiplicative'' FG for distinction), except that the above defined notion of convolution is used as the product operation, cf. Fig.~\ref{fig:GM}~(b) for an example CFG. 
In \cite{Mao:UAI2004} CFGs were presented as
a probabilistic graphical model to represent the joint probability distribution of a set of observed RVs that are constructed from a collection of
independent sets of latent RVs via linear combinations. In addition,  the authors of   \cite{Mao2005:FGFT} presented an elegant duality result between
FGs and CFGs via the Fourier transform.  Such a duality and CFGs have recently been exploited by \cite{Bickson} in what is known as linear
characteristic model (LCM) for solving inference problems with stable distributions. 

%and the statistical independence of the RVs represented by such model was investigated.
%In this work, we will demonstrate the NFG rout toward this problem via the external function semantics, which at least to us appears more natural.

\subsection{Cumulative Distribution Networks}
Let $X$ be a RV assuming its values from a finite ordered set $\X$. The \emph{cumulative distribution function} (CDF) of $X$ is defined as
$F_{X}(x) := \sum_{y \leq x} p_{X}(y)$, where $p_{X}$ is the probability distribution of $X$.  We note that this definition of 
CDF, as a function on $\X$, is slightly different from the classical definition of CDF, which is a function defined on the real 
line (or on the Euclidean space in the multivariate case). It nevertheless captures the same essence and is merely a different 
representation, suitable and convenient in the context of this paper.   Such a notion of CDF  can be extended to any collection
of RVs $X_1, \ldots, X_n$ assuming their values from the finite ordered sets $\X_1, \ldots, \X_n$ by defining their joint CDF
as $F_{X_1, \ldots, X_n}(x_1, \ldots, x_n) := \sum_{y_1 \leq x_1, \ldots, y_n \leq x_n} p_{X_1, \ldots, X_n}(y_1, \ldots, y_n)$,
where $p_{X_1, \ldots, X_n}$ is the joint probability distribution of $X_{1}, \ldots, X_n$. 
Note that while the marginal probability distribution is computed by \emph{summing} the joint probability distribution over the range of the
marginalized RVs, the marginal CDF is computed by \emph{evaluating} the joint CDF at the largest element of $\X_i$, for all marginalized
RVs $X_i$. (That is, if $I$ indexes the set of marginalized RVs and $X_i$ takes its values from the ordered set $\left\{ 1,
\ldots, |\X_i| \right\}$, then we evaluate the joint CDF at $|\X_i|$ for all $i \in I$.)
It is well known that CDFs satisfy 
a collection of properties as were articulated in standard textbooks and in \cite{Frey:CDN2008}. On the other hand,  any function 
satisfying such properties, which we shall refer to as ``CDF axioms'', may be regarded as a CDF and can be used to define a collection of RVs.

A \emph{cumulative distribution network} (CDN) \cite{Frey:CDN2008} is a multiplicative FG in which each local function satisfies 
the CDF axioms, then it is straightforward to show that the global function represented by the FG also satisfies the CDF axioms. 
The global function thus defines a collection of random variables, each represented by a variable node in the FG, and the CDN may 
serve as a probabilistic model. In \cite{Frey:CDN2008}, it was shown that CDNs are useful for structured ranking problems, and 
efficient inference algorithms for such problems were developed in these models. See Fig.~\ref{fig:GM}~(c) for an example CDN.

%% file: NFGs.tex
\section{Normal Factor Graphs}
\label{section:NFGs}
%\vspace{.3cm}

Now we give a quick overview of the framework of normal factor graphs (NFGs), %recently formulated in  \cite{Bashabsheh:HolTrans, Forney:MacWilliams2} 
and develop some notations and definitions for subsequent discussions.

\subsection{NFG and the exterior function}
%\noindent\textbf{Normal Factor Graphs} 
A \emph{normal factor graph} (NFG) \cite{Bashabsheh:HolTrans, Forney:MacWilliams2} is a graph $(V,E)$, 
with vertex set $V$ and edge set $E$, where the edge set $E$ consists of two types of edges, a set $T$ 
of regular edges (also called internal edges), each incident on two vertices, and a set $L$ of ``half edges" 
(also called external edges or dangling edges), each incident on exactly one vertex. Every edge $e \in E$ is 
associated a variable $x_e$ from a finite alphabet $\X_e$, and every vertex $v \in V$ is associated a \emph{local} 
function $f_v$ on the cartesian-product $\X_{E(v)} := \prod\limits_{e \in E(v)} \X_e$, where $E(v)$ 
is the set of (internal and external) edges incident on $v$.
At some places we may use $T(v)$ and $L(v)$ to denote the internal and external edges incident 
on vertex $v$, respectively, and it is clear that $E(v) = T(v) \cup L(v)$ for all $v$.
We use the symbol $\G$ to refer to an NFG, and sometimes write $\G(V,E, f_{V})$, where
$f_{V} := \{f_{v} : v \in V\}$, to emphasize the NFG parameters.
 
%At some places we find it practical to graphically mark the first argument of a local function with a ciliation and assume the rest of its arguments are encountered in a counter clock-wise manner.
An NFG $\G$ is associated with a function, called the {\em exterior function} and denoted by $Z_{\G}$, on the cartesian product $\X_{L} := \prod \limits_{e \in L} \X_{e}$, defined as
\[
Z_{\G}(x_{L}) := \sum_{x_{T}} \prod_{v \in V} f_v(x_{E(v)}).
\]
That is, the exterior function realized by an NFG is the product of all its local functions with the internal variables (edges) summed over.
An example NFG is shown in Fig.~\ref{fig:NFG_def}.
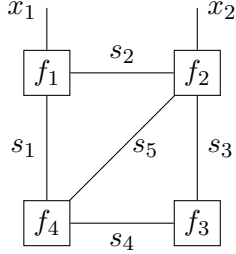
\begin{figure}[ht]
\centering
\begin{tikzpicture}[v/.style={node distance=2cm, draw, rectangle, minimum size=5mm}, d/.style={node distance=1cm}]
  \node(v1)[v]{$f_1$}; \node(v2)[v, right of=v1]{$f_2$};
  \node(v4)[v, below of=v1]{$f_4$}; \node(v3)[v, below of=v2]{$f_3$};
  \node(d1)[d, above of=v1]{}; \node(d2)[d, above of=v2]{};
  \path[every node/.style={}]
  (v1)edge node[above]{$s_2$}(v2) (v2)edge node[right]{$s_3$}(v3) (v2)edge node[right]{$s_5$}(v4) (v3)edge node[below]{$s_4$}(v4)
  (v4)edge node[left]{$s_1$}(v1)
  (v1)edge node[above left]{$x_1$}(d1) (v2)edge node[above right]{$x_2$}(d2)
  ;   
  \end{tikzpicture}
\caption{\small An NFG realizing the exterior function $\sum \limits_{s_1, \ldots, s_5} f_1(x_1,s_1,s_2) f_2(x_2,s_2,s_3,s_5) f_3(s_3,s_4)f_4(s_1,s_4,s_5)$.}
   \label{fig:NFG_def}
\end{figure}
 
%The function $Z_{\G}$ is referred to as the \emph{exterior function}\footnote{Inspired by the statistical physics literature, this function
%is also referred to as the \emph{partition function} in \cite{Forney2011:PartitionFunction}.} realized by $\G$.
Let $V = \{1, \ldots, |V|\}$, at some places for notational convenience, we may denote the right hand side of the above equation,
a ``sum-of-products form," by
$
\langle f_1, f_2, \ldots, f_{|V|} \rangle.
$
This notation is valid due to the distribution law, and the associativity and commutativity of addition and multiplication, making the bracketing and ordering of the arguments in the notation $\langle \cdot, \cdots, \cdot \rangle$ irrelevant. For example, the sum-of-products form encoded by the NFG in
Figure \ref{fig:NFG_def} may be written as $\langle f_1(x_1,s_1,s_2), f_2(x_2,s_2,s_3,s_5), f_3(s_3,s_4), f_4(s_1,s_4,s_5)\rangle$, or even $\langle f_1, f_2, f_3, f_4\rangle$ for simplicity.

We say that two NFGs are \emph{equivalent} if they realize the same exterior function.  At some places, we may extend this notion of equivalence to include other
graphical models and say, for instance, ``an FG is equivalent to an NFG,"  where we mean that the
product function of the FG is equal to the exterior function of the NFG. 

Finally, we call an NFG with no loops or parallel internal edges a \emph{simple} NFG, and an NFG with a 
bipartite underlying graph $(I \cup J, E)$ a \emph{bipartite} NFG, where $I$ and $J$ are the two independent
vertex sets. We impose no restriction on the cycle structure of NFGs.

\subsection{Special kinds of local functions}
The following (non-disjoint) classes of local functions will be of particular interest.

\noindent\textbf{\underline{Split functions}}
Let $\X_1, \ldots, \X_n$ be some finite sets, then we say a function $f$ on $\X_1 \times \cdots \times \X_n$
is a \emph{split} function via $x_1$, and refer to $x_1$ as the \emph{splitting variable}, if 
$$f(x_1, \ldots, x_n) = f_2(x_1,x_2) f_3(x_1,x_3) \ldots f_n(x_1,x_n),$$
for some bivariate functions
$f_2, \ldots, f_n$.
Note that it follows immediately that any bivariate function is trivially a split function (via any of its arguments).
Subsequently, if we do not explicitly specify the splitting variable of a split function,
then it is assumed to be the function's first argument.
Graphically, we draw a split function as in Fig.~\ref{fig:Ind}~(a), where
the in-ward directed edge is used to distinguish the splitting argument, and the
remaining arguments are successively encountered in a counter clock-wise manner with respect to the directed edge.

%\begin{figure}[ht]
%\centering
%\subfigure[]{\includegraphics[scale = .7]{fig/Split_fun_fig.eps}	} 			\hspace{2cm}
%\subfigure[]{\includegraphics[scale = .7]{fig/Cond_fun_fig.eps} 	} 		
%\caption{A graphical illustration of: (a) a split function, and (b) a conditional function.}
%\label{fig:split-cond}
%\end{figure}

\noindent\textbf{\underline{Conditional functions}}
Let $\X_1, \ldots, \X_n$ be some finite sets, then a function $f$ on $\X_1 \times \cdots \times \X_n$
is said to be a \emph{conditional} function of $x_1$ given $x_2, \ldots, x_n$
%, and refer to $x_1$ as the controlled variable, 
if  there is a constant $c$ such that $\sum_{x_1} f(x_1, \ldots, x_n) = c,$ for all $x_2, \ldots, x_n$.
It is apparent that a non-negative real conditional function with $c = 1$ is a conditional probability distribution.
%if $\sum_{x_1} f(x_1, \ldots, x_n) = f_2(x_2) \ldots f_{n}(x_n)$, for some univariate functions
%$f_2, \ldots, f_n$. It is apparent that a non-negative conditional function
%with $\prod_{i} f_i(x_i) = 1$ for all $(x_2, \ldots, x_n) \in \X_2 \times \ldots \times \X_n$ is a conditional distribution.
A conditional function is shown in Fig.~\ref{fig:Ind}~(b), where we use the same convention
of edge labeling as in the case of split
functions, but with an out-ward directed edge to mark the first argument. 
%We remark that it is possible that a function might be viewed as a conditional function 
%of one of its arguments that is not necessarily unique. An extreme example in this direction
%is a function that factors as the multiplication $f_1(x_1) \ldots f_n(x_n)$, which is easily verified to be a  
%conditional function of any one of its arguments conditioned on the rest.

%
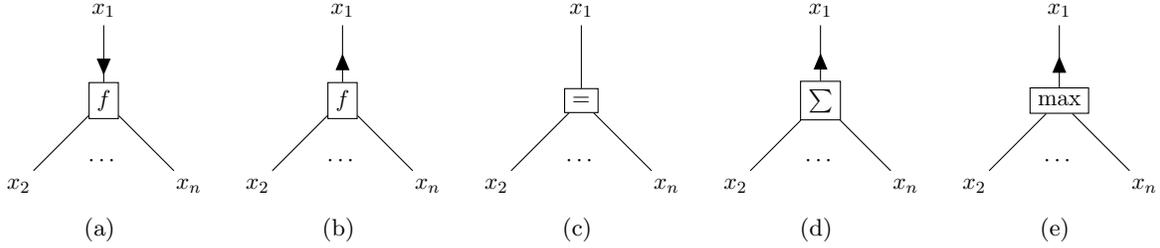
\begin{figure}[ht]
\centering
%\subfigure[]{\includegraphics[scale = .63]{fig/Split_fun_fig.eps}	} \hspace{0cm}
%\subfigure[]{\includegraphics[scale = .63]{fig/Cond_fun_fig.eps} 	} \hspace{0cm}		
%\subfigure[]{\includegraphics[scale = .63]{fig/Ind_Eq_fig.eps}} 			\hspace{0cm}
%\subfigure[]{\includegraphics[scale = .63]{fig/Ind_Sum_fig.eps}}			\hspace{0cm}
%\subfigure[]{\includegraphics[scale = .63]{fig/Ind_max_fig.eps}}		
  \subfigure[]{
  \begin{tikzpicture}[scale=.8, v/.style={node distance=2cm, draw, rectangle},
    d/.style={node distance=2cm}, hdts/.style={node distance=1cm}, every node/.append style={transform shape}]
    \node(f)[v]{$f$}; 
    \node(d1)[d, node distance=1.5cm, above of=f]{$x_1$};
    \node(d2)[d, below left of=f]{$x_2$}; \node(dn)[d, below right of=f]{$x_n$};
    \path[every node/.style={transform shape}]
    (f)edge[->-in=.2](d1) edge(d2) edge(dn)
    node[below of=f]{$\cdots$};
  \end{tikzpicture} }
  \subfigure[]{
  \begin{tikzpicture}[scale=.8, v/.style={node distance=2cm, draw, rectangle},
    d/.style={node distance=2cm}, hdts/.style={node distance=1cm}, every node/.append style={transform shape}]
    \node(f)[v]{$f$}; 
    \node(d1)[d, node distance=1.5cm, above of=f]{$x_1$};
    \node(d2)[d, below left of=f]{$x_2$}; \node(dn)[d, below right of=f]{$x_n$};
    \path[every node/.style={transform shape}]
    (f)edge[->-out=.5](d1) edge(d2) edge(dn)
    node[below of=f]{$\cdots$};
  \end{tikzpicture} }
  \subfigure[]{
  \begin{tikzpicture}[scale=.8, v/.style={node distance=2cm, draw, rectangle},
    d/.style={node distance=2cm}, hdts/.style={node distance=1cm}, every node/.style={transform shape}]
    \node(f)[v]{$=$}; 
    \node(d1)[d, node distance=1.5cm, above of=f]{$x_1$};
    \node(d2)[d, below left of=f]{$x_2$}; \node(dn)[d, below right of=f]{$x_n$};
    \path[every node/.style={transform shape}]
    (f)edge(d1) edge(d2) edge(dn)
    node[below of=f]{$\cdots$};
  \end{tikzpicture} }
  \subfigure[]{
  \begin{tikzpicture}[scale=.8, v/.style={node distance=2cm, draw, rectangle},
    d/.style={node distance=2cm}, hdts/.style={node distance=1cm}, every node/.append style={transform shape}]
    \node(f)[v]{$\sum$}; 
    \node(d1)[d, node distance=1.5cm, above of=f]{$x_1$};
    \node(d2)[d, below left of=f]{$x_2$}; \node(dn)[d, below right of=f]{$x_n$};
    \path[every node/.style={transform shape}]
    (f)edge[->-out=.5](d1) edge(d2) edge(dn)
    node[below of=f]{$\cdots$};
  \end{tikzpicture} }
  \subfigure[]{
  \begin{tikzpicture}[scale=.8, v/.style={node distance=2cm, draw, rectangle},
    d/.style={node distance=2cm}, hdts/.style={node distance=1cm}, every node/.append style={transform shape}]
    \node(f)[v]{$\max$}; 
    \node(d1)[d, node distance=1.5cm, above of=f]{$x_1$};
    \node(d2)[d, below left of=f]{$x_2$}; \node(dn)[d, below right of=f]{$x_n$};
    \path[every node/.style={transform shape}]
    (f)edge[->-out=.5](d1) edge(d2) edge(dn)
    node[below of=f]{$\cdots$};
  \end{tikzpicture} }
\caption{A graphical illustration of: (a) split function (b) conditional function (c) $\delta_{=}$, (d) $\delta_{\Sigma}$, and (e) $\delta_{\max}$.}
\label{fig:Ind}
\end{figure}

One may observe a sense of ``duality'' between a conditional function and a split function through the following lemma.

\begin{Lemma}
Let $f(x_1,x_2,x_3)$ be a positive real split function, then up to a scaling factor,
$f$ may be written as $p_{X_1}(x_1)p_{X_2|X_1}(x_2|x_1)p_{X_3|X_1}(x_3|x_1)$ for some
probability distributions $p_{X_1}$, $p_{X_2|X_1}$, and $p_{X_3|X_1}$.
\label{lemma:split_prob}
\end{Lemma}
\begin{proof}
Since $f$ is a positive real function, it may be viewed (up to a scaling factor) as a probability
distribution of some RVs $X_1, X_2$ and $X_3$. Hence, $f$ can be written as
$f(x_1,x_2,x_3) = p_{X_1}(x_1)p_{X_2|X_1}(x_2|x_1)p_{X_3|X_1X_2}(x_3|x_1,x_2)$.
Since $f$ is a split function via $x_1$, as we will see
later (cf. Lemma~\ref{lemma:indep})
, we have $X_2 \indep X_3 | X_1$, and the claim follows.
\end{proof} 

Now compare a split function $f(x_1, x_2, x_3)$ with a conditional function $g(x_1, x_2, x_3)$ where 
let us assume that the respective scaling constants making the functions into distributions are 
both $1$.  If we are to draw the Bayesian networks (BN) \cite{Pearl:1988} corresponding to the two 
distributions $f$ and $g$ respectively, we shall see that the directions of the edges in the BN of 
$f$ are completely opposite to those in the BN of $g$. Describing it in terms of causality, one may 
say: The distribution $f$ prescribes that conditioned on the RV $X_1$, we generate the RVs $X_2$ and $X_3$ \emph{independently},
whereas the distribution $g$ prescribes that  $X_2$ and $X_3$ generates $X_1$ {\em jointly}. This 
sense of ``duality'' or ``reciprocity"  (evidently existing in the two kinds of functions involving 
arbitrary number of variables) also justifies our notations of opposite edge directions in denoting the two kinds of functions.

\begin{figure}[ht]
\centering
%	\subfigure[]{\includegraphics[scale = .75]{fig/Split_BN_fig.eps}}	\hspace{.5cm}
%	\subfigure[]{\includegraphics[scale = .75]{fig/Cond_BN_fig.eps}}
\subfigure[]{
\begin{tikzpicture}[v/.style={node distance=2cm, draw, circle, inner sep=.5mm, minimum size=5mm}]
  \node(v1)[v]{$x_1$}; \node(v2)[v, below left of=v1]{$x_2$}; \node(v3)[v, below right of=v1]{$x_3$};
  \path[every node/.style={}]
  (v1)edge[>=triangle 45, ->](v2)  (v1)edge[>=triangle 45, ->](v3);
%  (v1)edge[->-in=.2](v2)
  ;
\end{tikzpicture}	}
\hspace{.5cm}
\subfigure[]{
\begin{tikzpicture}[v/.style={node distance=2cm, draw, circle, inner sep=.5mm, minimum size=5mm}]
  \node(v1)[v]{$x_1$}; \node(v2)[v, below left of=v1]{$x_2$}; \node(v3)[v, below right of=v1]{$x_3$};
  \path[every node/.style={}]
  (v2)edge[>=triangle 45, ->](v1)  (v3)edge[>=triangle 45, ->](v1);
\end{tikzpicture}	}
\caption{The BNs corresponding to: (a) a split function, and (b) a conditional function.}
\end{figure}
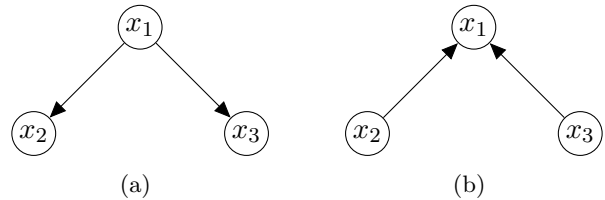

%In general, for a non-symmetric function which is neither a split nor a conditional function,
%we distinguish its first argument graphically using a dot, and assume the same counter-clockwise convention as before,
%cf. Fig.~\ref{fig:lemma-ind}~(c) and Fig.~\ref{fig:sum-transform}. Such practice is merely to facilitate
%some of the graphical representations, and we will
%not insist on it in every occasion.

%we find it helpful at some places, to indicate the ordering of its arguments graphically by marking the edge that corresponds to its first argument
%with a ``dot" and assume the rest of its arguments are encountered in a counter clock-wise manner \cite{Ali:NFG-LinAlg}, 
%cf. Fig.~\ref{fig:lemma-ind}~(c) and Fig.~\ref{fig:sum-transform}.

\noindent\textbf{\underline{Indicator functions and the Iverson's convention}}
An \emph{indicator function} is a $\{0,1\}$-valued function, and at many places we will use
the Iverson's convention $[P]$ to denote the indicator function on the true/false
proposition $P$ defined as $[P] := 1$ if and only if $P$ is true.
For any $x, y \in \X$,
we will often use the proposition $x \stackrel{?}{=} y$, which is defined to be ``true'' if $x=y$ and to be ``false'' otherwise.
Subsequently, we will use the symbol ``='' instead of ``$\stackrel{?}{=}$'', and assume the distinction from the usual use of
``='' for assignment is clear from the context--- Namely, any occurrence of ``='' inside, and only inside, the Iverson brackets refers to ``$\stackrel{?}{=}$''.
We will mainly be interested in the following indicator functions:

%\vspace{.15cm}
\noindent\textbf{Evaluation indicator}
Let $\X$ be a finite alphabet, the \emph{evaluation indicator} (evaluating at some $\overline{x} \in \X$) is an indicator function 
on $\X$ defined as $\delta_{\overline{x}}(x) := [x = \overline{x}]$ for all $x \in \X$, 

%\vspace{.15cm}
\noindent\textbf{Equality indicator} Let $\X$ be a finite alphabet, the \emph{equality indicator} on 
$n$ variables is defined as $\delta_{=}(x_1, \ldots, x_n) := \prod\limits_{i=2}^{n}[x_1 = x_i]$ for all $x_1, \ldots, x_n \in \X$. 
Note that an equality indicator is a split function via any of its arguments, hence, at many places,
% as it is easy to verify that $\delta_{=}(x_1, \ldots, x_n) = \prod \limits_{j = 1, j \neq i}^{n} \delta_{=}(x_i,x_j)$ for any $i$. Since this function is symmetric, 
we may illustrate it graphically without a directed edge.

%\vspace{.15cm}
\noindent\textbf{Constant-one indicator} Let $\X$ be a finite alphabet, the \emph{constant-one indicator} is a degenerate indicator function
defined as ${\textbf{1}}(x):=1$ for all $x \in \X$.

%\vspace{.15cm}
%\noindent\textbf{Sum indicator} Let $\X_2,\ldots, \X_n$ be some subgroups (not necessarily distinct) of a finite abelian group $\X_1$ (additively written), the \emph{sum indicator} is an indicator function $\delta_{\sum}$ on $\X_1 \times \ldots \times \X_n$ defined as
%$\delta_{\Sigma}(x_1, \ldots,x_n) = [x_1 = x_2 + \ldots + x_n]$.
%%such that $\delta_{\Sigma}(x_1, \ldots,x_n) = 1$ if and only if
%%$x_1 = x_2+\ldots+x_n$.
\noindent\textbf{Sum indicator} Let $\X$ be a finite abelian group (additively written), the 
\emph{sum indicator} on $n$ variables is defined as $\delta_{\Sigma}(x_1, \ldots,x_n) := [x_1 = x_2 + \cdots + x_n]$ for all $x_1, \ldots, x_n \in \X$.
%such that $\delta_{\Sigma}(x_1, \ldots,x_n) = 1$ if and only if
%$x_1 = x_2+\ldots+x_n$.
%
A closely related indicator function, which is more popular in the factor graphs and normal graphs literature, is the \emph{parity indicator} function, denoted
$\delta_{+}$, and defined as $\delta_{+}(x_1, \ldots, x_n) := [x_1 + \cdots + x_n = 0]$. It is clear
that $\delta_{\Sigma}(x_1, \ldots, x_n) = \delta_{+}(x_1,-x_2,\ldots,-x_n) = \delta_{+}(-x_1,x_2,\ldots,x_n)$. 
%It is clear that a bivariate parity indicator function is simply a ``sign inverter," i.e., $\delta_{+}(x_1,x_2) = 1$ if and only if $x_1 = -x_2$.

An elementary result concerning sum indicator function is the following lemma.
\begin{Lemma} For any functions $f$ on $\X_1\times \X_2$ and $g$ on
$\X_2\times \X_3$, where $\X_1$, $\X_2$ and $\X_3$ are abelian groups,
\[
\sum\limits_{t, u}f(x, t)g(u, z) \delta_{\sum}(y, t, u)=f(x, y)*g(y,z),
\]
where $\delta_{\sum}$ above is defined on $\X_2^{3}$.
\label{lem:conv_sum}
\end{Lemma}
\begin{proof}
Follows directly from the definition of the convolution.
\end{proof}

%\vspace{.15cm}
%\noindent\textbf{Max indicator} Let $I$ be a finite set, and let $\X = \prod_{i \in I}\X_i$, where $X_i$ is a finite totally ordered alphabet for all $i \in I$. The \emph{$\max$ indicator} on $n$ variables is an indicator function $\delta_{\max}$ on $\X^n$ defined as 
%$\delta_{\max}(x_1, \ldots,x_n) = \prod_{i \in I}\big[(x_{1})_i = \max\big((x_{2})_i, \ldots, (x_{n})_i\big)\big]$. %if and only if $x_1 = \max(x_2,\ldots,x_n)$.

\noindent\textbf{Max indicator} Let $\X$ be an ordered finite set. The \emph{max indicator} on $n$ variable is
defined as $\delta_{\max}(x_1, \ldots, x_n):=[x_1 = \max(x_2, \ldots, x_n)]$ for all $x_1, \ldots, x_n \in \X$. 
Let $I$ be a finite set and let $\X_i$ be an ordered finite set for all $i \in I$. The definition of the max indicator
is extended to the partially-ordered set $\X_I$ by defining 
$\delta_{\max}(x_{I}, \ldots, x'_{I}):=\prod\limits_{i \in I}\delta_{\max}(x_i, \ldots, x'_i)$ for all $x_{I}, \ldots, x'_{I} \in \X_{I}$.

A graphical illustration of the above local functions is shown in Fig.~\ref{fig:Ind}.
Note that the max and sum indicator functions are both conditional functions as illustrated by the
directed edges in Figs.~\ref{fig:Ind}~(d) and (e). 
It is worth noting that the bivariate max indicator and bivariate sum indicator are both equivalent to the bivariate equality indicator, which is a split and a conditional function.

\subsection{Vertex merging/splitting and holographic transformations}

In the framework of NFGs,  a pair of graphical procedures, known as the \emph{vertex merging} and 
\emph{vertex splitting} procedures are particularly useful. In a \emph{vertex merging} procedure,
two vertices representing functions $f$ and $g$ are replaced with a single vertex representing the 
function $\langle f,g \rangle$; conversely, in a \emph{vertex splitting} procedure, a vertex 
representing a function that can be expressed in the sum-of-products form $\langle f,g\rangle$ is 
replaced with an NFG realizing the function $\langle f,g \rangle$.
The two procedures are illustrated in Fig.~\ref{fig:merging-splitting}, and are closely related to the concept of opening/closing the box
\cite{Loeliger:Intro, Loeliger:quantum}. Note that when we put a dashed box around $f$ and $g$ we mean 
that they are replaced with the single function node $\langle f, g \rangle$, in other words, the last
two pictures in Fig.~\ref{fig:merging-splitting} refer to the same NFG. 
%This is similar to the notion of \emph{opening}/\emph{closing} the box, which was developed for FGs (and equivalently normal graphs) under
%the global function semantics~\cite{Loeliger:Intro} \cite{Pascal:Kalman}.
It is easy to see that the exterior function of any NFG is invariant
under the vertex merging/splitting procedure \cite{Bashabsheh:HolTrans}.
%The following lemma allows us to freely apply the vertex mergin/splitting procedure whenever needed \cite{Bashabsheh:HolTrans}.
%\begin{Lemma}
%The exterior function realized by any NFG is invariant under vertex grouping and vertex splitting.
%\label{Lemma:merging-splitting}
%\end{Lemma}
\begin{figure}[ht]
\centering
\begin{tikzpicture}[v/.style={draw, rectangle, minimum size=6mm}]
  \def\dist{1.5cm}; \def\shift{.26cm};
  \node(f) at (0, \dist)[v]{$f$}; \node(fdL) at (-\dist, \dist){}; \node(fdR) at (\dist, \dist){};
  \node(g) at (0, 0)[v]{$g$}; \node(gdL) at (-\dist,0){}; \node(gdR) at (\dist,0){};
  \draw (f)--(fdL); \draw (g)--(gdL);
  \draw[transform canvas={xshift=-\shift}] (f) -- (g); \draw[transform canvas={xshift=\shift}] (f) -- (g); 
  \draw[transform canvas={yshift=-\shift}] (f) -- (fdR); \draw[transform canvas={yshift=\shift}] (f) -- (fdR); 
  \draw[transform canvas={yshift=-\shift}] (g) -- (gdR); \draw[transform canvas={yshift=\shift}] (g) -- (gdR); 
  \node(hdts)[node distance=\dist/2, below of=f]{$\cdots$};
  \node(vdts1) at (\dist/2, \dist)[rotate=90]{$\cdots$}; \node(vdts2) at (\dist/2, 0)[rotate=90]{$\cdots$};
\end{tikzpicture}
\begin{tikzpicture}[]
  \def\dist{1.5cm};
  \node(merg) at (0,\dist){$\stackrel{\rm merging}{\longrightarrow}$};
  \node(merg) at (0,0){$\stackrel{\rm splitting}{\longleftarrow}$};
\end{tikzpicture}
\begin{tikzpicture}[v/.style={draw, rectangle, minimum size=6mm}]
  \def\dist{1.5cm}; \def\shift{.26cm};
  \node(f) at (0, \dist)[v]{$f$}; \node(fdL) at (-\dist, \dist){}; \node(fdR) at (\dist, \dist){};
  \node(g) at (0, 0)[v]{$g$}; \node(gdL) at (-\dist,0){}; \node(gdR) at (\dist,0){};
  \draw (f)--(fdL); \draw (g)--(gdL);
  \draw[transform canvas={xshift=-\shift}] (f) -- (g); \draw[transform canvas={xshift=\shift}] (f) -- (g); 
  \draw[transform canvas={yshift=-\shift}] (f) -- (fdR); \draw[transform canvas={yshift=\shift}] (f) -- (fdR); 
  \draw[transform canvas={yshift=-\shift}] (g) -- (gdR); \draw[transform canvas={yshift=\shift}] (g) -- (gdR); 
  \node(hdts)[node distance=\dist/2, below of=f]{$\cdots$};
  \node(vdts1) at (\dist/2, \dist)[rotate=90]{$\cdots$}; \node(vdts2) at (\dist/2, 0)[rotate=90]{$\cdots$};
  \draw[dashed] (hdts) \clbox{-.5}{-1.1}{.5}{1.1};
\end{tikzpicture}
\begin{tikzpicture}[]
  \def\dist{1.5cm};
  \node(iff) at (\dist/2,\dist*.55){$\Leftrightarrow$};
  \node(dump) at (0,0){};
\end{tikzpicture}
\begin{tikzpicture}[every node/.style={minimum size=6mm}]
  \def\dist{1.5cm}; \def\shift{.26cm};
  \node(f) at (0, \dist){}; \node(fdL) at (-\dist, \dist){}; \node(fdR) at (\dist, \dist){};
  \node(g) at (0, 0){}; \node(gdL) at (-\dist,0){}; \node(gdR) at (\dist,0){};
  \draw[transform canvas={xshift=-.2cm}] (f)--(fdL); \draw[transform canvas={xshift=-.2cm}] (g)--(gdL);
  \draw[transform canvas={yshift=-\shift, xshift=.2cm}] (f) -- (fdR); \draw[transform canvas={yshift=\shift, xshift=.2cm}] (f) -- (fdR); 
  \draw[transform canvas={yshift=-\shift, xshift=.2cm}] (g) -- (gdR); \draw[transform canvas={yshift=\shift, xshift=.2cm}] (g) -- (gdR); 
   \node(hdts) at (0,\dist/2){$\langle f, g \rangle$};
  \node(vdts1) at (\dist/2, \dist)[rotate=90]{$\cdots$}; \node(vdts2) at (\dist/2, 0)[rotate=90]{$\cdots$};
  \draw (hdts) \clbox{-.5}{-1.1}{.5}{1.1};
\end{tikzpicture}
    \caption{Vertex merging and vertex splitting.}
    \label{fig:merging-splitting}
\end{figure}
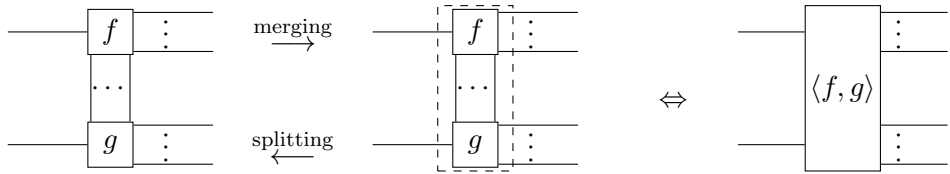

The simple pair of procedures enables the notion of ``holographic transformation'' for NFGs, which transforms each local function of an NFG while keeping the graph topology. Further, the ``generalized Holant theorem" relates the exterior function realized by the holographically transformed NFG with that realized by the original NFG. We summarize these notions below, and refer the interested reader to \cite{Bashabsheh:HolTrans} for more details.

Suppose that $f$ and $g$ are two bivariate functions on $\X\times \X$ such that 
$\big\langle f(x, s), g(s, x')\big\rangle = \delta_{=}(x, x')$ for all $(x, x') \in \X\times \X$, then inserting to an NFG edge (regular or half),
representing an $\X$-valued variable,  
the pair of functions $f$ and $g$ %(with variable $s$ representing the edge connecting $f$ and $g$) 
is equivalent to inserting the function $\delta_{=}$, which can be verified not to 
change the exterior function. The functions $f$ and $g$ in this case are called an inverse-pair
of transformers, and such a graphical procedure is called {\em inverse-pair transformer insertion}. 

Given an NFG ${\cal G}$ with a set of half edges $L$, the following procedure defines a transformed NFG with the same topology as $\G$:
\begin{enumerate} 
\item[(H1)] In each half edge $x_i$, insert a bivariate function $g_i(x_i, y_i)$--- We may refer to such transformers as
  \emph{external transformers}.
\item[(H2)] In each internal edge, insert an inverse-pair of transformers--- We may refer to such transformers as \emph{internal
  transformers}.
\item[(H3)] For each original vertex $v$ in ${\cal G}$, apply vertex merging procedure to merge $f_v$ and its surrounding vertices.
\end{enumerate}
Such a transformation of NFG is known as a \emph{holographic transformation}. 
If we denote the resulting NFG by ${\cal G}'$, then it is clear that 
$Z_{\G'}(y_L)=
\big\langle Z_{\G}(x_L), \prod_{i\in L} g_i(x_i, y_i)\big\rangle,$
since only Step (H1) above affects the exterior function.  This is essentially the {\em generalized Holant theorem} of \cite{Bashabsheh:HolTrans},
which in subsequent discussions will be referred to as the GHT. %, and such a transformation of NFG is known as a holographic transformation. 

%% file: NFGs_Prob2.tex
\section{NFG Models}
\label{section:NFGs_prob}

We now present a generic NFG probabilistic model.  Formally, an {\em NFG probabilistic model} or simply an {\em NFG model} is an NFG whose exterior function is up to scale the joint distribution of 
some RVs (each represented by a half edge) and which satisfies the following two properties:
1) the NFG is bipartite and simple; 2) half edges are only incident on one independent vertex set and there is exactly one half
edge incident on each vertex in this set; we call these vertices {\em interface vertices}, and call the ones in the other vertex
set {\em latent vertices}.  We will call the corresponding functions indexed by these two vertex sets {\em interface functions}
and {\em latent functions} respectively, although we will be quite loose in speaking of a vertex and a function exchangeably as
we do for a variable and an edge/half edge. We will customarily denote the set of interface vertices by $I$ and the set of latent
vertices by $J$. \footnote{We note that demanding no half edge incident on the latent functions entails no loss of generality,
since if there is such a half edge, one may always insert a bivariate equality indicator function (or equivalently a bivariate
max indicator or sum indicator) into the half edge, which converts the NFG to an equivalent one with this half edge turned into
a regular edge. Since  the bivariate equality indicator is both a split function and a conditional function, inserting such a
function has no impact on our later restriction on the interface functions, where we require them to be all split functions or all conditional functions. }

Note that since each interface function has exactly one half edge incident on it, unless it is more convenient to make the distinction, we will subsequently identify the set of half edges using $I$, i.e., an interface vertex will index both its function and the half edge incident on it.
An example NFG model is shown in Figure \ref{fig:NFG_3_pic} (a), where the top-layer vertices are
the interface vertices, and the bottom-layer vertices are the latent vertices.  If necessary, we may formally denote such NFG using such a notation as $\G(I \cup J, E, f_{I \cup J})$. 

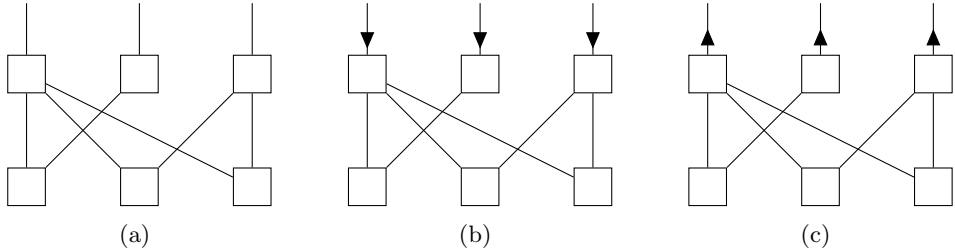
\begin{figure}
\centering%{
  \subfigure[]{
    \begin{tikzpicture}[v/.style={node distance=1.5cm, draw, rectangle}, d/.style={node distance=1.2cm}, every node/.append style={transform shape, minimum size=5mm}]
      \node(u1)[v]{}; \node(u2)[v, right of=u1]{}; \node(u3)[v, right of=u2]{}; 
      \node(v1)[v, below of=u1]{}; \node(v2)[v, below of=u2]{}; \node(v3)[v, below of=u3]{}; 
      \node(d1)[d, above of=u1]{}; \node(d2)[d, above of=u2]{}; \node(d3)[d, above of=u3]{}; 
      \path[every node/.style={transform shape}]
      (u1)edge(v1)edge(v2)edge(v3)	(u2)edge(v1)	(u3)edge(v2)edge(v3)
      (u1)edge(d1)	(u2)edge(d2)	(u3)edge(d3);
    \end{tikzpicture} }
    \hspace{.5cm}
  \subfigure[]{
    \begin{tikzpicture}[v/.style={node distance=1.5cm, draw, rectangle}, d/.style={node distance=1.2cm}, every node/.append style={transform shape, minimum size=5mm}]
      \node(u1)[v]{}; \node(u2)[v, right of=u1]{}; \node(u3)[v, right of=u2]{}; 
      \node(v1)[v, below of=u1]{}; \node(v2)[v, below of=u2]{}; \node(v3)[v, below of=u3]{}; 
      \node(d1)[d, above of=u1]{}; \node(d2)[d, above of=u2]{}; \node(d3)[d, above of=u3]{}; 
      \path[every node/.style={transform shape}]
      (u1)edge(v1)edge(v2)edge(v3)	(u2)edge(v1)	(u3)edge(v2)edge(v3)
      (u1)edge[->-in=.2](d1)	(u2)edge[->-in=.2](d2)	(u3)edge[->-in=.2](d3);
    \end{tikzpicture} }
    \hspace{.5cm}
  \subfigure[]{
    \begin{tikzpicture}[v/.style={node distance=1.5cm, draw, rectangle}, d/.style={node distance=1.2cm}, every node/.append style={transform shape, minimum size=5mm}]
      \node(u1)[v]{}; \node(u2)[v, right of=u1]{}; \node(u3)[v, right of=u2]{}; 
      \node(v1)[v, below of=u1]{}; \node(v2)[v, below of=u2]{}; \node(v3)[v, below of=u3]{}; 
      \node(d1)[d, above of=u1]{}; \node(d2)[d, above of=u2]{}; \node(d3)[d, above of=u3]{}; 
      \path[every node/.style={transform shape}]
      (u1)edge(v1)edge(v2)edge(v3)	(u2)edge(v1)	(u3)edge(v2)edge(v3)
      (u1)edge[->-out=.5](d1)	(u2)edge[->-out=.5](d2)	(u3)edge[->-out=.5](d3);
    \end{tikzpicture} }
\caption{(a) an NFG model (b) a constrained NFG model (c) a generative NFG model}
\label{fig:NFG_3_pic}
\end{figure}

In this modelling framework, we will focus on two ``dual'' families of models, which we call the {\em constrained NFG models} and
the {\em generative NFG models} respectively (see, e.g., Figure \ref{fig:NFG_3_pic} (b) and (c) respectively for a quick preview).
We will demonstrate how these models are related to the previous models such as FG and CFG.  We will also introduce ``transformed
NFG models'' in Section~\ref{section:transformed_NFGs}, a special case of which reduces to CDN.

\input{constrainedNFG_basic}

\input{constrainedNFG_FG}

\input{generativeNFG}

\input{generativeNFG_CFG}

\input{independence}

%% file: constrainedNFG_basic.tex
\subsection{Constrained NFG model}

A \emph{constrained} NFG  model is an NFG model in which all interface functions are split functions via their respective external variables.

%bipartite 
%NFG $\G(I \cup J, E,E^{\rm ext}, f_{I \cup J})$ %whose exterior function is the probability distribution of the RVs $X_{E^{\rm ext}}$, and
%such that  for all $i \in I$,
% $f_{i}$ is a split function via its external (i.e., half-edge) variable.

%Before we discuss this class of NFGs in more details, we pause to present an example in hope that it will bring more intuition into the understanding of constrained NFGs. To this end, we start with the following remark.

To bring more intuition into this definition, we first take a slight digression and show in the following lemma that it is possible to ``shape'' a distribution by "random rejection".

\begin{Lemma}
Let $X$ be a RV with a probability distribution $p_{X}$, where $X$ assumes its values from a finite set $\X$, and let
$h$ be a normalized non-negative real function on $\X$ with a non-empty support, where the normalization is in the sense that $\max \limits_{x \in \X} h(x) = 1$. Draw $x$ from $p_X$ and accept it with probability $h(x)$ and reject it with probability $1-h(x)$. If $x$ is accepted, output $x$; otherwise repeat the process until some other $x'$ is drawn and accepted. Denote the output random variable by $Y$. Then the probability distribution $p_Y$ of $Y$ is, up to scale, $h(y)p_{X}(y)$, for all $y\in \X$.
\label{remark:switch}
\end{Lemma}

\begin{proof}
  Let $Z$ be a $\left\{ 0,1 \right\}$ RV representing the random rejection in the lemma, i.e., a sample $x \in \X$ is rejected if
  $Z=0$, accepted if $Z=1$, and the probability that $Z=1$ is $h(x)$. Then the statement $Y = y$ is equivalent
to $(X,Z) = (y,1)$, and we have $p_{Y}(y) = p(X = y)p(Z = 1 | X = y) = p_{X}(y) h(y).$
\end{proof}

%\begin{proof}
%Assume $N$ samples of $X$ are drawn and let $n(x)$ be the number of occurrences of the element $x$ for all $x \in \X$.
%Then, $N = \sum_{x \in \X} n(x)$ and $p_{X}(x) = \lim_{N \to \infty} \frac{n(x)}{N}$.
%Now the number of samples of value $x$ that are accepted is equal to $h(x)n(x)$ since the fraction of times
%the coin lands head is $h(x)$. Hence, the total number of accepted samples is%
%\footnote{The subscript in $M_N$ is to emphasize its dependence on $N$. Note that as $N$ approaches infinity, so does $M_N$, since $h$ has a non-trivial support.}
%$M_N = \sum_{x \in \X} h(x)n(x)$ and
%\[
%p_{Y}(y) = \lim_{N \to \infty} \frac{h(y)n(y)}{M_N} = \lim_{N \to \infty} \frac{h(y)n(y) \frac{1}{N}}{\sum_{x} h(x)n(x) \frac{1}{N}}
%= \frac{h(y)p_{X}(y)}{\sum_{x}h(x)p_{X}(x)},
%\]
%and the claim follows since $\sum_{x} h(x)n(x)$ is a constant.
%\end{proof}

The idea of ``distribution shaping'' via "random rejection" is central to the semantics of constrained NFG models, which we  demonstrate in the example next.

\begin{Example}
Let $\G$ be a constrained NFG as in Fig.~\ref{fig:const_ex} where
$f_1(x_1,s_1,s'_1)$ and $f_2(x_2,s_2,s'_2)$ are positive functions that split via $x_1$ and $x_2$, respectively,
and $h_1, h_2$ and $h_3$ are non-negative functions (with non-empty supports). 
From Lemma~\ref{lemma:split_prob} we may express, up to a respective scaling factor,  $f_1$ as 
$p_{X_1}(x_1)p_{S_1|X_1}(s_1|x_1) p_{S'_1|X_1}(s'_1|x_1)$ and $f_2$ as
$p_{X_2}(x_2)p_{S_2|X_2}(s_2|x_2)p_{S'_2|X_2}(s'_2|x_2)$, 
%$f_1(x_1,s_1,s'_1) = p_{X_1}(x_1)p_{S_1|X_1}(s_1|x_1) p_{S'_1|X_1}(s'_1|x_1),$ and
%$f_2(x_2,s_2,s'_2) = p_{X_2}(x_2)p_{S_2|X_2}(s_2|x_2)p_{S'_2|X_2}(s'_2|x_2),$
for some distributions $p_{X_1}$, $p_{X_2}$, $p_{S_1|X_1}$, $p_{S'_1|X_1}$, $p_{S_2|X_2}$  and $p_{S'_2|X_2}$.
The RVs represented by the NFG may be regarded as being generated by the following process. 
\begin{enumerate}
\item  Draw $(x_1,x_2)$ from distribution $p_{X_1}(x_1)p_{X_2}(x_2)$
where $p_{X_1}$ and $p_{X_2}$ are as specified by our choices above. Note that the two components of the drawn vector are independent.

\item Draw vector $(s_1, s'_1)$ from the distribution $p_{S_1|X_1}(s_1|x_1)p_{S'_1|X_1}(s'_1|x_1)$ and 
draw $(s_2, s'_2)$ from $p_{S_2|X_2}(s_2|x_2)p_{S'_2|X_2}(s'_2|x_2)$. It is clear that the joint distribution of $(x_1, x_2,  s_1, s'_1, s_2, s'_2)$ is up to scale $f_1(x_1, s_1, s'_1)f_2(x_2, s_2, s'_2)$. 

\item Let $H(s_1, s'_1, s_2, s'_2):= c\cdot h_1( s_1)h_2(s'_1, s_2)h_3(s'_2),$ where $c$ is a normalizing constant such that the maximum value of $H(\cdot)$ is $1$. Accept the drawn vector $(x_1, x_2, s_1, s'_1, s_2, s'_2)$ with probability 
$H(s_1, s'_1, s_2, s'_2)$ and reject it with probability $1-H(s_1, s'_1, s_2, s'_2)$.

\item If the drawn $(x_1, x_2, s_1, s'_1, s_2, s'_2)$ is rejected, repeat the procedure from step 1, until the drawn $(x_1, x_2,  s_1, s'_1, s_2, s'_2)$ is accepted.  By Lemma~\ref{remark:switch}, the accepted vector has a distribution equal, up to scale, to
$
f_1(x_1, s_1, s'_1)f_2(x_2, s_2, s'_2)H(s_1, s'_1, s_2, s'_2).
$
\item Output $(x_1, x_2)$.  Then clearly the output vector has distribution that is up to scale the exterior function of the NFG.
\end{enumerate}
\label{ex:const}
\end{Example}

The procedure introduced in the example above generalizes in an obvious way to arbitrary constrained NFG models. 
Instead of precisely, but repetitively, stating the procedure for the general setting, 
we make the following remarks. The interface functions completely specify how the external variables are 
drawn and how the internal variables are drawn conditioned on the drawn external configuration. 
The drawn internal configuration then undergoes a ``random rejection" according to the product of all 
latent functions. The external configuration giving rise to an accepted internal configuration then
necessarily follows the distribution prescribed by the exterior function of the NFG.

%\begin{wrapfigure}{r}{.2\textwidth}%[ht]
\begin{figure}[ht]
\centering
%\vspace{-13pt}
%\includegraphics[scale = .6]{fig/Const_ex_fig.eps}
%\vspace{-7pt}
  \begin{tikzpicture}[v/.style={node distance=1.5cm, draw, rectangle}, d/.style={node distance=1.2cm}, every node/.append style={transform shape, minimum size=5mm}]
    \node(u1)[v]{$f_1$}; \node(d1)[d, above of=u1]{};  
    \node(v1)[v, below left of=u1]{$h_1$}; \node(v2)[v, below right of=u1]{$h_2$};  \node(u2)[v, above right of=v2]{$f_2$};  \node(v3)[v, below right of=u2]{$h_3$};
    \node(d2)[d, above of=u2]{};
    \path[every node/.style={transform shape}]
    (u1)edge node[left]{$s_1$}(v1)edge node[right]{$s'_1$}(v2)	(u2)edge node[left]{$s_2$}(v2)edge node[right]{$s'_2$}(v3)
    (u1)edge[->-in=.2] node[above left]{$x_1$}(d1)	(u2)edge[->-in=.2] node[above left]{$x_2$}(d2);
  \end{tikzpicture} 
\caption{Example~\ref{ex:const}.}
\label{fig:const_ex}
\end{figure}
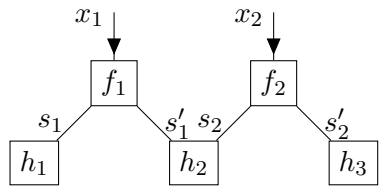
%\end{wrapfigure}

Analogously, one may view a constrained NFG model as a ``probabilistic checking system'': independent ``inputs'' (external
variables)  excite the ``internal states'' (internal variables) of the system via interface functions; the state configuration is
``checked'' probabilistically by the latent functions;  only the external configurations that pass the internal check are kept. In
general, the internal checking mechanism induces dependence among the external variables, which were {\em a priori} independent.
In the special case when the latent functions are all indicator functions, the checking system is in fact deterministic, reducing
to a set of constraints on the internal states, cf. Section~\ref{sec:coding}. This has been the motivation behind the name ``constrained NFG model''.  As we
will show momentarily that constrained NFG models and FG models are equivalent, the ``probabilistic checking system'' perspective
of constrained models provides a different  and new interpretation of the FG models.

%% file: constrainedNFG_FG.tex
\subsection{Constrained NFG models  are equivalent to  FGs}
%We begin with the following proposition.
%\begin{Prop}
%Let $\G(I \cup J, E, E^{\rm ext}, f_{I \cup J})$ be a constrained NFG model
%where for every interface function $f_i$,  $i \in I$, is an equality indicator, then
%% such that $f_i(x_{E(i)}) = \delta_{=}(x_{E(i)})$ for all $i \in I$. Then,
%$
%Z_{\G}(x_{E^{\rm ext}}) = \prod \limits_{j \in J} f_{j}(x_{E^{\rm ext}(j)}, x_{{\rm ne}(j)}).
%$
%\label{prop:mul-NFG}
%\end{Prop}
%%\begin{proof}
%%Follows directly from the definitions of the exterior function.
%%\end{proof}

Suppose that a constrained NFG model is such that every interface function is an equality indicator function. 
It is known \cite{Forney2001:Normal} that one may convert such NFG to an FG according to the following
procedure: {\em For each interface vertex, replace it by a variable vertex representing its
half-edge variable and  remove the half edge.}

\begin{Prop}
If  in a  constrained NFG model all interface functions are equality indicators, then the above procedure gives rise to an FG equivalent to the NFG. 
\label{prop:mul-NFG}
\end{Prop}
\begin{proof}
Let $\G(I \cup J, E, f_{I \cup J})$ be the NFG in hand where $f_{i} = \delta_{=}$ for all $i \in I$.
The resulting FG has an underlying graph $(I \cup J, E)$ where $I$ and $J$ are the variable and function indexing sets, respectively.
Hence, the global function of the FG is the multiplication $\prod_{j \in J} f_j(x_{{\rm ne}(j)})$. 
On the other hand, if we use $T(v)$ to denote the set of internal edges incident on node $v$ in the NFG, then the exterior function
of the NFG is $\left\langle \prod_{j \in J} f_{j}(s_{T(j)}), \prod_{i \in I} \delta_{=}(x_i,s_{T(i)}) \right\rangle$, which, if $i'$ and $j'$ are connected by an 
edge $t$, accounts to substituting
$x_{i'}$ in place of the argument $s_{t}$ of $f_{j'}$ in the product $\prod_{j \in J} f_{j}(s_{T(j)})$, for all adjacent $i'$ and $j'$. The claim follows by noting
that $T(j) = \{\{i,j\}:i \in {\rm ne}(j)\}$.
\end{proof}

The proposition essentially suggests that the joint distribution represented by such a constrained NFG model factors
multiplicatively and therefore can be represented by an FG. In fact the converse is also true, namely that any FG can be converted
to an equivalent constrained NFG model with all interface functions being equality indicators. This is illustrated in Figure \ref{fig:EX1_Mul}.

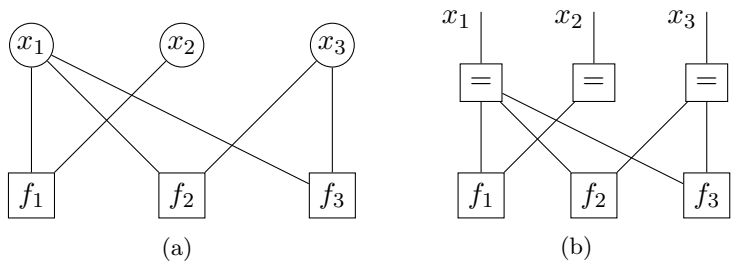
\begin{figure}[ht]
\centering
	\setcounter{subfigure}{0}
%	\vspace{-2pt}
%%	\subfigure[]{\includegraphics[scale = .85]{fig/EX1_FG_fig.eps}} \hspace{.2cm}
%%	\subfigure[]{\includegraphics[scale = .85]{fig/EX1_NFG_Mul_fig.eps}} \hspace{.6cm} %\\
%	\subfigure{\includegraphics[scale = .65]{fig/FG_fig.eps}} \hspace{1.2cm}
%	\subfigure{\includegraphics[scale = .65]{fig/NFG_Mul_fig.eps}}
%	\vspace{-2pt}
    \subfigure[]{
    \begin{tikzpicture}[var/.style={node distance=2cm, draw, circle, inner sep=.5mm, minimum size=4mm}, 
      fun/.style={node distance=2cm, draw, rectangle, minimum size=4mm}]
      \node(x1)[var]{$x_1$}; \node(x2)[var, right of=x1]{$x_2$}; \node(x3)[var, right of=x2]{$x_3$};
      \node(f1)[fun, below of=x1]{$f_1$}; \node(f2)[fun, right of=f1]{$f_2$}; \node(f3)[fun, right of=f2]{$f_3$};
      \path[every node/.style={}]
      (x1)edge(f1)edge(f2)edge(f3)
      (x2)edge(f1)
      (x3)edge(f2)edge(f3);
    \end{tikzpicture} }
    \hspace{.5cm}
    \subfigure[]{
    \begin{tikzpicture}[v/.style={node distance=1.5cm, draw, rectangle}, d/.style={node distance=1.2cm}, every node/.append style={transform shape, minimum size=5mm}]
      \node(u1)[v]{$=$}; \node(u2)[v, right of=u1]{$=$}; \node(u3)[v, right of=u2]{$=$}; 
      \node(v1)[v, below of=u1]{$f_1$}; \node(v2)[v, below of=u2]{$f_2$}; \node(v3)[v, below of=u3]{$f_3$}; 
      \node(d1)[d, above of=u1]{}; \node(d2)[d, above of=u2]{}; \node(d3)[d, above of=u3]{}; 
      \path[every node/.style={transform shape}]
      (u1)edge(v1)edge(v2)edge(v3)	(u2)edge(v1)	(u3)edge(v2)edge(v3)
      (u1)edge node[above left]{$x_1$}(d1)	(u2)edge node[above left]{$x_2$}(d2)	(u3)edge node[above left]{$x_3$}(d3);
    \end{tikzpicture} }
  \caption{An FG and its equivalent constrained NFG.}
	\label{fig:EX1_Mul}
\end{figure}

%Next we show that constrained NFGs are in fact equivalent to FGs.
%This follows from the previous discussion and the following proposition.
Next we show that any constrained NFG is in fact equivalent to one with equality interface function.
%\begin{Prop}
%Any constrained NFG model can be converted to an equivalent constrained NFG model with the same
%underlying graph, and in which all interface functions are equality indicators.
%\label{prop:Const_Eq}
%\end{Prop}
%\begin{proof}
%Each interface function $f_i$ is the product $\prod_{t \in T(i)} f_{t}(x_i,s_t)$ of bivariate functions $f_{t}$,
%and from Proposition~\ref{prop:mul-NFG}, it can be written as the sum-of-products form
%$\left\langle \prod_{t \in T(i)} f_{t}(s'_t,s_t), \delta_{=}(x_i,s'_{T(i)}) \right\rangle$. Upon vertex splitting,
%each interface vertex can be replaced with the NFG representing its corresponding sum-of-product form. The claim follows upon
%vertex merging each hidden node $f_j$ with its adjacent bivariate functions, i.e., by replacing
%$f_j$ with $\left\langle f_j, \prod_{t \in T(j)} f_{t} \right\rangle$.
%\end{proof}
Given an arbitrary constrained NFG $\G$ where each interface function $f_{i}$ splits as $\prod_{t\in T(i)}f_{t}$ for some bivariate functions $f_t$.
The following procedure converts $\G$ into a constrained NFG with the same underlying graph as $\G$, and in which all interface functions
are equality indicators:
\begin{itemize}
  \item[1)] Replace each interface function with an equality indicator.
  \item[2)] Replace each hidden function $f_j$ with $\big\langle f_j, f_{t}:t\in T(j) \big\rangle$.
\end{itemize}

\begin{Prop}
%Any constrained NFG model can be converted to an equivalent constrained NFG model with the same
%underlying graph, and in which all interface functions are equality indicators.
In the above procedure, the original and resulting constrained NFGs are equivalent, and have the same underlying graph.
\label{prop:Const_Eq}
\end{Prop}
\begin{proof}
The fact that the two NFGs have the same underlying graph is clear. To prove equivalence, each interface function $f_i$ is the product $\prod_{t \in T(i)} f_{t}(x_i,s_t)$ of bivariate functions $f_{t}$,
and from Proposition~\ref{prop:mul-NFG}, it can be written as the sum-of-products form
$\big\langle \prod_{t \in T(i)} f_{t}(s'_t,s_t), \delta_{=}(x_i,s'_{T(i)}) \big\rangle$. Upon vertex splitting,
each interface vertex can be replaced with the NFG representing its corresponding sum-of-product form. The claim follows upon
vertex merging each hidden node $f_j$ with its adjacent bivariate functions, i.e., by replacing
$f_j$ with $\big\langle f_j, f_{t}:t \in T(j) \big\rangle$.
\end{proof}

\begin{figure}[ht]
\centering
%\subfigure[]{\includegraphics[scale = .48]{fig/Const_Eq_fig.eps}}		%\hspace{.0cm}
%\subfigure[]{\includegraphics[scale = .48]{fig/Const_Eq2_fig.eps}}	%\hspace{.0cm}
%\subfigure[]{\includegraphics[scale = .48]{fig/Const_Eq3_fig.eps}}
\subfigure[]{
  \begin{tikzpicture}[u/.style={node distance=\dist, draw, rectangle}, v/.style={node distance=\dist/2, draw, rectangle},
      d/.style={node distance=\dist/3}, every node/.append style={transform shape}]
      \def\dist{4cm};
      \node(u1)[u]{$g$}; \node(u2)[u, right of=u1]{$h$}; 
      \node(d1)[d, above of=u1]{}; \node(d2)[d, above of=u2]{}; 
      \node(v1)[v, below of=u1]{$f_1$}; \node(v2)[v, right of=v1]{$f_2$}; \node(v3)[v, right of=v2]{$f_3$}; 
      \path[every node/.style={transform shape}]
      (u1)edge(v1)edge(v2)edge(v3)	(u2)edge(v2)edge(v3)
      (u1)edge[->-in=.2](d1)	(u2)edge[->-in=.2](d2);
    \end{tikzpicture} }
    \hspace{.5cm}
\subfigure[]{
  \begin{tikzpicture}[u/.style={node distance=\dist*2, draw, rectangle}, v/.style={node distance=\dist, draw, rectangle},
      d/.style={node distance=\dist/3*2}, t/.style={node distance=\dist/2, draw, rectangle, inner sep=.5mm}, every node/.append style={transform shape}]
      \def\dist{2cm};
      \node(u1)[u]{$=$}; \node(u2)[u, right of=u1]{$=$}; 
      \node(d1)[d, above of=u1]{}; \node(d2)[d, above of=u2]{}; 
      \node(v1)[v, below of=u1]{$f_1$}; \node(v2)[v, right of=v1]{$f_2$}; \node(v3)[v, right of=v2]{$f_3$}; 
      \node at (0,-\dist*.6)[t](g1){$g_1$}; \node at (\dist*.6,-\dist*.6)(g2)[t]{$g_2$}; \node at (\dist,-\dist*.6)(h1)[t]{$h_1$};
      \node at (\dist*2,-\dist*.6)(h2)[t]{$h_2$}; \node at (\dist*2-\dist/2,-\dist*.6)(g3)[t]{$g_3$}; 
      \path[every node/.style={transform shape}]
      (u1)edge(g1)edge(g2)edge(g3) (g1)edge(v1) (g2)edge(v2) (g3)edge(v3)
      (u2)edge(h1)edge(h2) (h1)edge(v2) (h2)edge(v3)
      (u1)edge(d1)	(u2)edge(d2);
      % Close the box
      \draw[dashed] (v1) \clbox{-.5}{-.5}{.4}{1.2};
      \draw[dashed] (v2) \clbox{-1.4}{-.5}{.4}{1.2};
      \draw[dashed] (v3) \clbox{-1.4}{-.5}{.4}{1.2};
    \end{tikzpicture} }
    \hspace{.5cm}
%\subfigure[]{
%  \begin{tikzpicture}[u/.style={node distance=\dist*2, draw, rectangle}, v/.style={node distance=\dist*2, draw, rectangle},
%      d/.style={node distance=\dist/3*2}, t/.style={node distance=\dist/2, draw, rectangle, inner sep=.5mm}, every node/.append style={transform shape}]
%      \def\dist{1cm}; \def\distu{\dist/3*2};
%      \node at (0,0)(v1)[v]{$f_1$}; \node at (\dist*2, 0)(v2)[v]{$f_2$}; \node at (\dist*4, 0)(v3)[v]{$f_3$}; 
%      \node at (\dist, 2*\dist)(u1)[u]{$g$}; \node at (3*\dist,2*\dist)(u2)[u]{$h$}; 
%      \node(d1)[d, above of=u1]{}; \node(d2)[d, above of=u2]{}; 
%      \node(g1) at (\dist/2, \dist)[t]{$g_1$}; \node at (\dist*3/2,\dist)(g2)[t]{$g_2$}; \node at (\dist*5/2,\dist)(g2)[t]{$g_3$};
%%      \node at (\dist*2,-\dist/2)(g2)[t]{$h_2$};
%%      \node at (\dist*2-\dist/2,-\dist/2)(g2)[t]{$h_1$};
%%      \path[every node/.style={transform shape}]
%%      (u1)edge(v1)edge(v2)edge(v3)	(u2)edge(v2)edge(v3)
%%      (u1)edge[->-in=.2](d1)	(u2)edge[->-in=.2](d2);
%    \end{tikzpicture} }
\caption{Converting a constrained NFG model to one in which all interface functions are equality indicators. 
(a) An example of a constrained NFG where by assumption $g$ splits into $g_1, g_2$ and $g_3$, and $h$ splits into $h_1$ and $h_2$, (b) vertex
splitting of interface nodes, followed by vertex merging of each hidden function with its neighboring bivariate functions. 
%Every interface function in the NFG in (a) can be  expressed  a product of functions including an equality indicator. (This is because, for example, a split function $f(x, s_1, s_2)$ can be written as $a(x)b(x, s_1)c(x, s_2)$ by Lemma \ref{lemma:split_prob}, further as $d(x, s_1)c(x, s_2)$, and finally as $\sum_{y_1, y_2}\delta_{=}(x, y_1, y_2)d(y_1, s_1)c(y_2, s_2)$.) This gives rise to an equivalent NFG in (b), which upon vertex merging, reduces to (c).
}
\label{fig:proof_cNFG2cNFGwtEq} 
%}
\end{figure}
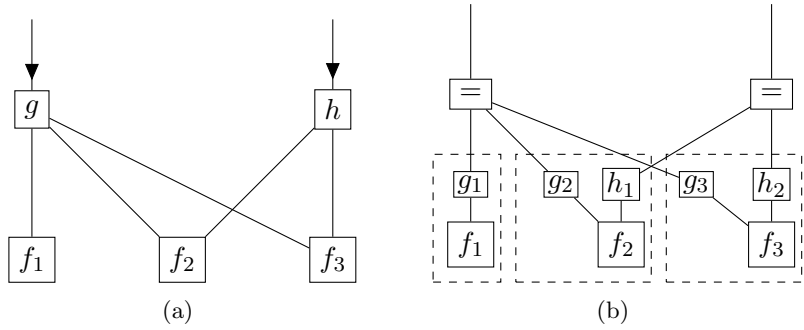

The conversion stated in the proposition and an illustration of the proof are shown in Figure \ref{fig:proof_cNFG2cNFGwtEq}. Invoking Proposition \ref{prop:mul-NFG}, the following theorem is immediate.
\begin{Theorem}
Every constrained NFG can be converted to an equivalent FG.
\end{Theorem}

%% file: generativeNFG.tex
\subsection{Generative NFG model}
A \emph{generative} NFG model is an NFG model in which every interface function is a
 conditional function of its half edge variable given its remaining arguments.
%Generative NFGs may be used to efficiently describe models
%where a collection of independent sets of RVs are used to generate some new RVs, as illustrated by the following example
%In such NFGs, one may view the set of RVs $X_{I}$ as being generated from the RVs 
%$X_{E^{\rm int}}$ according to the rules specified by the conditional functions $f_i$ for all $i$, as illustrated in the following example. 
%
The following example gives sufficient insight of the modelling semantics of a generative NFG. 

%\begin{wrapfigure}{r}{.22\textwidth}%[ht]
\begin{figure}[ht]
\centering
\begin{tikzpicture}[v/.style={node distance=2cm, draw, rectangle, minimum size=5mm}, d/.style={node distance=1cm}]
  \node(u1)[v]{$p_{X_1|S_1}$}; \node(u2)[v, right of=u1]{$p_{X_2|S_2S_3}$};
  \node(v1)[v, below of=u1]{$p_{S_1S_2}$}; \node(v2)[v, below of=u2]{$p_{S_3}$};
  \node(d1)[d, above of=u1]{}; \node(d2)[d, above of=u2]{};
  \path[every node/.style={}]
  (u1)edge node[left]{$s_1$}(v1) (u2)edge node[right]{$s_2$}(v1) (u2)edge node[right]{$s_3$}(v2)
  (u1)edge[->-out=.6] node[above left]{$x_1$}(d1) (u2)edge[->-out=.6] node[above right]{$x_2$}(d2);   
  \end{tikzpicture}
\caption{Example~\ref{Ex:Prob1}.}
\label{fig:Prob_Ex1}
\end{figure}
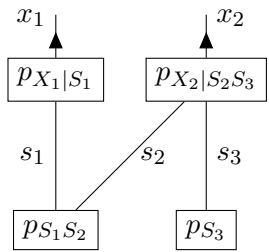
%\end{wrapfigure}
\begin{Example}
Let $(S_1, S_2)$ be jointly distributed according to $p_{S_1S_2}$ and $S_3$ be distributed according to $p_{S_3}$, where $(S_1, S_2)$ is
independent of $S_3$. Suppose that $(X_1, X_2)$ depends on $(S_1, S_2, S_3)$ according to conditional distribution
$p_{X_1X_2|S_1S_2S_3}(x_1, x_2|s_1, s_2, s_3):= p_{X_1|S_1}(x_1|s_1) p_{X_2|S_2S_3}(x_2|s_2, s_3)$, it is then easy to verify that 
the joint distribution $p_{X_1X_2}(x_1,x_2)$ is given by the sum-of-products form
$\big\langle p_{S_1S_2}, p_{S_3}, p_{X_1|S_1}, p_{X_2|S_2S_3} \big\rangle,$
and hence, the RVs $(X_1, X_2)$ are represented by the generative NFG in Fig.~\ref{fig:Prob_Ex1}.
\label{Ex:Prob1}
\end{Example}

The NFG in this example is a generative NFG model where $p_{X_1|S_1}$ and $p_{X_2|S_2S_3}$
are interface functions and $p_{S_1S_2}$ and $p_{S_3}$ are latent functions. In this case, the latent functions serve as independent sources of randomness, which ``generate'' the internal RVs ($S_1, S_2$ and $S_3$).  The internal RVs then ``generate'' the external RVs via the interface functions.

In an arbitrary NFG model, since every latent function may be viewed as the joint distribution of its involved internal RVs, subject to a scaling
factor, they can be regarded as independent ``generating sources''; since each interface function is a conditional function, or, up to a scale, a
conditional distribution of the external RV given its internal RVs, the product of these conditional functions may be regarded as the
conditional distribution of all external RVs conditioned on the internal RVs. The product of all local functions is then up to scale the joint
distribution of all external and internal RVs. The semantics of NFG then dictates that the joint distribution shall be summed over all internal
variables, and the resulting exterior function is therefore the distribution of the external RVs, up to scale. In a sense, a generative NFG model
describes how the external random variables are generated from some independent hidden sources.

%% file: generativeNFG_CFG.tex
\subsection{A subclass of generative NFG models is equivalent to CFGs}
\label{section:NFGs-CFGs}
%[\textbf{Add discussion on Fourier transform--- Uncomment and polish the text below\ldots}]
In this section, we rely on the Fourier transform in some of the discussions.
Let $\X$ be a finite abelian group, we use $\X^{\wedge}$ to denote the character group (written additively) of $\X$,
defined as the set of homomorphisms from $\X$ to $\C$. It is well known that $(\X^{\wedge})^{\wedge}$ is isomorphic
to $\X$, and for any $x \in \X$ and $\hat{x} \in \X^{\wedge}$, $x(\hat{x}) = \hat{x}(x)$. We use
$\kappa_{\X}(x,\hat{x})$ to denote both $x(\hat{x})$ and $\hat{x}(x)$, and use $\hat{\kappa}_{\X}(x,\hat{x})$ to denote
$\kappa(x,-\hat{x})/|\X|$. For any function on $\X$, we define its Fourier transform as the sum-of-product form %induced by the
%bivariate function $\kappa_{\X}(x,\hat{x})$. That is, for any function $f \in \C^{\X}$, its Fourier transform is defined
%as
$
\widehat{f}(\hat{x}) = \big\langle \kappa_{\X}(x,\hat{x}), f(x) \big\rangle
$, for all $\hat{x} \in \X^{\wedge}$.
It is not hard to show that $\kappa$ and $\hat{\kappa}$ are an inverse-pair, and hence given $\widehat{f}$, one may
recover $f$ using the Fourier inverse as
$f(x)=\big\langle \widehat{f}(\hat{x}), \hat{\kappa}(x,\hat{x}) \big\rangle,$ for all $x \in \X$. 
It is well known that if $\X$ is the direct product $\prod_{i\in I} \X_{i}$ of the finite abelian groups $\X_i$, 
then $(\X)^{\wedge}$ is the direct product $\prod_{i \in I} \X_{i}^{\wedge}$, and it follows that
$
\kappa_{\X}(x,\hat{x}) = \prod_{i \in I} \kappa_{\X_i}(x_i,\hat{x}_i),
$
for all $(x,\hat{x}) \in \X \times \X^{\wedge}$, and similarly for $\hat{\kappa}_{\X}$.

Suppose that a generative NFG model is such that every interface function is a sum indicator function. We may convert such an NFG to a CFG according to the following procedure: {\em For each interface vertex, replace it by a variable vertex representing its half-edge variable and  remove the half edge.}

\begin{Prop}
If  in a  generative NFG model all interface functions are sum indicators, then the above procedure gives rise to a CFG equivalent to the NFG. 
\label{prop:conv-NFG}
\end{Prop}
\begin{proof}
The proof follows the following steps: 
1) Modify the NFG by replacing each interface function with a parity check indicator and
inserting a degree two parity check indicator (a sign inverter) on each half edge, Fig.~\ref{fig:conv-NFG}~(b). 
(This does not alter the exterior function due to the relation between the sum and the parity indicator function.) 
2) Perform a holographic transformation on the resulting NFG by inserting the inverse-pair $\kappa_{\X_e}$ 
and $\hat{\kappa}_{\X_e}$ into each regular edge $e$ (with $\kappa_{\X_e}$ adjacent to a hidden function or an inserted sign inverter),
and inserting the transformers $\kappa_{\X_e}$ into each dangling edge $e$, Fig.~(c). 
3) By noting that (up to a scaling factor%
\footnote{It is not hard to show that all the scaling factors cancel out, and hence, all subsequent equalities are exact.}%
) the Fourier and Fourier inverse of $\delta_{+}$ are $\delta_{=}$, we obtain (after deleting all degree-two 
equality indicators resulting from the sign inverters) a constrained NFG in which each interface function is an equality indicator and each hidden 
function is the Fourier transform of the corresponding hidden function in the original NFG, Fig.~(d). Hence, from the GHT %Theorem~? %\ref{Thm:GHT} 
and Proposition~\ref{prop:mul-NFG}, we have
$
\widehat{Z}_{\G}(\hat{x}_{I}) = \prod \limits_{j \in J} \widehat{f}_{j}(\hat{x}_{{\rm ne}(j)}),
$
and the claim follows from the multivariate multiplication-convolution duality theorem under the Fourier transform \cite{Mao2005:FGFT}.
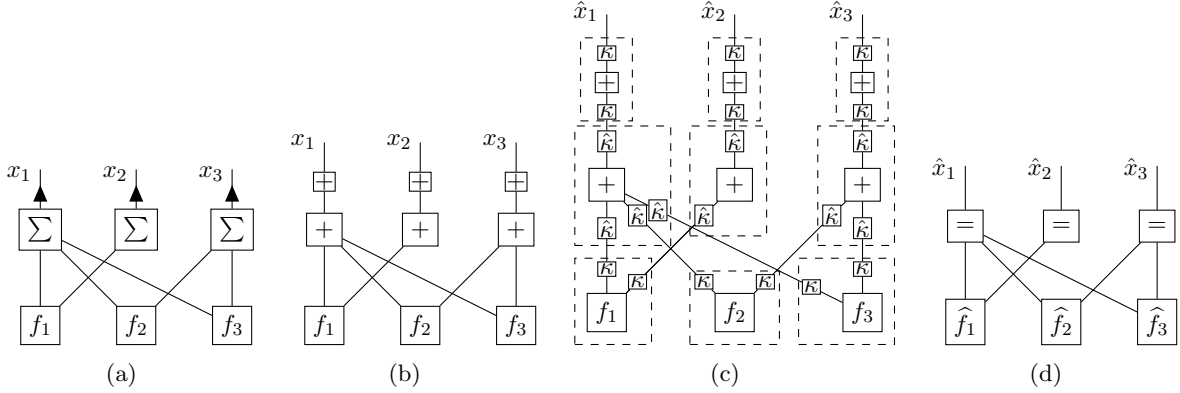
\begin{figure}[ht]
\centering
%\subfigure[]{\includegraphics[scale = .85]{fig/NFG_conv_fig.eps}}		\hspace{.5cm}
%\subfigure[]{\includegraphics[scale = .85]{fig/NFG_conv2_fig.eps}}	\hspace{.5cm}
%\subfigure[]{\includegraphics[scale = .85]{fig/NFG_conv3_fig.eps}}	\hspace{.5cm}
\subfigure[]{
   \begin{tikzpicture}[scale=.85, v/.style={node distance=1.5cm, draw, rectangle}, d/.style={node distance=1.2cm}, every node/.append style={transform shape, minimum size=5mm}]
      \node(u1)[v]{$\sum$}; \node(u2)[v, right of=u1]{$\sum$}; \node(u3)[v, right of=u2]{$\sum$}; 
      \node(v1)[v, below of=u1]{$f_1$}; \node(v2)[v, below of=u2]{$f_2$}; \node(v3)[v, below of=u3]{$f_3$}; 
      \node(d1)[d, above of=u1]{}; \node(d2)[d, above of=u2]{}; \node(d3)[d, above of=u3]{}; 
      \path[every node/.style={transform shape}]
      (u1)edge(v1)edge(v2)edge(v3)	(u2)edge(v1)	(u3)edge(v2)edge(v3)
      (u1)edge[->-out=.6] node[above left]{$x_1$}(d1)	(u2)edge[->-out=.6] node[above left]{$x_2$}(d2)	(u3)edge[->-out=.6] node[above left]{$x_3$}(d3);
    \end{tikzpicture} }
%    \hspace{.5cm}
\subfigure[]{
   \begin{tikzpicture}[scale=.85, v/.style={node distance=\dist, draw, rectangle}, d/.style={node distance=\dist/2}, 
     s/.style={node distance=\dist/2, draw, rectangle, inner sep=.2mm}, every node/.append style={transform shape}]
     \def\dist{1.5cm};
      \node(u1)[v]{$+$}; \node(u2)[v, right of=u1]{$+$}; \node(u3)[v, right of=u2]{$+$}; 
      \node(v1)[v, below of=u1]{$f_1$}; \node(v2)[v, below of=u2]{$f_2$}; \node(v3)[v, below of=u3]{$f_3$}; 
      \node(s1)[s, above of=u1]{$+$}; \node(s2)[s, above of=u2]{$+$}; \node(s3)[s, above of=u3]{$+$}; 
      \node(d1)[d, above of=s1]{}; \node(d2)[d, above of=s2]{}; \node(d3)[d, above of=s3]{}; 
      \path[every node/.style={transform shape}]
      (u1)edge(v1)edge(v2)edge(v3)	(u2)edge(v1)	(u3)edge(v2)edge(v3)
      (s1)edge node[above left]{$x_1$}(d1)edge(u1)	(s2)edge node[above left]{$x_2$}(d2)edge(u2)	(s3)edge node[above left]{$x_3$}(d3)edge(u3);
    \end{tikzpicture} }
%    \hspace{.5cm}
\subfigure[]{
   \begin{tikzpicture}[scale=.85, v/.style={node distance=\dist, draw, rectangle}, d/.style={node distance=\dist*.6}, 
     s/.style={node distance=\dist*.8, draw, rectangle, inner sep=.3mm}, t/.style={node distance=\dist*.34, draw, rectangle, inner sep=.3mm}, 
     ts/.style={node distance=\dist*.23, draw, rectangle, inner sep=.3mm}, every node/.append style={transform shape}]
     \def\dist{2cm};
      \node(u1)[v]{$+$}; \node(u2)[v, right of=u1]{$+$}; \node(u3)[v, right of=u2]{$+$}; 
      \node(s1)[s, above of=u1]{$+$}; \node(s2)[s, above of=u2]{$+$}; \node(s3)[s, above of=u3]{$+$}; 
      \node(v1)[v, below of=u1]{$f_1$}; \node(v2)[v, below of=u2]{$f_2$}; \node(v3)[v, below of=u3]{$f_3$}; 
      \node(d1)[d, above of=s1]{}; \node(d2)[d, above of=s2]{}; \node(d3)[d, above of=s3]{}; 
      \node(A11)[t, below of=u1]{$\hat{\kappa}$}; \node(A12)[t, below right of=u1]{$\hat{\kappa}$}; \node(A13) at (\dist*.4,-\dist*.2)[t]{$\hat{\kappa}$};
      \node(B11)[t, above of=v1]{$\kappa$}; \node(B12)[t, above left of=v2]{$\kappa$}; \node(B13) at (\dist*1.6,-\dist*.8)[t]{$\kappa$};
      \node(A21)[t, below left of=u2]{$\hat{\kappa}$}; 
      \node(B21)[t, above right of=v1]{$\kappa$}; 
      \node(A32)[t, below left of=u3]{$\hat{\kappa}$}; \node(A33)[t, below of=u3]{$\hat{\kappa}$};
      \node(B32)[t, above right of=v2]{$\kappa$}; \node(B33)[t, above of=v3]{$\kappa$};
      \node(A1)[t, above of=u1]{$\hat{\kappa}$}; \node(A2)[t, above of=u2]{$\hat{\kappa}$}; \node(A3)[t, above of=u3]{$\hat{\kappa}$}; 
      \node(B1)[ts, below of=s1]{$\kappa$}; \node(B2)[ts, below of=s2]{$\kappa$}; \node(B3)[ts, below of=s3]{$\kappa$}; 
      \node(sB1)[ts, above of=s1]{$\kappa$}; \node(sB2)[ts, above of=s2]{$\kappa$}; \node(sB3)[ts, above of=s3]{$\kappa$}; 
      \path[every node/.style={transform shape}]
      (u1)edge(A11)edge(A12)edge(A13)
      (A11)edge(B11) (B11)edge(v1) (A12)edge(B12) (B12)edge(v2) (A13)edge(B13) (B13)edge(v3)
      (u2)edge(A21)	(A21)edge(B21) 	(B21)edge(v1)
      (A21)edge(B21) (B21)edge(v1)
      (u3)edge(A32)edge(A33)
      (A32)edge(B32) (B32)edge(v2) (A33)edge(B33) (B33)edge(v3)
      (u1)edge(A1) (A1)edge(B1) (B1)edge(s1) (s1)edge(sB1) (sB1)edge node[above left]{$\hat{x}_1$}(d1)
      (u2)edge(A2) (A2)edge(B2) (B2)edge(s2) (s2)edge(sB2) (sB2)edge node[above left]{$\hat{x}_2$}(d2)
      (u3)edge(A3) (A3)edge(B3) (B3)edge(s3) (s3)edge(sB3) (sB3)edge node[above left]{$\hat{x}_3$}(d3);
      \draw[dashed](v1)\clbox{-.5}{-.5}{.7}{.85}; \draw[dashed](v2)\clbox{-.7}{-.5}{.7}{.65}; \draw[dashed](v3)\clbox{-1}{-.5}{.5}{.85};
      \draw[dashed](u1)\clbox{-.5}{-.95}{1}{.92}; \draw[dashed](u2)\clbox{-.7}{-.8}{.5}{.92}; \draw[dashed](u3)\clbox{-.7}{-.95}{.5}{.92};
      \draw[dashed](s1)\clbox{-.41}{-.6}{.4}{.7};
      \draw[dashed](s2)\clbox{-.41}{-.6}{.4}{.7};
      \draw[dashed](s3)\clbox{-.41}{-.6}{.4}{.7};
    \end{tikzpicture} }
%    \hspace{.5cm}
\subfigure[]{
   \begin{tikzpicture}[scale=.85, v/.style={node distance=1.5cm, draw, rectangle}, d/.style={node distance=1.2cm}, every node/.append style={transform shape, minimum size=5mm}]
      \node(u1)[v]{$=$}; \node(u2)[v, right of=u1]{$=$}; \node(u3)[v, right of=u2]{$=$}; 
      \node(v1)[v, below of=u1]{$\widehat{f}_1$}; \node(v2)[v, below of=u2]{$\widehat{f}_2$}; \node(v3)[v, below of=u3]{$\widehat{f}_3$}; 
      \node(d1)[d, above of=u1]{}; \node(d2)[d, above of=u2]{}; \node(d3)[d, above of=u3]{}; 
      \path[every node/.style={transform shape}]
      (u1)edge(v1)edge(v2)edge(v3)	(u2)edge(v1)	(u3)edge(v2)edge(v3)
      (u1)edge node[above left]{$\hat{x}_1$}(d1)	(u2)edge node[above left]{$\hat{x}_2$}(d2)	
      (u3)edge node[above left]{$\hat{x}_3$}(d3);
    \end{tikzpicture} }
\caption{Proof of Proposition~\ref{prop:conv-NFG}: (a) An example NFG, (b) Step 1, (c) Step 2, and (d) Step 3.}
\label{fig:conv-NFG}
\end{figure}
\end{proof}
An example illustrating the steps of the proof is shown in Fig.~\ref{fig:conv-NFG}.
Of course one may attempt to prove the claim by direct evaluation of the exterior function as demonstrated in the following example.

\begin{Example} Consider the NFG 
in Figure \ref{fig:Prob_Ex1}, where the interface functions are taken as sum indicators. Then the exterior function of this NFG is
\noindent $
\sum\limits_{s_1, s_2, s_3}
p_{S1S2}(s_1, s_2) p_{S_3}(s_3) \delta_{\sum}(x_1, s_1) \delta_{\sum}(x_2, s_2, s_3)
$
%\noindent
%$ \stackrel{(a)}{=}  
%\sum\limits_{s_1, s_2, s_3} p_{S1S2}(s_1, s_2) p_{S_3}(s_3) \delta_{=}(x_1, s_1) \delta_{\sum}(x_2, s_2, s_3)$
\noindent
$ \stackrel{(a)}{=}  \sum\limits_{s_2, s_3}  p_{S_1S_2}(x_1, s_2) p_{S3}(s_3) \delta_{\sum}(x_2, s_2, s_3)$%\\
%\noindent
$ \stackrel{(b)}{=}  p_{S_1S_2}(x_1, x_2)* p_{S3}(x_2)$,
where (a) identifies the bivariate sum indicator with equality indicator and removes it, %(b) then removes it, 
and (b) is due to Lemma \ref{lem:conv_sum}. The reader is invited to examine the structure of the original NFG and that of the CFG representing the above convolutional factorization. 
\end{Example}

Indeed, for any generative NFG model in which interface functions are all sum indicators, the procedure above
Proposition~\ref{prop:conv-NFG} applied to an interface function is
equivalent to either applying step (a) above (for degree-2 vertices) or applying Lemma \ref{lem:conv_sum} (for vertices of degree higher than 2).

It is easy to see that the this procedure is reversible, in the sense that one may apply it in reverse direction and convert any CFG to a generative NFG model with all interface functions being sum indicators. Figure \ref{fig:EX1_Conv} shows an equivalent pair of CFG and generative NFG model.

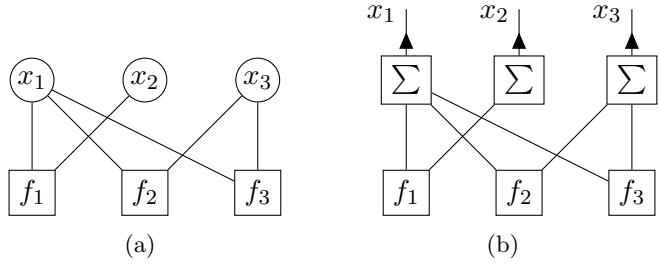
\begin{figure}[ht]
\centering
    \setcounter{subfigure}{0}
%	\subfigure{\includegraphics[scale = .7]{fig/FG_fig.eps}} \hspace{1.5cm}
%	\subfigure{\includegraphics[scale = .7]{fig/NFG_Conv_fig.eps}}
    \subfigure[]{
    \begin{tikzpicture}[var/.style={node distance=1.5cm, draw, circle, inner sep=.5mm, minimum size=4mm}, 
      fun/.style={node distance=1.5cm, draw, rectangle, minimum size=4mm}]
      \node(x1)[var]{$x_1$}; \node(x2)[var, right of=x1]{$x_2$}; \node(x3)[var, right of=x2]{$x_3$};
      \node(f1)[fun, below of=x1]{$f_1$}; \node(f2)[fun, right of=f1]{$f_2$}; \node(f3)[fun, right of=f2]{$f_3$};
      \path[every node/.style={}]
      (x1)edge(f1)edge(f2)edge(f3)
      (x2)edge(f1)
      (x3)edge(f2)edge(f3);
    \end{tikzpicture} }
    \hspace{.5cm}
    \subfigure[]{
    \begin{tikzpicture}[v/.style={node distance=1.5cm, draw, rectangle}, d/.style={node distance=1.2cm}, every node/.append style={transform shape, minimum size=5mm}]
      \node(u1)[v]{$\sum$}; \node(u2)[v, right of=u1]{$\sum$}; \node(u3)[v, right of=u2]{$\sum$}; 
      \node(v1)[v, below of=u1]{$f_1$}; \node(v2)[v, below of=u2]{$f_2$}; \node(v3)[v, below of=u3]{$f_3$}; 
      \node(d1)[d, above of=u1]{}; \node(d2)[d, above of=u2]{}; \node(d3)[d, above of=u3]{}; 
      \path[every node/.style={transform shape}]
      (u1)edge(v1)edge(v2)edge(v3)	(u2)edge(v1)	(u3)edge(v2)edge(v3)
      (u1)edge[->-out=.5] node[above left]{$x_1$}(d1)	(u2)edge[->-out=.5] node[above left]{$x_2$}(d2)	(u3)edge[->-out=.5] node[above left]{$x_3$}(d3);
    \end{tikzpicture} }
    \caption{An equivalent pair of CFG (left) and generative NFG model (right).}	
    \label{fig:EX1_Conv}
\end{figure}

%The following result was first shown in \cite[Theorem 1]{Mao:UAI2004}, and is a direct consequence of Proposition~\ref{prop:conv-NFG}.
%\begin{Cor}
%Let $P(I \cup J, E, E^{\rm ext}, p_{I \cup J})$ be a simple composite RVs problem
%such that $p_i(z_i|x_{E(i)}) = \delta_{\Sigma}(z_i,x_{E(i)})$ for all $i$, then
%\[
%p_{Z_IX_{E^{\rm ext}}}(z_I,x_{E^{\rm ext}}) = \prod^{*} \limits_{j \in J} p_{j}(x_{E^{\rm ext}(j)}, z_{{\rm ne}(j)}),
%\]
%\end{Cor}
%\begin{proof}
%By Lemma~\ref{lemma:P-NFG}, the joint distribution $p_{Z_IX_{E^{\rm ext}}}$ is the exterior 
%function of the generative NFG $\G(I \cup J, E \cup I, E^{\rm ext} \cup I, p_{I \cup J})$, and Proposition~\ref{prop:conv-NFG}
%completes the proof.
%\end{proof}

%Below we continue Example~\ref{Ex:Prob2} by specifying the conditional distributions as sum indicators, cf. \cite[Example~2]{Mao:UAI2004}.
%\begin{Example}[continues = Ex:Prob2] See \cite[Example~2]{Mao:UAI2004}.
%If we choose $p_a(x_a|x_3,x_6,x_7) = \delta_{\Sigma}(x_a,x_3,x_6,x_7)$ and $p_b(x_b|x_4,x_8) = \delta_{\Sigma}(x_b,x_4,x_8)$.
%Then $X_{E^{\rm ext}}$ is represented by the NFG in Fig.~\ref{fig:ProbEx2}~(b), which by the previous proposition is equivalent to the CFG in Fig.~\ref{fig:ProbEx2}~(c). 
%%\begin{figure}[ht]
%\centering
%\subfigure[]{\includegraphics[scale = .7]{fig/ProbEx2_Sum_fig.eps}} \hspace{1cm}
%\subfigure[]{\includegraphics[scale = .7]{fig/ProbEx2_CFG_fig.eps}}
%\caption{Example~\ref{Ex:Prob2}.}
%\label{fig:ProbEx2_conv}
%\end{figure}
%\end{Example}

%% file: independence.tex
\subsection{Independence}
%\subsubsection{Conditional independence in constrained NFGs}
%\label{section:const_indep}

We now show that there exists a ``duality'' between a constrained NFG model and a generative NFG model in their implied independence properties. 

\begin{figure}
\centering
%\begin{tabular}{cc}
%\scalebox{0.7}{\includegraphics{fig/const_indep_fig.eps}} &
%\scalebox{0.7}{\includegraphics{fig/marg-pxyz_fig.eps}}\\
%(a) & (b)
%\end{tabular}
    \subfigure[]{
    \begin{tikzpicture}[v/.style={node distance=1.5cm, draw, rectangle}, d/.style={node distance=1.2cm}, every node/.append style={transform shape, minimum size=5mm}]
      \node(u1)[v]{$g_1$}; \node(u2)[v, right of=u1]{$g_2$}; \node(u3)[v, right of=u2]{$g_3$}; 
      \node(v1)[v, below of=u1]{$f_1$}; \node(v2)[v, below of=u3]{$f_2$}; 
      \node(d1)[d, above of=u1]{}; \node(d2)[d, above of=u2]{}; \node(d3)[d, above of=u3]{}; 
      \path[every node/.style={transform shape}]
      (u1)edge(v1)	(u2)edge(v1)edge(v2)	(u3)edge(v2)
      (u1)edge[->-in=.2] node[above left]{$x$}(d1)	(u2)edge[->-in=.2] node[above left]{$y$}(d2)	(u3)edge[->-in=.2] node[above left]{$z$}(d3);
    \end{tikzpicture} }
    \hspace{.5cm}
    \subfigure[]{
    \begin{tikzpicture}[v/.style={node distance=1.5cm, draw, rectangle}, d/.style={node distance=1.2cm}, every node/.append style={transform shape, minimum size=5mm}]
      \node(u1)[v]{$g_1$}; \node(u2)[v, right of=u1]{$g_2$}; \node(u3)[v, right of=u2]{$g_3$}; 
      \node(v1)[v, below of=u1]{$f_1$}; \node(v2)[v, below of=u3]{$f_2$}; 
      \node(d1)[d, above of=u1]{}; \node(d2)[d, above of=u2]{}; \node(d3)[d, above of=u3]{}; 
      \path[every node/.style={transform shape}]
      (u1)edge(v1)	(u2)edge(v1)edge(v2)	(u3)edge(v2)
      (u1)edge[->-out=.5] node[above left]{$x$}(d1)	(u2)edge[->-out=.5] node[above left]{$y$}(d2)	(u3)edge[->-out=.5] node[above left]{$z$}(d3);
    \end{tikzpicture} }
\caption{(a) $X\independent Z| Y$ and (b)  $X\independent Z$.}
\label{fig:indep}
\end{figure}

\begin{Lemma}
For the NFG models in Fig.~\ref{fig:indep}, we have $X\independent Z | Y$ in the constrained model
and $X \independent Z$ in the generative model.
\label{lemma:indep}
\end{Lemma}

\begin{proof}
For the constrained NFG, it is sufficient to show that $p(x,y,z)$ is a split function via $y$, see e.g. \cite{Lauritzen:GM}.
To this end, we have $g_2$ is a split function via $y$, say it splits into bivariate functions
$g_{2,1}$ and $g_{2,2}$. Hence, applying the vertex splitting procedure for $g_2$ followed
by the vertex merging of each hidden function and its adjacent bivariate function, it becomes clear that
$p(x,y,z) = f'_{1}(x,y) f'_{2}(y,z)$ where $f'_1 = \langle f_1, g_1,g_{2,1}\rangle$ and
$f'_2 = \langle f_2, g_3,g_{2,2}\rangle$.
%\begin{figure}[ht]
%\centering
%    \subfigure[]{
%    \begin{tikzpicture}[v/.style={node distance=\dist, draw, rectangle}, d/.style={node distance=1.2cm}, 
%      s/.style={node distance=\dist*.5, draw, rectangle, inner sep=.5mm}, every node/.append style={transform shape, minimum size=5mm}]
%      \def\dist{1.5cm};
%      \node(u1)[v, draw=none]{}; \node(u2)[v, right of=u1]{$g_2$}; \node(u3)[v, draw=none, right of=u2]{}; 
%      \node(s1)[s, below of=u1]{$g_1$}; \node(s21)[s] at ([shift={(-135:\dist/sqrt(2))}]u2) {$g_{2,1}$}; 
%      \node(s22)[s]at([shift={(-45:\dist/sqrt(2))}]u2){$g_{2,2}$}; \node(s3)[s, below of=u3]{$g_3$};
%      \node(v1)[v, below of=u1]{$f_1$}; \node(v2)[v, below of=u3]{$f_2$}; 
%      \node(d1)[d, above of=u1]{}; \node(d2)[d, above of=u2]{}; \node(d3)[d, above of=u3]{}; 
%      \path[every node/.style={transform shape}]
%      (u1)edge(v1)	(u2)edge(v1)edge(v2)	(u3)edge(v2)
%      (u1)edge[->-in=.2] node[above left]{$x$}(d1)	(u2)edge[->-in=.2] node[above left]{$y$}(d2)	(u3)edge[->-in=.2] node[above left]{$z$}(d3);
%    \end{tikzpicture} }
%\caption{Proof of Lemma~\ref{lemma:indep} for the constrained model.}
%\end{figure}

For the generative model, we prove the claim graphically in Fig.~\ref{fig:marginal-indep}.
Marginalizing $y$, the probability distribution $p(x,z)$ is realized by the NFG in Fig.~\ref{fig:marginal-indep}~(a), which, by the definition of a conditional function, is equivalent to the one in (b). Marginalizing again, we obtain the
NFGs in (c) and (d) for the marginals $p(x)$ and $p(z)$, respectively. Hence, the multiplication $p(x)p(z)$ is realized by the NFG in (e), which is equivalent (again by the definition of a conditional function) to the one in (f). Noting that the NFG in the left side of (f) realizes the scalar 1,
and comparing with (a), we see that $p(x,z)$ and $p(x)p(z)$ are realized by the same NFG, and hence, must be equal.
\begin{figure}[ht]
\centering
    \subfigure[$p(x,z)$]{
    \begin{tikzpicture}[scale=.9, v/.style={node distance=1.5cm, draw, rectangle},
      d/.style={node distance=1.2cm}, every node/.append style={transform shape, minimum size=5mm}]
      \node(u1)[v]{$g_1$}; \node(u2)[v, right of=u1]{$g_2$}; \node(u3)[v, right of=u2]{$g_3$}; 
      \node(v1)[v, below of=u1]{$f_1$}; \node(v2)[v, below of=u3]{$f_2$}; 
      \node(d1)[d, above of=u1]{}; \node(d2)[d, draw, above of=u2]{\textbf{1}}; \node(d3)[d, above of=u3]{}; 
      \path[every node/.style={transform shape}]
      (u1)edge(v1)	(u2)edge(v1)edge(v2)	(u3)edge(v2)
      (u1)edge[->-out=.5] node[above left]{$x$}(d1)	(u2)edge[->-out=.5] (d2)	(u3)edge[->-out=.5] node[above left]{$z$}(d3);
    \end{tikzpicture} }
    \hspace{.25cm}
    \subfigure[$p(x,z)$]{
    \begin{tikzpicture}[scale=.9, v/.style={node distance=1.5cm, draw, rectangle}, voff/.style={node distance=.3cm, draw, rectangle}, 
      d/.style={node distance=1.2cm}, every node/.append style={transform shape, minimum size=5mm}]
      \node(u1)[v]{$g_1$}; 
      \node(u2)[v, draw=none, right of=u1]{}; \node(u2l)[voff, left of=u2, label=left:$c$]{\textbf{1}}; \node(u2r)[voff, right of=u2]{\textbf{1}}; 
      \node(u3)[v, right of=u2]{$g_3$}; 
      \node(v1)[v, below of=u1]{$f_1$}; \node(v2)[v, below of=u3]{$f_2$}; 
      \node(d1)[d, above of=u1]{}; \node(d3)[d, above of=u3]{}; 
      \path[every node/.style={transform shape}]
      (u1)edge(v1)	(u2l)edge(v1) (u2r)edge(v2)	(u3)edge(v2)
      (u1)edge[->-out=.5] node[above left]{$x$}(d1)	 	(u3)edge[->-out=.5] node[above left]{$z$}(d3);
    \end{tikzpicture} }
    \hspace{.25cm}
    \subfigure[$p(x)$]{
    \begin{tikzpicture}[scale=.9, v/.style={node distance=1.5cm, draw, rectangle}, voff/.style={node distance=.3cm, draw, rectangle}, 
      d/.style={node distance=1.2cm}, every node/.append style={transform shape, minimum size=5mm}]
      \node(u1)[v]{$g_1$}; 
      \node(u2)[v, draw=none, right of=u1]{}; \node(u2l)[voff, left of=u2, label=left:$c$]{\textbf{1}}; \node(u2r)[voff, right of=u2]{\textbf{1}}; 
      \node(u3)[v, right of=u2]{\textbf{1}}; 
      \node(v1)[v, below of=u1]{$f_1$}; \node(v2)[v, below of=u3]{$f_2$}; 
      \node(d1)[d, above of=u1]{}; \node(d3)[d, above of=u3]{}; 
      \path[every node/.style={transform shape}]
      (u1)edge(v1)	(u2l)edge(v1) (u2r)edge(v2)	(u3)edge(v2)
      (u1)edge[->-out=.5] node[above left]{$x$}(d1);	 	
    \end{tikzpicture} }
    \hspace{.25cm}
    \subfigure[$p(z)$]{
    \begin{tikzpicture}[scale=.9, v/.style={node distance=1.5cm, draw, rectangle}, voff/.style={node distance=.3cm, draw, rectangle}, 
      d/.style={node distance=1.2cm}, every node/.append style={transform shape, minimum size=5mm}]
      \node(u1)[v]{\textbf{1}}; 
      \node(u2)[v, draw=none, right of=u1]{}; \node(u2l)[voff, left of=u2, label=left:$c$]{\textbf{1}}; \node(u2r)[voff, right of=u2]{\textbf{1}}; 
      \node(u3)[v, right of=u2]{$g_3$}; 
      \node(v1)[v, below of=u1]{$f_1$}; \node(v2)[v, below of=u3]{$f_2$}; 
      \node(d1)[d, above of=u1]{}; \node(d3)[d, above of=u3]{}; 
      \path[every node/.style={transform shape}]
      (u1)edge(v1)	(u2l)edge(v1) (u2r)edge(v2)	(u3)edge(v2)
      (u3)edge[->-out=.5] node[above left]{$z$}(d3);
    \end{tikzpicture} }
    \hspace{.25cm}
    \subfigure[$p(x)p(z)$]{
    \begin{tikzpicture}[scale=.9, v/.style={node distance=1.5cm, draw, rectangle}, voff/.style={node distance=.3cm, draw, rectangle}, 
      d/.style={node distance=1.2cm}, every node/.append style={transform shape, minimum size=5mm}]
      \node(u1)[v]{\textbf{1}}; 
      \node(u2)[v, draw=none, right of=u1]{}; \node(u2l)[voff, left of=u2, label=left:$c$]{\textbf{1}}; \node(u2r)[voff, right of=u2]{\textbf{1}}; 
      \node(u3)[v, right of=u2]{$g_3$}; 
      \node(v1)[v, below of=u1]{$f_1$}; \node(v2)[v, below of=u3]{$f_2$}; 
      \node(d1)[d, above of=u1]{}; \node(d3)[d, above of=u3]{}; 
      \path[every node/.style={transform shape}]
      (u1)edge(v1)	(u2l)edge(v1) (u2r)edge(v2)	(u3)edge(v2);
      \draw[dashed] (u2) \clbox{-1}{-.5}{1}{.5};
    \end{tikzpicture} 
    \begin{tikzpicture}[scale=.9, v/.style={node distance=1.5cm, draw, rectangle}, voff/.style={node distance=.3cm, draw, rectangle}, 
      d/.style={node distance=1.2cm}, every node/.append style={transform shape, minimum size=5mm}]
      \node(u1)[v]{$g_1$}; 
      \node(u2)[v, draw=none, right of=u1]{}; \node(u2l)[voff, left of=u2, label=left:$c$]{\textbf{1}}; \node(u2r)[voff, right of=u2]{\textbf{1}}; 
      \node(u3)[v, right of=u2]{$g_3$}; 
      \node(v1)[v, below of=u1]{$f_1$}; \node(v2)[v, below of=u3]{$f_2$}; 
      \node(d1)[d, above of=u1]{}; \node(d3)[d, above of=u3]{}; 
      \path[every node/.style={transform shape}]
      (u1)edge(v1)	(u2l)edge(v1) (u2r)edge(v2)	(u3)edge(v2)
      (u1)edge[->-out=.5] node[above left]{$x$}(d1)	 	(u3)edge[->-out=.5] node[above left]{$z$}(d3);
      \draw[dashed] (u2) \clbox{-1}{-.5}{1}{.5};
    \end{tikzpicture} }
    \hspace{.25cm}
    \subfigure[$p(x)p(z)$]{
    \begin{tikzpicture}[scale=.9, v/.style={node distance=1.5cm, draw, rectangle},
      d/.style={node distance=1.2cm}, every node/.append style={transform shape, minimum size=5mm}]
      \node(u1)[v]{\textbf{1}}; \node(u2)[v, right of=u1]{$g_2$}; \node(u3)[v, right of=u2]{\textbf{1}}; 
      \node(v1)[v, below of=u1]{$f_1$}; \node(v2)[v, below of=u3]{$f_2$}; 
      \node(d1)[d, above of=u1]{}; \node(d2)[d, draw, above of=u2]{\textbf{1}}; \node(d3)[d, above of=u3]{}; 
      \path[every node/.style={transform shape}]
      (u1)edge(v1)	(u2)edge(v1)edge(v2)	(u3)edge(v2)
      	(u2)edge[->-out=.5] (d2);	
    \end{tikzpicture} 
    \begin{tikzpicture}[scale=.9, v/.style={node distance=1.5cm, draw, rectangle},
      d/.style={node distance=1.2cm}, every node/.append style={transform shape, minimum size=5mm}]
      \node(u1)[v]{$g_1$}; \node(u2)[v, right of=u1]{$g_2$}; \node(u3)[v, right of=u2]{$g_3$}; 
      \node(v1)[v, below of=u1]{$f_1$}; \node(v2)[v, below of=u3]{$f_2$}; 
      \node(d1)[d, above of=u1]{}; \node(d2)[d, draw, above of=u2]{\textbf{1}}; \node(d3)[d, above of=u3]{}; 
      \path[every node/.style={transform shape}]
      (u1)edge(v1)	(u2)edge(v1)edge(v2)	(u3)edge(v2)
      (u1)edge[->-out=.5] node[above left]{$x$}(d1)	(u2)edge[->-out=.5] (d2)	(u3)edge[->-out=.5] node[above left]{$z$}(d3);
    \end{tikzpicture} }
 \caption{Proof of Lemma~\ref{lemma:indep} for the generative model.} %: (a) , (b) $p(x,z)$, (c) $p(x)$, (d) $p(z)$, (e) $p(x)p(z)$, and (f) $p(x)p(z)$.	}
\label{fig:marginal-indep}
\end{figure}
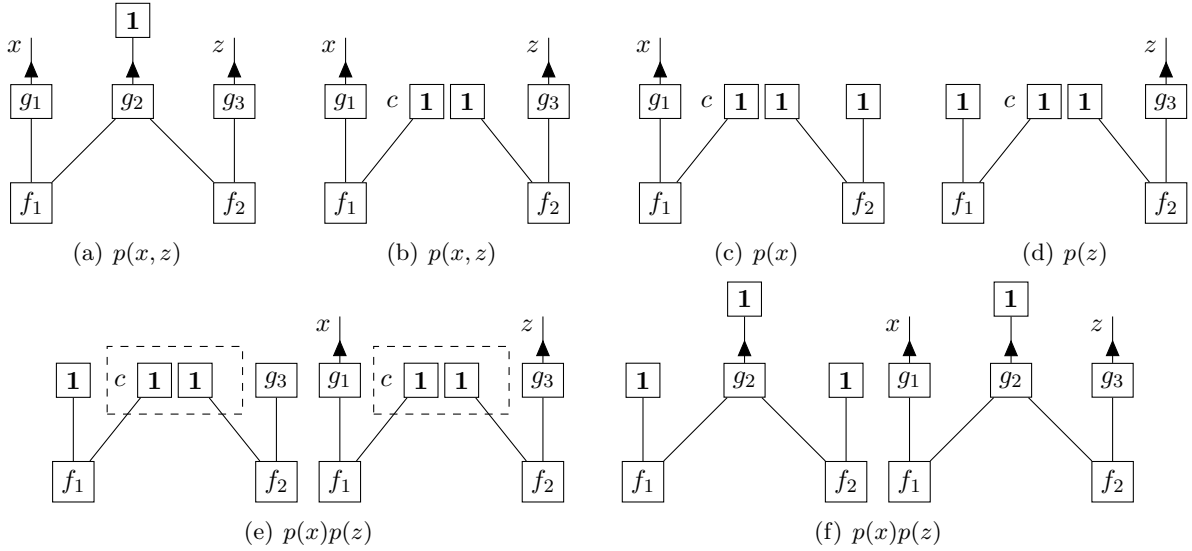 
\end{proof}

We remark that the conditional independence part of the lemma can be proved graphically in a similar manner to the marginal
independence part where marginalization is replaced with evaluation. Further, it is interesting to observe that the two NFG models
have identical graphs, but the constrained model implies  conditional independence property whereas the generative model implies
the ``dual'' marginal independence property.

In an NFG model, suppose that  ${\cal A}$, ${\cal B}$ and ${\cal S}$ are three disjoint subsets of vertices. We say ${\cal A}$ and ${\cal B}$ are {\em
separated} by ${\cal S}$ if every path from a vertex in ${\cal A}$ to a vertex in ${\cal B}$ go through some vertex in ${\cal S}$. In this case, if
${\cal A}$, ${\cal B}$ and ${\cal S}$ are all subsets of the interface vertex set $I$, then, recalling that every external variable is also indexed by the interface
vertex it is incident with,  we also say that RV sets $X_{\cal A}$ and $X_{\cal B}$ are separated by the RV set $X_{\cal S}$. 
For any subset $I' \subseteq I$, let ${\rm ne}(I'):=\{ {\rm ne}(i): i\in I'\}$. By merging the vertices in ${\cal A}$
into one vertex, and similarly for the interface nodes ${\cal S}$ and ${\cal B}':=I \backslash ({\cal A}\cup {\cal S})$, and performing the
same merging for the hidden nodes ${\rm ne}({\cal A})$
and $J\backslash {\rm ne}({\cal A})$. Then the resulting NFG has the same graph topology as the ones in
Fig.~\ref{fig:indep}, as it is clear from the separation property that ${\rm ne}({\cal B}') \subseteq J\backslash {\rm ne}(\cal A)$. From the fact that the split
and conditional properties are preserved under such mergings, the previous lemma extends
in a straightforward manner to any NFG model, and we have the following theorem.

\begin{Theorem}
\label{thm:indep}
Let ${\cal G}(I\cup J, E, f_{I\cup J})$ be an NFG model and ${\cal A}$, ${\cal B}$ and ${\cal S}$ be three disjoint interface
vertex subsets, i.e., subsets of $I$. Suppose that $X_{\cal A}$ and $X_{\cal B}$
are separated by $X_{\cal S}$. Then
\begin{enumerate}
\item If the NFG is a constrained NFG model, then $X_{\cal A} \independent X_{\cal B}| X_{\cal S}$.
\item If the NFG is a generative NFG model, then $X_{\cal A} \independent X_{\cal B}.$
\end{enumerate}
\end{Theorem}

Part 1 of Theorem \ref{thm:indep} is essentially the global Markov property (see, e.g., \cite{Lauritzen:GM}) on an FG model (noting that constrained
NFG models are equivalent to FGs). Part 2 of the theorem, in the special case when all interface functions are sum indicators, was proved in
\cite{Mao:UAI2004} in the context of CFGs (noting that such NFG models reduce to CFGs). That is, Part 2 of 
Theorem \ref{thm:indep} generalizes such a result from CFGs to arbitrary generative NFG models.  We now provide some insights for this result.

Consider the NFG in Figure \ref{fig:indep} (b). The fact $X\independent Z$ can be reasoned by the fact that latent
functions $f_1$ and $f_2$, giving rise to $X$ and $Z$ respectively, serve as independent sources of randomness.
Indeed, it is precisely due to $X$ and $Z$ sharing no common latent functions that when $Y$ is ignored $X$ and $Z$
become independent. The same is true for arbitrary generative NFG models, where if $X_{\cal A}$ and $X_{\cal B}$
are separated by $X_{\cal S}$, then we necessarily have $X_{\cal A}$ and $X_{\cal B}$ share no common latent functions.
%It is also this fact that is critical to establishing Part 2 of Theorem \ref{thm:indep}.

We remark that the marginal independence, i.e., Part 2 of Theorem~\ref{thm:indep} holds for the more general class
of NFG models characterized by the property that for each interface function $f_{i}$, it holds that
\[
  \sum_{x_{i}}f_{i}(x_{i}, x_{T(i)}) = \prod_{t\in T(i)}f_{t}(x_{t}),
  \]
  for some univariate functions $f_{t}$. We may refer to an NFG model whose interface functions satisfy this property
  as an \emph{extended generative model}, and it is clear that a generative model is an extended generative model. 
  (A conditional function trivially satisfies the property above.) It is not hard to show
  that the class of constrained models and the class of extended generative models are closed under internal holographic transformations,
  from which, it follows that the independence properties in Theorem~\ref{thm:indep} are invariant under internal holographic
  transformations.

%% file: transformedNFG.tex
\section{Transformed NFG Models}
\label{section:transformed_NFGs}
In some applications, instead of modelling the joint probability distribution of the RVs, we may wish our model to represent a certain transformation
of the joint distribution,  and it becomes clear that the NFG modelling framework introduced in this paper is particularly convenient for this
purpose. Moreover, this framework provides a generic transformation technique and enables an infinite family of such transformations.
In subsequent discussions, a \emph{transformed NFG model}, or simply a transformed model, refers to any NFG obtained from an NFG model (generative or constrained)
by a holographic transformation, where the external transformers in Step (H1) of the holographic transformation are not
necessarily trivial--- At some places we may refer to the original NFG as the \emph{base model}.

Next we show that a particular class of NFG models, upon an appropriate choice of holographic transformation, results in CDNs. 
Let $\X: = \{1, \ldots, |\X|\}$ and let bivariate function $A_{\X}$ on $\X \times \X$ be such that
$A_{\X}(x,x') = 1$ if $x' \leq x$ and $A_{\X}(x,x') = 0$, otherwise. We call $A_{\X}$ a {\em cumulus} function.
Let bivariate function $D_{\X}$ on $\X \times \X$ be such that 
$D_{\X}(x,x') = 1$ if $x = x'$, $D_{\X}(x,x') = -1$ if $x = x'+1$, and $D_{\X}(x,x') = 0$ otherwise. We call 
$D_{\X}$ a {\em difference} function. 

In the case where $\X$ is the Cartesian product $\prod_{i \in I} \X_i$ where $\X_{i}:=\left\{ 1,\cdots,|\X_i| \right\}$ and $I$ is a finite indexing set, then
the previous definitions are extended to the partially-ordered set $\X$ in a component-wise manner by setting
$A_{\X}(x_I,x'_I):=\prod_{i \in I}A_{\X_{i}}(x_i,x'_i)$ and $D_{\X}(x_I,x'_I):=\prod_{i \in I}D_{\X_{i}}(x_i,x'_i)$ for all
$x_I,x'_I \in \X$.
In our notations for cumulus and difference function vertices, we distinguish the first argument using a dot to
mark the corresponding edge, cf. Fig.~\ref{fig:max-eq}~(a).

\begin{figure}
\centering
%\subfigure{\includegraphics[scale = .56]{fig/Ind_max-eq_fig.eps}}			\hspace{0cm}
%\subfigure{\includegraphics[scale = .56]{fig/Ind_Eq_fig.eps}} 			%\hspace{0cm}
\subfigure[]{
  \begin{tikzpicture}[v/.style={rectangle, draw, node distance=\vdist, minimum size=6mm}, d/.style={node distance=\ddist}, de/.style={node distance=\vdist}]
    \def\vdist{2cm} \def\ddist{1.2cm} 
    \node(max)[v]{$\max$}; 
    \node(A)[v, node distance=\vdist*3/4, above of=max]{$A$};
    \node(D1)[v, below left of=max]{$D$}; \node(Dn)[v, below right of=max]{$D$}; 
    \node(d0)[d, above of=A]{}; \node(d1)[d, below of=D1]{}; \node(dn)[d, below of=Dn]{};
    \path[every node/.style={}]
    (max)edge[->-out=.4](A)
    (max)edge[--o](D1)edge[--o](Dn) node[node distance=\vdist/2, below of=max]{$\dots$}
    (A)edge[o--] node[above left]{$x_1$}(d0) (D1)edge node[below right]{$x_{2}$}(d1) (Dn)edge node[below right]{$x_{n}$}(dn);
  \end{tikzpicture}  }
   %
%    \node(equal)[d, node distance=3cm, right of=max]{$=$};
    %
%    \node(eq)[v, node distance=3cm, right of=equal]{$=$}; 
%    \node(d1)[de, above of=eq]{}; \node(d2)[de, below left of=eq]{}; \node(dn)[de, below right of=eq]{}; 
%    \path[every node/.style={}]
%    (eq)edge(d1)edge(d2)edge(dn);
%    \tikz[pin distance=\vdist]
  \subfigure[]{
  \begin{tikzpicture}[v/.style={rectangle, draw, node distance=3cm, minimum size=6mm}]
    \def\vdist{2cm};
    \tikzstyle{every pin edge}=[] \tikzstyle{pe}=[pin distance=\vdist]
    \node[v, pin={[pe] above:$x_1$}, pin={[pe] below left:$x_2$}, pin={[pe] below right:$x_n$}]{$=$};
  \end{tikzpicture}  }
 \caption{Two equivalent NFGs}
\label{fig:max-eq}
\end{figure}
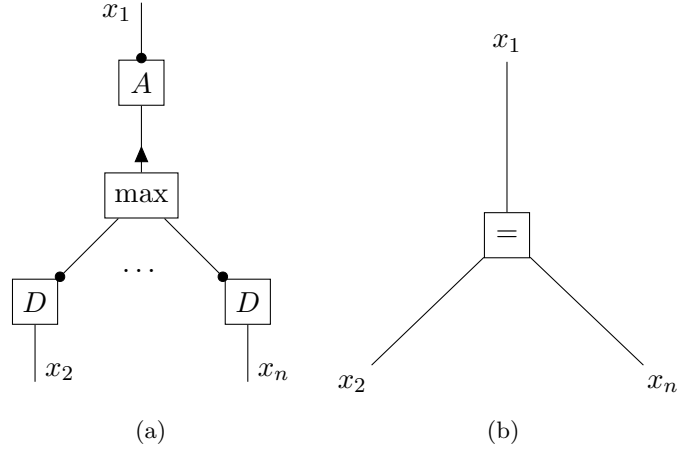
The following lemma will be useful in characterizing CDNs as a subclass of transformed NFG models.
\begin{Lemma} 
  Let $\X=\prod\limits_{i \in I} \X_i$ where $\X_i:=\left\{ 1, \cdots, |\X_i| \right\}$ for all $i$ in the finite indexing set $I$. Then, \\
1.  $\big\langle A_{\X}(x,y), D_{\X}(y,x') \big\rangle = \delta_{=}(x, x')$. \\ %=\big\langle A_{\X}(y, x), D_{\X}(x', y) \big\rangle$.  \\
%2. For any set of RVs $X_I$, $\langle p_{X_I}(x'_I), \prod_{i\in I} A_{\X_i}(x_i, x'_i) \rangle$ is $F_{X_I}(x_I). $\\
2. For any set of RVs $X_I$, $\big\langle p_{X_I}(x'), A_{\X}(x, x') \big\rangle$ is $F_{X_I}(x)$ for all $x \in \X$.\\
3. The two NFGs in Figure \ref{fig:max-eq} are equivalent where each edge variable assumes its values from $\X$.
\label{lem:cumulus}
\end{Lemma}
\begin{proof}
%1. We have $\langle A(x,y), D(y,x') \rangle = \sum_{y \in \X} A(x,y) D(y,x') = A(x,x') - A(x,x'+1) = \delta_{=}(x,x').$
  Parts 1 and 2 are immediate from the definitions of $A_{\X}$ and $D_{\X}$. For part 3,
%\begin{figure}[ht]
%\subfigure[]{\includegraphics[scale = .85]{fig/Ind_max-eq_n2_fig.eps}}
%\caption{}
%\end{figure}
let $\G$ be as in Fig.~\ref{fig:max-eq}~(a). 
First we prove the result for $n = 3$ and $|I|=1$, i.e., $\X = \{1, \cdots, |\X|\}$, where by the definitions of the difference transform, the
exterior function, and the max indicator function, we have
%$Z_{\G}$ as in the bottom of the page.
\begin{eqnarray*}
Z_{\G}(x_1,x_2,x_3) = \left \{
\begin{array}{ll}
A_{\X}(x_1,\max(x_2,x_3)) - A_{\X}(x_1,\max(x_2+1,x_3)) \ -  & \\
A_{\X}(x_1,\max(x_2,x_3+1)) + A_{\X}(x_1,\max(x_2 + 1,x_3 + 1)) &, x_2,x_3 < |\X|  \\
A_{\X}(x_1,\max(x_2,x_3)) - A_{\X}(x_1,\max(x_2+1,x_3)) &, x_2 < |\X|, x_3 = |\X| \\
A_{\X}(x_1,\max(x_2,x_3)) - A_{\X}(x_1,\max(x_2,x_3+1)) &, x_2 = |\X|, x_3 < |\X| \\
A_{\X}(x_1,\max(x_2,x_3)) &, x_2, x_3 = |\X|
\end{array}
\right .
\end{eqnarray*}
%%%%%%%%%%%%%%%%%%%%%%%%%%%%%%%%%%%%%%%%%%%%%%%%%
%%%==============================================
%% LONG EQUATION
%\begin{figure*}[!b]
%% ensure that we have normalsize text
%\normalsize
%\vspace*{4pt}
%\hrulefill
%\begin{eqnarray*}
%Z_{\G}(x_1,x_2,x_3) = \left \{
%\begin{array}{ll}
%A(x_1,\max(x_2,x_3)) - A(x_1,\max(x_2+1,x_3)) \ -  & \\
%A(x_1,\max(x_2,x_3+1)) + A(x_1,\max(x_2 + 1,x_3 + 1)) &, x_2,x_3 < |\X|  \\
%A(x_1,\max(x_2,x_3)) - A(x_1,\max(x_2+1,x_3)) &, x_2 < |\X|, x_3 = |\X| \\
%A(x_1,\max(x_2,x_3)) - A(x_1,\max(x_2,x_3+1)) &, x_2 = |\X|, x_3 < |\X| \\
%A(x_1,\max(x_2,x_3)) &, x_2, x_3 = |\X|
%\end{array}
%\right .
%\end{eqnarray*}
%% Restore the current equation number.
%%\setcounter{equation}{\value{MYtempeqncnt}}
%% IEEE uses as a separator
%% The spacer can be tweaked to stop underfull vboxes.
%\end{figure*}
%%%==============================================
%%%%%%%%%%%%%%%%%%%%%%%%%%%%%%%%%%%%%%%%%%%%%%%%
From the definition of the cumulus, it is clear that any possible non-zero values of $Z_{\G}$ may occur only if
$x_1 \geq \max(x_2,x_3)$. Assume $x_1 > \max(x_2,x_3)$ and note that in this case it is impossible to simultaneously have
$x_2 = |\X|$ and $x_3 = |\X|$, then
\begin{eqnarray*}
Z_{\G}(x_1,x_2,x_3) \hspace{-.5cm} &&= \left \{
\begin{array}{ll}
1-1-1+1 & , x_2,x_3 < |\X| \\
1 - 1 & , x_2 < |\X|, x_3 = |\X| \\
1 - 1 & , x_2 = |\X|, x_3 < |\X|
\end{array}
\right . \\
&&= 0.
\end{eqnarray*}
Hence, assume $x_1 = \max(x_2,x_3)$ and further suppose $x_2 < x_3$, then
\begin{eqnarray*}
Z_{\G}(x_1,x_2,x_3) \hspace{-.5cm} &&= \left \{
\begin{array}{ll}
1-1 & , x_2,x_3 < |\X| \\
1 - 1 & , x_2 < |\X|, x_3 = |\X|
\end{array}
\right . \\
&&= 0.
\end{eqnarray*}
By symmetry, we also have $Z_{\G}(x_1,x_2,x_3) = 0$ for $x_2 > x_3$. The only possibility left is $x_2 = x_3$, in which case, it is clear that
$Z_{\G}(x_1,x_2,x_3) = 1$. For $|I| > 1$, the claim follows from the fact that the cumulus, difference, $\max$ and equality functions are defined
in a component-wise manner.

\begin{figure}[!h]
\centering
%\subfigure[]{\includegraphics[scale = .75]{fig/Ind_max-eq_n1_fig.eps}} \hspace{.2cm}
%\subfigure[]{\includegraphics[scale = .75]{fig/Ind_max-eq_n2_fig.eps}}
\subfigure[]{
  \begin{tikzpicture}[scale=.8, v/.style={node distance=\vdist, draw, rectangle, inner sep=1mm}, vy/.style={node distance=\vdist/2, draw, rectangle, inner sep=1mm},
    every node/.append style={transform shape}]
    \tikzstyle{every pin edge}=[]; \tikzstyle{pe}=[pin distance=\vdist/3];
    \def\vdist{2cm};
    \node(v1)[v]{$\max$}; \node(v2)[v, right of=v1]{$\max$}; %\node(v3)[v, right of=v2]{$\max$};
    	\node(vdots)[v, draw=none, right of=v2]{$\cdots$}; \node(vL)[v, right of=vdots]{$\max$};
    \node(A)[vy, above of=v1]{$A$}; \node(d1)[v, draw=none, node distance=\vdist/2, above of=A]{$x_1$};
    \node(D1)[vy, pin={[pe]below:$x_2$}, below of=v1]{$D$}; \node(D2)[vy, pin={[pe]below:$x_3$}, below of=v2]{$D$}; %\node(D3)[vy, pin={[pe]below:$x_4$}, below of=v3]{$D$}; 
    \node(DL)[vy, pin={[pe]below:$x_{n-1}$}, below of=vL]{$D$};
    \node(Dh)[vy, pin={[pe]right:$x_n$}, right of=vL]{$D$};
    \path[every node/.style={transform shape}]
    (A)edge[o--](d1)
    (v1)edge[->-out=.6](A)edge[--o](D1)
    (v2)edge[->-out=.3](v1)edge[--o](D2) 
%    (v3)edge[->-out=.3](v2)edge[--o](D3) 
    (vdots)edge(v2)  (vL)edge[->-out=.3](vdots)  
    (vL)edge[--o](DL)edge[--o](Dh);
    \draw[dashed] (v1) \clbox{-.5}{-.5}{7.5}{.5};
  \end{tikzpicture}
}
\subfigure[]{
  \begin{tikzpicture}[scale=.8, v/.style={node distance=\vdistx, draw, rectangle, inner sep=1mm}, vy/.style={node distance=\vdist/2, draw, rectangle, inner sep=1mm},
    every node/.append style={transform shape}]
    \tikzstyle{every pin edge}=[]; \tikzstyle{pe}=[pin distance=\vdist/3];
    \def\vdist{2cm}; \def\vdistx{1.2cm};
    \node(v1)[v]{$\max$}; \node(D1x)[v, right of=v1]{$D$}; \node(A1x)[v, node distance=.8cm, right of=D1x]{$A$};
    \node(v2)[v, right of=A1x]{$\max$}; \node(D2x)[v, right of=v2]{$D$}; %\node(A2x)[v, right of=D2x]{$A$};
    \node(vdots)[v, node distance=.8cm, draw=none, right of=D2x]{$\cdots$};
    \node(ALx)[v, node distance=.8cm, right of=vdots]{$A$}; \node(vL)[v, right of=ALx]{$\max$}; \node(DLx)[v, pin={[pe]right:$x_{n}$}, right of=vL]{$D$};
    \node(A)[vy, above of=v1]{$A$}; \node(d1)[v, draw=none, node distance=\vdist/2, above of=A]{$x_1$};
    \node(D1)[vy, pin={[pe]below:$x_2$}, below of=v1]{$D$}; \node(D2)[vy, pin={[pe]below:$x_3$}, below of=v2]{$D$}; 
    \node(DL)[vy, pin={[pe]below:$x_{n-1}$}, below of=vL]{$D$};
    \path[every node/.style={transform shape}]
    (A)edge[o--](d1)
    (v1)edge[->-out=.6](A)edge[--o](D1)edge[--o](D1x)
    (D1x)edge[--o](A1x)
    (v2)edge[->-out=.7](A1x)edge[--o](D2)edge[--o](D2x) 
    (vL)edge[->-out=.7](ALx)edge[--o](DL)edge[--o](DLx) 
%    (v3)edge[->-out=.3](v2)edge[--o](D3) 
    (vdots)edge(D2x)edge[--o](ALx);
    \draw[dashed] (v1) \clbox{-.55}{-1.5}{1.5}{1.5};
    \draw[dashed] (v2) \clbox{-1.55}{-1.5}{1.55}{.5};
    \draw[dashed] (vL) \clbox{-1.55}{-1.5}{1.55}{.5};
  \end{tikzpicture}
}
\caption{Proof of Part 3 of Lemma~\ref{lem:cumulus} for arbitrary $n$.}
\label{fig:max-eq-n}
\end{figure}
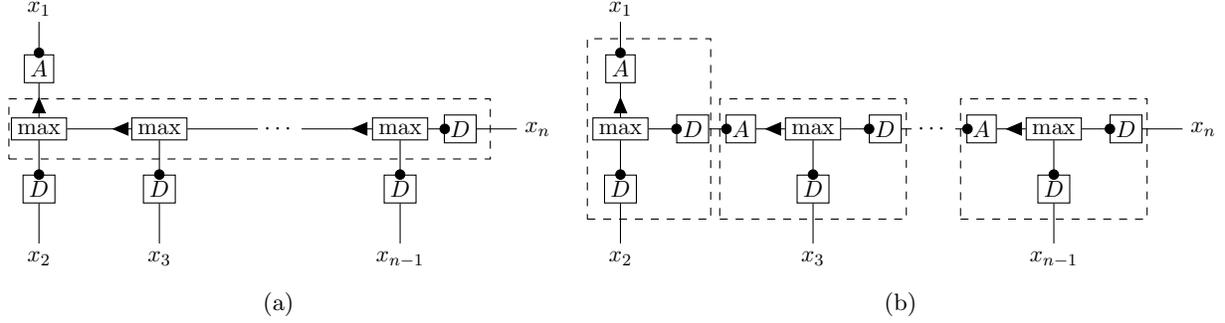
For general $n$, using the fact that $\max(x_2,x_3,\ldots,x_n) = \max(x_2,\max(x_3,\ldots \max(x_n-1,x_n)\ldots ))$,
we can express the NFG in Fig.~\ref{fig:max-eq}~(a) using the NFG in Fig.~\ref{fig:max-eq-n}~(a).
Inserting the inverse-pair $A_{\X}$ and $D_{\X}$ on each edge inside the dashed box in Fig.~\ref{fig:max-eq-n}~(a), we
obtain the equivalent NFG in (b). Invoking the established part of the lemma for $n=3$
under the vertex merging shown in (b), and the claim follows.
\end{proof}
In this lemma, Part 1 suggests that  $A_{\X}$ and $D_{\X}$ are an inverse-pair transformers, and  Part 2 suggests that cumulus functions may serve to transform a probability distribution to a CDF.

%\begin{Theorem}
%Let a generative NFG model be such that every interface function is a max indicator, and every latent function is a probability distribution of the
%involved interval variables.  Apply a holographic transformation where in Step H1 each $g_i$ is chosen as $A_{\X_i}$, and in Step H2, the inverse-pair
%transformers are chosen as the pair of culumus and difference functions defined on the corresponding edge variable alphabet (the cumulus function
%faces the latent function). Then the resulting NFG from this procedure can be converted to an equivalent CDN.
%\end{Theorem}

Given a generative NFG in which each interface function is a max indicator and each hidden function is a probability distribution of the
variables incident on it. Let $\G$ be the transformed model obtained from such generative NFG by inserting the cumulus transformer $A_{\X_i}$
on each half-edge. The following procedure converts $\G$ into a CDN:
\begin{itemize}
  \item[1)] Replace each max indicator and its adjacent transformer with a variable node representing the transformer's half-edge
    variable, and delete the half edge.
  \item[2)] Replace each hidden function, which is a probability distribution, with the corresponding cumulative distribution.
\end{itemize}
\begin{Theorem}
  If in a transformed model all external transformers are cumulus transformers, all interface functions are max indicators, and all
  hidden functions are probability distributions, then the above procedure gives rise to a CDN equivalent to the transformed model.
  \label{theorem:NFG-CDN}
\end{Theorem}
\begin{proof}
%The fact that the resulting FG is a CDN is obvious since each local function is a CDF.  To prove equivalence, 
  Perform a holographic transformation on the transformed model by inserting into each internal edge $e$ the inverse-pair
transformers $A_{\X_{e}}$ and $D_{\X_{e}}$, with the cumulus facing the hidden function and the difference transformer facing the max indicator.
Merging each hidden node with its neighboring cumulus transformers, by Part~(2) of Lemma~\ref{lem:cumulus}, the resulting node represents the desired CDF.
Merging each max indicator with its neighboring difference transformers and the already existing cumulus, by Part~(3) of Lemma~\ref{lem:cumulus}, 
we arrive to a constrained NFG in which each interface function is an equality indicator, and the claim follows by Proposition~\ref{prop:mul-NFG}.
\end{proof}

Figure \ref{fig:proof_CDN_NFG} demonstrates the proof on an example NFG.
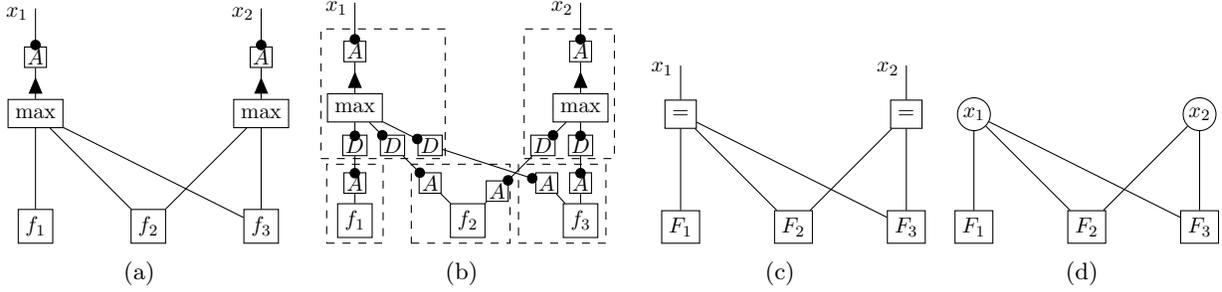
\begin{figure}[ht!]
\centering
%\begin{tabular}{cccc}
%\scalebox{0.55}{\includegraphics{fig/NFG_CDN_max_fig.eps}} & 
%\scalebox{0.55}{\includegraphics{fig/NFG_CDN_max1_fig.eps}} &
%\scalebox{0.55}{\includegraphics{fig/NFG_CDN_max2_2_fig.eps}} & 
%\scalebox{0.55}{\includegraphics{fig/CDN_fig.eps}} \\
%(a) & (b) & (c) & (d) 
%\end{tabular}
%\begin{tabular}{cc}
\subfigure[]{
  \begin{tikzpicture}[scale=.75, v/.style={node distance=\dist, draw, rectangle, minimum size=5mm},
      t/.style={node distance=\dist/3, draw, rectangle, inner sep=.5mm},
      d/.style={node distance=\dist/2}, every node/.append style={transform shape}]
      \def\dist{2cm};
      \node(u1)[v]{$\max$}; \node(dump)[v, draw=none, right of=u1]{}; \node(u2)[v, right of=dump]{$\max$}; 
      \node(v1)[v, below of=u1]{$f_1$}; \node(v2)[v, right of=v1]{$f_2$}; \node(v3)[v, right of=v2]{$f_3$}; 
      \node(A1)[t, node distance=\dist/2, above of=u1]{$A$}; \node(A2)[t, node distance=\dist/2, above of=u2]{$A$}; 
      \node(d1)[d, above of=A1]{}; \node(d2)[d, above of=A2]{}; 
      \path[every node/.style={transform shape}]
      (u1)edge(v1)edge(v2)edge(v3)	(u2)edge(v2)edge(v3)
      (u1)edge[->-out=.7](A1)	(u2)edge[->-out=.7](A2)
      (A1)edge[o--] node[above left]{$x_1$}(d1) (A2)edge[o--] node[above left]{$x_2$}(d2);
    \end{tikzpicture} }
\subfigure[]{
  \begin{tikzpicture}[scale=.75, v/.style={node distance=\dist, draw, rectangle, minimum size=5mm},
      t/.style={node distance=\dist/3, draw, rectangle, inner sep=.5mm},
      d/.style={node distance=\dist/2}, every node/.append style={transform shape}]
      \def\dist{2cm};
      \node(u1)[v]{$\max$}; \node(dump)[v, draw=none, right of=u1]{}; \node(u2)[v, right of=dump]{$\max$}; 
      \node(v1)[v, below of=u1]{$f_1$}; \node(v2)[v, right of=v1]{$f_2$}; \node(v3)[v, right of=v2]{$f_3$}; 
      \node(D11) at (0,-\dist*1/3)[t]{$D$}; \node(D12) at (\dist*1/3,-\dist*1/3)[t]{$D$}; \node(D13) at (\dist*2/3,-\dist*1/3)[t]{$D$}; 
      \node(D22) at (\dist*5/3,-\dist*1/3)[t]{$D$}; \node(D23) at (\dist*2,-\dist*1/3)[t]{$D$}; 
      \node(A11) at (0,-\dist*2/3)[t]{$A$}; \node(A12) at (\dist*2/3,-\dist*2/3)[t]{$A$}; \node(A13) at (\dist*1.7,-\dist*2/3)[t]{$A$}; 
      \node(A22)[t] at ([shift={(45:\dist*.37)}]v2){$A$}; \node(A23) at (\dist*2,-\dist*2/3)[t]{$A$};
      \node(A1)[t, node distance=\dist/2, above of=u1]{$A$}; \node(A2)[t, node distance=\dist/2, above of=u2]{$A$}; 
      \node(d1)[d, above of=A1]{}; \node(d2)[d, above of=A2]{}; 
      \path[every node/.style={transform shape}]
      (u1)edge[--o](D11)edge[--o](D12)edge[--o](D13)	
      (D11)edge[--o](A11) (D12)edge[--o](A12) (D13)edge[--o](A13)
      (A11)edge(v1) (A12)edge(v2) (A13)edge(v3)
      (u2)edge[--o](D22)edge[--o](D23)	
      (D22)edge[--o](A22) (D23)edge[--o](A23)
      (A22)edge(v2) (A23)edge(v3)
      (u1)edge[->-out=.7](A1)	(u2)edge[->-out=.7](A2)
      (A1)edge[o--] node[above left]{$x_1$}(d1) (A2)edge[o--] node[above left]{$x_2$}(d2);
      \draw[dashed] (v1) \clbox{-.5}{-.4}{.5}{1};
      \draw[dashed] (v2) \clbox{-1}{-.4}{.75}{1};
      \draw[dashed] (v3) \clbox{-1.1}{-.4}{.45}{1};
      \draw[dashed] (u1) \clbox{-.6}{-.9}{1.6}{1.4};
      \draw[dashed] (u2) \clbox{-1}{-.9}{.6}{1.4};
    \end{tikzpicture} } %\\
    \subfigure[]{
    \begin{tikzpicture}[scale=.75, v/.style={node distance=\dist, draw, rectangle, minimum size=5mm},
      t/.style={node distance=\dist/3, draw, rectangle, inner sep=.5mm},
      d/.style={node distance=\dist/2}, every node/.append style={transform shape}]
      \def\dist{2cm};
      \node(u1)[v]{$=$}; \node(dump)[v, draw=none, right of=u1]{}; \node(u2)[v, right of=dump]{$=$}; 
      \node(v1)[v, below of=u1]{$F_1$}; \node(v2)[v, right of=v1]{$F_2$}; \node(v3)[v, right of=v2]{$F_3$}; 
%      \node(A1)[t, node distance=\dist/2, above of=u1]{$A$}; \node(A2)[t, node distance=\dist/2, above of=u2]{$A$}; 
      \node(d1)[d, above of=u1]{}; \node(d2)[d, above of=u2]{}; 
      \path[every node/.style={transform shape}]
      (u1)edge(v1)edge(v2)edge(v3)	(u2)edge(v2)edge(v3)
      (u1)edge node[above left]{$x_1$}(d1)	(u2)edge node[above left]{$x_2$}(d2);
%      (A1)edge[o--] (d1) (A2)edge[o--] node[above left]{$x_2$}(d2);
    \end{tikzpicture} }
    \subfigure[]{
    \begin{tikzpicture}[scale=.75, var/.style={node distance=2cm, draw, circle, inner sep=.5mm, minimum size=4mm}, 
      fun/.style={node distance=2cm, draw, rectangle, minimum size=5mm}, every node/.append style={transform shape}]
      \node(x1)[var]{$x_1$}; \node(dump)[var, draw=none, right of=x1]{}; \node(x2)[var, right of=dump]{$x_2$};
      \node(f1)[fun, below of=x1]{$F_1$}; \node(f2)[fun, right of=f1]{$F_2$}; \node(f3)[fun, right of=f2]{$F_3$};
      \path[every node/.style={transform shape}]
      (x1)edge(f1)edge(f2)edge(f3)
%      (x2)edge(f1)
      (x2)edge(f2)edge(f3);
    \end{tikzpicture} }
%  \end{tabular}
%\caption{ Panel (a) with the cumulus functions removed is the original NFG model; (a), (b)  and (c) show the transformation steps
%  (1), (2) and (3) respectively. By the generalize Holant theorem, the three NFGs realize the same exterior function. By Part (2)
%  of Lemma \ref{lem:cumulus}, these exterior functions are the joint distribution of the external variables. By Part (3) of Lemma
%  \ref{lem:cumulus}, each top dashed box merges to an equality indicator, again by Part (2) of Lemma \ref{lem:cumulus}, each
%  bottom dashed box merges to a CDF. By conversion procedure of Proposition \ref{prop:mul-NFG}, (c) is equivalent to the FG in
%  (d), which is a CDN.  }
    \caption{(a) An example of a transformed model in accordance to Theorem~\ref{theorem:NFG-CDN}.
    (b) Inserting the inverse-pair cumulus-difference transformers on each internal edge, and vertex merging each
    node with its neighboring transformers, by the GHT, the resulting NFG is equivalent to the one in (a).
    (c) By Lemma~\ref{lem:cumulus} an equivalent NFG to the one in (b), where $F_i$ is the cumulus transform of $f_i$, which is a
    CDF if $f_i$ is a probability distribution.
    (d) By Proposition~\ref{prop:mul-NFG}, an equivalent CDN to the NFG in (c).}
\label{fig:proof_CDN_NFG}
\end{figure}

Before we proceed, we remark that the independence properties implied by a transformed model are precisely the ones implied by its
base NFG, i.e., by ignoring all the external transformers. By the remark succeeding Theorem~\ref{thm:indep}, such independence
properties are further invariant under internal holographic transformations. Since the base NFG of the transformed model in
Theorem~\ref{theorem:NFG-CDN} is a generative model, it becomes clear that the independence properties implied by CDNs are the
marginal independence ones, i.e., Part~2 of Theorem~\ref{thm:indep}.

Besides CDNs, we also remark that it is possible to regard linear characteristic models (LCMs) \cite{Bickson} as a special case of
transformed NFG models. 
In fact, the NFG in Fig.~\ref{fig:conv-NFG}~(d) is a transformed model that is equivalent to an LCM, and the holographic transformation presented in
the proof of Proposition~\ref{prop:conv-NFG}, Fig.~\ref{fig:conv-NFG}~(c), provides the basis for understanding an LCM as a
transformed model whose base NFG is a generative NFG equivalent to a CFG,
Fig.~\ref{fig:conv-NFG}~(a). We skip details, and hope the framework of transformed NFG models is clear enough to see such equivalence.

%% file: coding.tex
\section{Linear codes}
\label{sec:coding}

In this section we discuss the connection between NFG models and linear codes.
%discuss a yet more prevailing duality of such models. 
The material in this section is well-known \cite{Forney:course-notes}, yet, we choose to address it in this work in light of
the constrained and generative semantics.

Any code $\cal C$, linear or not linear, can be described in terms of its \emph{membership} function, namely, the indicator function
$\delta_{\cal C}(y):=[y \in {\cal C}]$ for all $y \in \X$ and for some finite set $\X$. Clearly $\delta_{\cal C}$ may be viewed, up to
scaling factor, as a distribution function over $\X$ and in the case of linear codes, as we will see, may be described in terms of
a constrained or a generative NFG model.

A \emph{linear code} of length $n$ and dimension $k$ over a finite field $\F$ is a $k$ dimensional subvector space
of $\F^n$.
Classically, a linear code $\cal C$ can be expressed in two dual ways. At one hand,
${\cal C} = \{ f(x) : x \in \F^{k} \}$ for some linear function $f:\F^k \rightarrow \F^n$. This can also
be written as ${\cal C} = \{\big(f_1(x), \ldots, f_n(x) \big) : x \in \F^k \}$ for some linear maps $f_i : \F^k \rightarrow \F$ for all $i$.
That is, the code indicator function $\delta_{\cal C}(y)$ for all $y \in \F^{n}$, can be expressed as the sum-of-products form
\[
  \delta_{\cal C}(y_1, \ldots, y_n) = \sum\limits_{x \in \F^k} \prod\limits_{i=1}^{n} [y_i = f_i(x)],
\] for all $y_i \in \F$.
Clearly, each local function $[y_i = f_i(x)]$ is a conditional function of $y_i$ given $x$, and so,
the indicator function of $\cal C$ can be realized using a generative NFG model.
In fact, since each linear function $f_i:\F^k \rightarrow \F$ can be written as
$f_i(x_1, \ldots, x_k) = a_{i1} x_1 + \ldots + a_{ik} x_k$ for some $a_{ij} \in \F$, it follows
that each interface function is a sum indicator function involving $y_i$ and $x_j$ for all $j$ such that $a_{ij} \neq 0$.
%
%Moreover, if for each $i$ we define $J_i$ as the indexing set of the arguments of $f_i$ with nonzero coefficients, i.e.,
%$J_i := \{j \in \{1, \ldots, k\} : a_{ij} \neq 0\}$, then each local interface function involves only $y_i$ and $x_{J_i}$. 
%
The role of the hidden functions is to provide replicas of each variable $x_{j}$ for all $j \in \{1, \cdots, k\}$, and hence are taken as equality indicators.
More explicitly, for each $x_j$, we have a hidden equality indicator of degree equals the number of interface functions
involving $x_j$.
This guarantees that
each variable appears in the desired number of interface functions while respecting the degree restrictions \cite{Forney2001:Normal}.
From this we arrive to a generative NFG with $n$ interface nodes (each is a sum indicator), $k$ hidden nodes (each is an equality indicator),
%and a set of internal edges given by $\big\{ \left\{ i,j \right\}: j \in J_{i}, i\in \left\{ 1,\cdots,n \right\} \big\}$.
and there is an edge connecting nodes $i$ and $j$ if and only if $a_{ij}$ is nonzero.
Fig.~\ref{fig:gen-const}~(a) illustrates the generative NFG model for the Hamming code (with $n=7$ and $k=4$).

On the other hand, the parity check interpretation of a linear code dictates that the elements of
a linear code $\cal C$ must satisfy a collection of homogeneous linear equations. That is,
${\cal C} = \{y \in \F^{n} : f(y) = 0\}$ for some linear map $f: \F^{n} \rightarrow \F^{n-k}$.
This can also be written as ${\cal C} = \{y \in \F^{n} : (f_1(y), \ldots, f_{n-k}(y)) = 0\}$ for some linear maps $f_j: \F^{n} \rightarrow \F$.
Hence, the code indicator function can be expressed as the product
\[
\delta_{\cal C}(y) = \prod_{j} [f_j(y) = 0],
\]
which (since $f_j$ is linear) can further be simplified as $\delta_{\cal C}(y_1, \ldots, y_n) = \prod_{j} [\sum_{i} a_{ij} y_i = 0]$ for some $a_{ij} \in \F$.
From Proposition~\ref{prop:mul-NFG}, we can see that the code indicator function is realized by a constrained NFG model in which each interface node $i$
is an equality indicator, each hidden node $j$ is a parity indicator, and there is an edge connecting nodes $i$ and $j$ if and only if
$a_{ij}$ is nonzero, Fig.~\ref{fig:const-gen}~(a).

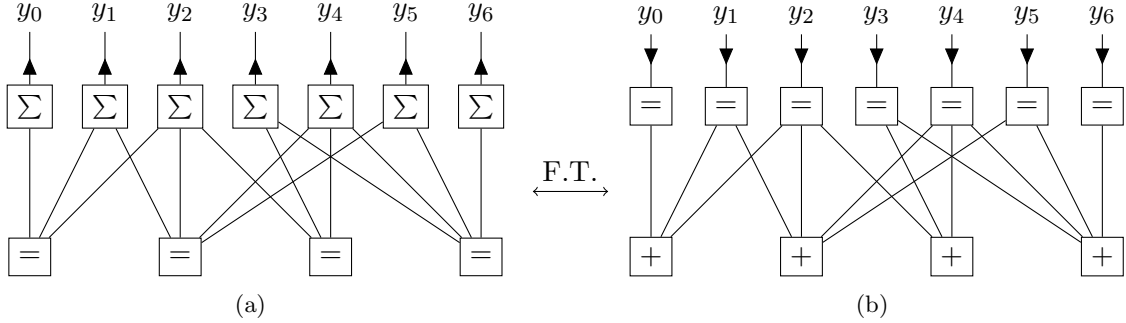
\begin{figure}[ht]
%\footnotesize
\center \setlength{\unitlength}{.5cm}
%    \subfigure[]{\includegraphics[scale = .7]{fig/Hamming_Gen_fig.eps}}	%\hspace{.5cm}
\subfigure[]{
  \begin{tikzpicture}[scale=1, u/.style={node distance=\xdist, draw, rectangle, minimum size=5mm}, 
      v/.style={node distance=\ydist, draw, rectangle, minimum size=5mm}, every node/.append style={transform shape}]
      \def\xdist{1cm}; \def\ydist{2cm};
      \tikzstyle{every pin edge}=[ ->-out=.5]; \tikzstyle{pe}=[pin distance=\xdist*.7];
      \node(u1)[u, pin={[pe]above:$y_0$}]{\footnotesize{$\sum$}}; \node(u2)[u, right of=u1, pin={[pe]above:$y_1$}]{\footnotesize{$\sum$}}; 
      \node(u3)[u, right of=u2, pin={[pe]above:$y_2$}]{\footnotesize{$\sum$}}; \node(u4)[u, right of=u3, pin={[pe]above:$y_3$}]{\footnotesize{$\sum$}};
      \node(u5)[u, right of=u4, pin={[pe]above:$y_4$}]{\footnotesize{$\sum$}}; \node(u6)[u, right of=u5, pin={[pe]above:$y_5$}]{\footnotesize{$\sum$}};
      \node(u7)[u, right of=u6, pin={[pe]above:$y_6$}]{\footnotesize{$\sum$}};
     \node(v1)[v, below of=u1]{=}; \node(v2)[v, below of=u3]{=}; \node(v3)[v, below of=u5]{=}; \node(v4)[v, below of=u7]{=};
      \path[every node={transform shape}]
      (v1)edge(u1)edge(u2)edge(u3)
      (v2)edge(u2)edge(u3)edge(u5)edge(u6)
      (v3)edge(u3)edge(u4)edge(u5)
      (v4)edge(u4)edge(u5)edge(u6)edge(u7);
    \end{tikzpicture} }  
    %
%    \hspace{-.3cm}
%    \begin{tikzpicture} 
%      \draw[<-] (0,0)--(1,0);	\node[above] at (1.5,0) {Hamming};	\draw[->] (2,0)--(3,0);	\node[below] at (0,0) {generator}; \node[below] at (3,0){parity};
%    \end{tikzpicture}	
%    \hspace{-.3cm}
    %
    \begin{tikzpicture} \draw[<->] (0,1cm)--node[above]{F.T.}(1cm,1cm) ; \node(dump) at (0,0){};	\end{tikzpicture} 
  \subfigure[]{
  \begin{tikzpicture}[scale=1, u/.style={node distance=\xdist, draw, rectangle, minimum size=5mm}, 
      v/.style={node distance=\ydist, draw, rectangle, minimum size=5mm}, every node/.append style={transform shape}]
      \def\xdist{1cm}; \def\ydist{2cm};
      \tikzstyle{every pin edge}=[ ->-in=.5]; \tikzstyle{pe}=[pin distance=\xdist*.7];
      \node(u1)[u, pin={[pe]above:$y_0$}]{$=$}; \node(u2)[u, right of=u1, pin={[pe]above:$y_1$}]{$=$}; 
      \node(u3)[u, right of=u2, pin={[pe]above:$y_2$}]{$=$}; \node(u4)[u, right of=u3, pin={[pe]above:$y_3$}]{$=$};
      \node(u5)[u, right of=u4, pin={[pe]above:$y_4$}]{$=$}; \node(u6)[u, right of=u5, pin={[pe]above:$y_5$}]{$=$};
      \node(u7)[u, right of=u6, pin={[pe]above:$y_6$}]{$=$};
      \node(v1)[v, below of=u1]{+}; \node(v2)[v, below of=u3]{$+$}; \node(v3)[v, below of=u5]{+}; \node(v4)[v, below of=u7]{+};
      \path[every node={transform shape}]
      (v1)edge(u1)edge(u2)edge(u3)
      (v2)edge(u2)edge(u3)edge(u5)edge(u6)
      (v3)edge(u3)edge(u4)edge(u5)
      (v4)edge(u4)edge(u5)edge(u6)edge(u7);
    \end{tikzpicture}   }
    \caption{(a) Generator realization of Hamming code, and (b) parity realization of dual Hamming code.}
\label{fig:gen-const}
\end{figure}
\begin{figure}[ht]
  \centering
    \subfigure[]{
  \begin{tikzpicture}[scale=1, u/.style={node distance=\xdist, draw, rectangle, minimum size=5mm}, 
      v/.style={node distance=\ydist, draw, rectangle, minimum size=5mm}, every node/.append style={transform shape}]
      \def\xdist{1cm}; \def\ydist{2cm};
      \tikzstyle{every pin edge}=[ ->-in=.5]; \tikzstyle{pe}=[pin distance=\xdist*.7];
      \node(u1)[u, pin={[pe]above:$y_0$}]{$=$}; \node(u2)[u, right of=u1, pin={[pe]above:$y_1$}]{$=$}; 
      \node(u3)[u, right of=u2, pin={[pe]above:$y_2$}]{$=$}; \node(u4)[u, right of=u3, pin={[pe]above:$y_3$}]{$=$};
      \node(u5)[u, right of=u4, pin={[pe]above:$y_4$}]{$=$}; \node(u6)[u, right of=u5, pin={[pe]above:$y_5$}]{$=$};
      \node(u7)[u, right of=u6, pin={[pe]above:$y_6$}]{$=$};
      \node(v1)[v, below of=u2]{+}; \node(v2)[v, below of=u4]{$+$}; \node(v3)[v, below of=u6]{+}; 
      \path[every node={transform shape}]
      (v1)edge(u1)edge(u2)edge(u4)edge(u5)
      (v2)edge(u1)edge(u3)edge(u4)edge(u6)
      (v3)edge(u2)edge(u3)edge(u4)edge(u7);
    \end{tikzpicture}   }
    %
%    \hspace{-.2cm}
%    \begin{tikzpicture}
%      \def\ydist{2cm};
%      \draw[<-] (0,0)--(1,0);	\node[below] at(1.5,0){dual Hamming};	\draw[->] (2,0)--(3,0);	\node[above] at (0,0) {parity}; \node[above] at (3,0){generator};
%	\node(dump) at (0, -\ydist){};
%    \end{tikzpicture}					
%    \hspace{-.2cm}
    %
  \begin{tikzpicture} \draw[<->] (0,1cm)--node[above]{F.T.}(1cm,1cm) ; \node(dump) at (0,0){};	\end{tikzpicture} 
  \subfigure[]{
  \begin{tikzpicture}[scale=1, u/.style={node distance=\xdist, draw, rectangle, minimum size=5mm}, 
      v/.style={node distance=\ydist, draw, rectangle, minimum size=5mm}, every node/.append style={transform shape}]
      \def\xdist{1cm}; \def\ydist{2cm};
      \tikzstyle{every pin edge}=[ ->-out=.5]; \tikzstyle{pe}=[pin distance=\xdist*.7];
      \node(u1)[u, pin={[pe]above:$y_0$}]{\footnotesize{$\sum$}}; \node(u2)[u, right of=u1, pin={[pe]above:$y_1$}]{\footnotesize{$\sum$}}; 
      \node(u3)[u, right of=u2, pin={[pe]above:$y_2$}]{\footnotesize{$\sum$}}; \node(u4)[u, right of=u3, pin={[pe]above:$y_3$}]{\footnotesize{$\sum$}};
      \node(u5)[u, right of=u4, pin={[pe]above:$y_4$}]{\footnotesize{$\sum$}}; \node(u6)[u, right of=u5, pin={[pe]above:$y_5$}]{\footnotesize{$\sum$}};
      \node(u7)[u, right of=u6, pin={[pe]above:$y_6$}]{\footnotesize{$\sum$}};
%      \node(u1)[u, pin={[pe]above:$y_0$}]{$+$}; \node(u2)[u, right of=u1, pin={[pe]above:$y_1$}]{$+$}; 
%      \node(u3)[u, right of=u2, pin={[pe]above:$y_2$}]{$+$}; \node(u4)[u, right of=u3, pin={[pe]above:$y_3$}]{$+$};
%      \node(u5)[u, right of=u4, pin={[pe]above:$y_4$}]{$+$}; \node(u6)[u, right of=u5, pin={[pe]above:$y_5$}]{$+$};
%      \node(u7)[u, right of=u6, pin={[pe]above:$y_6$}]{$+$};
      \node(v1)[v, below of=u2]{=}; \node(v2)[v, below of=u4]{$=$}; \node(v3)[v, below of=u6]{=}; 
      \path[every node={transform shape}]
      (v1)edge(u1)edge(u2)edge(u4)edge(u5)
      (v2)edge(u1)edge(u3)edge(u4)edge(u6)
      (v3)edge(u2)edge(u3)edge(u4)edge(u7);
    \end{tikzpicture}   }
    \caption{(a) Parity realization of Hamming code, and (b) generator realization of dual Hamming code.}
\label{fig:const-gen}
\end{figure}
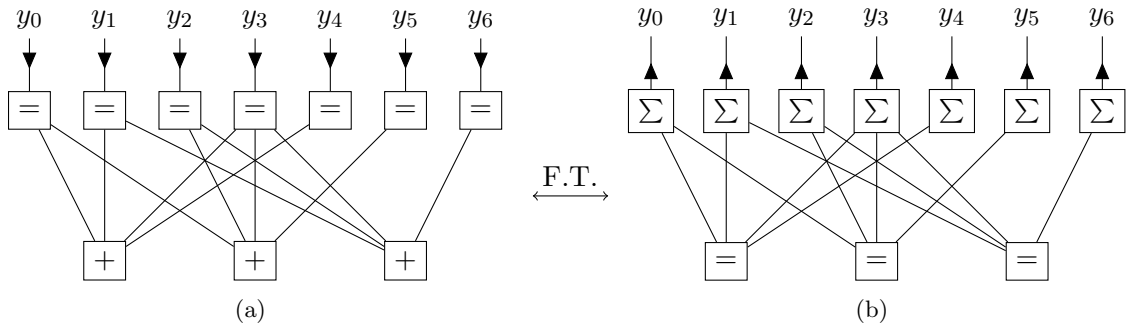

In summary, a constrained NFG model of a linear code represents a parity realization, and a
generative NFG model represents a generator realization. From the duality between the sum indicator
and the equality indicator under the Fourier transform, the duality between a generator
realization of a code and a parity realization of the dual code may be explained as follows: Starting with a generative NFG model
(generator realization) of a linear code, Fig.~\ref{fig:gen-const}~(a), and performing a holographic transformation with Fourier transformers, one obtains a
constrained NFG model (a parity realization) of the dual code, Fig.~\ref{fig:gen-const}~(b). Conversely, starting with a constrained NFG model (a parity realization) of
the code, Fig.~\ref{fig:const-gen}~(a), one ends with a generative NFG model (a generator realization) of the dual code,
Fig.~\ref{fig:const-gen}~(b).

%% file: infer.tex
\section{Evaluation of the exterior function}

In this section we discuss the algorithmic aspect of evaluating the exterior function of NFGs.
We start with the ``elimination algorithm'' for NFGs, which is essentially the well-known elimination algorithm of inference on undirected graphical models
\cite{JOrdan:GMs}. The elimination algorithm is exact but its complexity depends on the ordering at which the elimination is performed.\footnote{The
problem of finding the ``best'' elimination ordering is known to be NP-hard, where the term ``best'' is in the sense of minimizing the largest
node-degree (of nodes that do not factor multiplicatively) arising while performing the elimination algorithm, and the minimization is over all possible orderings.}
A more efficient algorithm, but exact only on NFGs with no cycles, is the sum-product algorithm \cite{frank:factor}, which we
discuss in the language of NFGs \cite{Forney2011:PartitionFunction}
in the second part of this section. 
%Finally, we comment on more efficient computation of the exterior function using an ``indirect approach,'' where a holographic transformation, if it exists, is used to convert the NFG into one that is more appropriate for computations.
Finally, we discuss an ``indirect approach'' for evaluating the exterior function, where a holographic transformation, if it
exists, is used to convert the NFG into one that is more appropriate for the such computation. 
\subsection{Elimination Algorithm}
%Given an NFG, not necessarily a bipartite, the following algorithm computes its external function.
The following algorithm computes the exterior function of an arbitrary NFG (i.e. not necessarily a bipartite).
%\begin{figure}[ht]
\begin{algorithm}[Elimination] Given an NFG $\G$. \\ 
  \emph{\texttt{While $\G$ is not a single node. Do \\ 
    $\{ $\\
     Pick an adjacent pair of vertices $v_1$ and $v_2$ in $\G$; \\
%     $\G$ = contract$(\G,v_1,v_2)$; \\
     Compute $f_{v_1v_2}(x_{E(v_1) \cup E(v_2) \backslash E(v_1)\cap E(v_2)}) :=  \big\langle f_{v_1}, f_{v_2}  \big\rangle; $\\
     Update $\G$ by removing $f_{v_1}$ and $f_{v_2}$, and adding $f_{v_1v_2}$; \\
    $\} $
}}
\end{algorithm}
%\caption{Elimination algorithm}
%\end{figure}

%\begin{algorithm}[Contract] Given an NFG $\G(V,E,f_V)$. \\ 
%  \emph{\texttt{For each pair of adjacent vertices $v_1$ and $v_2$. Do \\ 
%    $\{ $\\
%      $f_{v_1v_2}(x_{E(v_1) \cup E(v_2) \backslash E(v_1)\cap E(v_2)}) =  \big\langle f_{v_1}, f_{v_2}  \big\rangle; $\\
%      Remove $f_{v_1}, f_{v_2}$ from the NFG, and add $f_{v_1v_2}$ to the NFG; \\
%    $\} $
%}}
%\end{algorithm}
%Note that $\big\langle f_{v}: v \in \{v_1,v_2\} \big\rangle$ is equal to $\sum \limits_{x_{e}} f_{v_1}(x_{E(v_1)})f_{v_2}(x_{E(v_2)})$
%if $e = \{v_1,v_2\}$ with $v_1 \neq v_2$, and is equal to $\sum \limits_{x_{e}} f_{v}(x_{E(v)})$ if $v_1 = v_2 = v$.
Evidently, this algorithm runs in a finite time %\footnote{All NFGs in this work are defined on finite graphs} 
and terminates with a single node whose function is the desired exterior function. 
This is essentially the vertex merging procedure applied recursively on pairs of adjacent vertices. Clearly, the
elimination algorithm may equivalently be viewed as an elimination algorithm on the edges of the NFG, where
in each step it eliminates all the edges between a pair of adjacent vertices. 
More precisely, one may say, the elimination algorithm is a merging algorithm on the nodes, and is an elimination
algorithm on the edges of the NFG.
In subsequent discussions, we will 
freely alternate between such two views. Note that even if we start with a simple NFG,
parallel edges may still arise\footnote{It is not hard to see that parallel edges do not appear at any step of the elimination algorithm if and only if the NFG is cycle-free, i.e., is a tree.} 
while applying the elimination algorithm. However,
loops may never arise since in each step we eliminate all parallel edges at once. 
%We remark that the complexity of the elimination algorithm depends on the ordering at which elimination is performed. (It is known that the problem of finding an optimum ordering is intractable.)

\begin{Example}
  Let $\G$ be as in Fig~\ref{fig:elimination}~(a).
  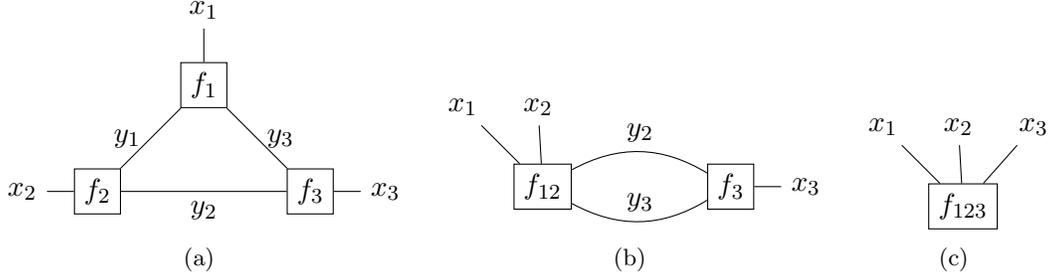
\begin{figure}[ht]
    \center
    \subfigure[]{
    \begin{tikzpicture}
      [node distance=2cm, every node/.style={rectangle, draw}]
      \node(p1){$f_1$};
      \node(p2)[below left of=p1]{$f_2$};
      \node(p3)[below right of=p1]{$f_3$};
      \node(d1)[draw=none, node distance=1cm, above of=p1]{$x_1$}; 
      \node(d2)[draw=none, node distance=1cm, left of=p2]{$x_2$}; 
      \node(d3)[draw=none, node distance=1cm, right of=p3]{$x_3$}; 
      \path[every node/.style={}]
      (p1)edge node[left]{$y_1$}(p2) (p2)edge node[below]{$y_2$}(p3) (p3)edge node[right]{$y_3$}(p1)
      (p1)edge(d1) (p2)edge(d2) (p3)edge(d3);
    \end{tikzpicture} }
    \subfigure[]{
     \begin{tikzpicture}
      [node distance=2.5cm, every node/.style={rectangle, draw}]
      \node(p1){$f_{12}$};
      \node(p2)[right of=p1]{$f_3$};
      \node(d1)[draw=none, node distance=1.5cm, above left of=p1]{$x_1$}; \node(d2)[draw=none, node distance=1cm, right of=d1]{$x_2$}; 
      \node(d3)[draw=none, node distance=1cm, right of=p2]{$x_3$}; 
      \path[every node/.style={}]
      (p1)edge[bend left] node[above]{$y_2$}(p2) (p1)edge[bend right] node[above]{$y_3$}(p2) 
      (p1)edge (d1) (p1)edge (d2) (p2)edge(d3); 
    \end{tikzpicture} }
%    \subfigure[]{
%    \begin{tikzpicture}[node distance=2.5cm, every node/.style={rectangle, draw}, every loop/.style={}]
%      \node(p1){$f_{e_2}$};
%      \node(d1)[draw=none, node distance=1.5cm, above of=p1]{$(x_1,x_2,x_3)$};
%      \path[every node/.style={}]
%      (p1)edge[loop right] node[right]{$y_3$}(p1) (p1)edge (d1);
%    \end{tikzpicture} }
    \subfigure[]{
     \begin{tikzpicture}[node distance=2.5cm, every node/.style={rectangle, draw}, every loop/.style={}]
      \node(p1){$f_{123}$};
      \node(d1)[draw=none, node distance=1.5cm, above left of=p1]{$x_1$}; \node(d2)[draw=none, node distance=1cm, right of=d1]{$x_2$};
      \node(d3)[draw=none, node distance=1cm, right of=d2]{$x_3$};
      \path[every node/.style={}]
      (p1)edge (d1) edge (d2) edge (d3);
    \end{tikzpicture} }
  \caption{Elimination algorithm example.}
    \label{fig:elimination}
  \end{figure}
  Applying the elimination algorithm, we obtain:\\
    Eliminating $y_{1}$ gives the NFG in Fig.~(b)  with $f_{12}(x_1,x_2,y_2,y_3) = \big\langle f_1(x_1,y_1,y_3), f_2(x_2,y_1,y_2) \big\rangle,$ \\
    eliminating $y_{2}, y_{3}$ gives the NFG in Fig.~(c)  with $f_{123}(x_1,x_2,x_3) = \big\langle f_{12}(x_1,x_2,y_2,y_3), f_3(x_3,y_2,y_3)\big\rangle,$ \\
    and it is easy to verify that $Z_{\G} = f_{123}$.
\end{Example}

For any node $v$ in an NFG, let $\deg(v):=|E(v)|$ be the number of external and internal edges incident on $v$. The following lemma
determines the complexity of eliminating a pair of adjacent vertices in an NFG.
\begin{Lemma}
  The complexity of eliminating a pair of adjacent vertices $u,v$ in the elimination algorithm is of order
  $|\X|^{\deg(u)+\deg(v)-|E(u) \cap E(v)|}$.
%  In the special case where $f_v = \delta_{\sum}$ or $\delta_{\max}$, then the complexity is $(\deg(v)-1)|\X|^{\deg(u)}$, and
%  when $f_{v} = \delta_{=}$ the complexity is $|\X|^{\deg(u)}$. \textbf{[Verify, these values are not precise!].}
  \label{lemma:elem-complexity}
\end{Lemma}
\begin{proof}
  For each $x \in \X_{E(u) \cup E(v) \backslash E(u) \cap E(v)}$, we need $|\X_{E(u) \cap E(v)}|$ computations,
  and the claim follows by noting that $|E(u)\cup E(v) \backslash E(u) \cap E(v)| = |E(u)|+|E(v)|-2|E(u)\cap E(v)|$.
\end{proof}

The following example shows that for some indicator functions of interest, the elimination complexity
can significantly be reduced.
\begin{Example}
  Consider the NFG in Fig.~\ref{fig:elim-complexity}~(a), where each edge is associated the same alphabet $\X$.
  In general, the complexity of computing the exterior function is of order
  $|\X|^{n+1}$. In the special case where $f$ is the indicator function $\delta_{\sum}$ or $\delta_{\max}$, then the complexity
  is $(n-1)|\X|^{2}$. This is because such indicator functions may further be factorized as shown in Figs.~\ref{fig:elim-complexity}~(b) and (c).
  To see this, the elimination algorithm may start by eliminating $t_1, \dots, t_n$, which induces no complexity as each elimination 
%  may be performed in a fixed number of computations (with respect to the alphabet size), 
  simply accounts to pointing to the proper entry %\footnote{A sum or max indicator function on $m$ variables has $m-1$ degrees of freedom, and hence, eliminating one of its edges may be done with no incurred complexity.  As an extreme example in this direction, the equality indicator on $m$ variables has only one degree of freedom, and so eliminating any $m-1$ of its edges may be done with no induced complexity.} 
  of each function $f_{i}$, resulting in Fig.~(d), where $f'_{i}$ is properly defined according to $\delta_{\sum}$ or $\delta_{\max}$. Now eliminating each
  $e_i$ (starting with $e_{n-1}$) costs $|\X|^{2}$ computations, giving rise to $(n-1)|\X|^{2}$ computations in total.
  Note that if $f = \delta_{=}$, then the complexity of computing the exterior function is $(n-1)|\X|$.
  (For each $x \in \X$, we need $(n-1)$ multiplications.)
  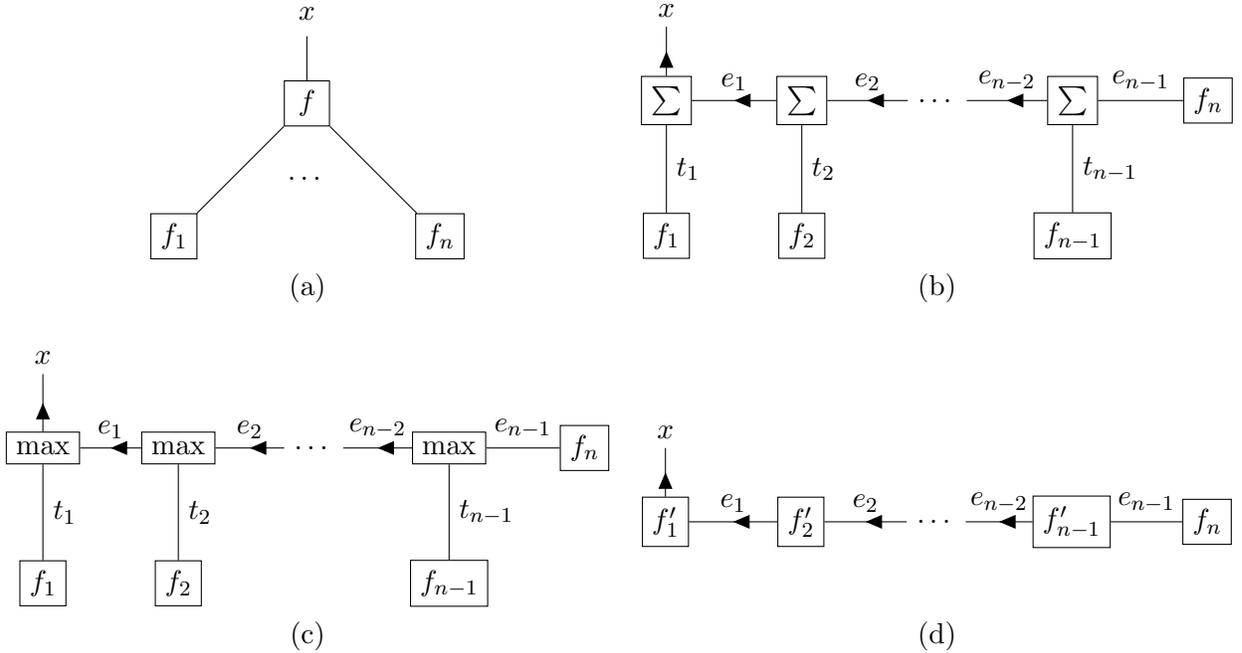
\begin{figure}[ht]
    \center
    \begin{tabular}{cc} 
%    \subfigure[]{
    \begin{tikzpicture}[node distance=2.5cm, every node/.style={rectangle, draw, minimum size=6mm}]
      \node(g){$f$}; \node(d)[node distance=1.2cm, draw=none, above of=g]{$x$};
      \node(f1)[below left of=g]{$f_1$}; \node(fn)[below right of=g]{$f_n$}; 
      \path[every node/.style={}]
      (g)edge (f1) node[node distance=1cm, below of=g]{$\dots$} (g)edge(fn) edge(d);
    \end{tikzpicture} %}
%    \hspace{.5cm}
%    \subfigure[]{
    &
    \begin{tikzpicture}[node distance=1.8cm, every node/.style={rectangle, draw}]
      \node(s1){$\sum$}; \node(s2)[right of=s1]{$\sum$}; \node(dts)[draw = none, right of=s2]{$\dots$}; \node(snm1)[right of=dts]{$\sum$};
      \node(f1)[below of=s1]{$f_1$}; \node(f2)[below of=s2]{$f_2$}; \node(fnm1)[below of=snm1]{$f_{n-1}$}; \node(fn)[right of=snm1]{$f_n$};
      \node(d)[draw=none, node distance=1.2cm, above of=s1]{$x$};
      \path[every node/.style={}]
      (s1)edge[->-out=.5](d)
      (s1)edge node[right]{$t_1$} (f1) (s2)edge node[right]{$t_2$}(f2) (snm1)edge node[right]{$t_{n-1}$}(fnm1) edge node[above]{$e_{n-1}$}(fn)
      (s2)edge[->-out=.5] node[above]{$e_1$} (s1) (dts)edge[->-out=.5] node[above]{$e_2$} (s2) (snm1)edge[->-out=.5] node[above]{$e_{n-2}$} (dts); 
    \end{tikzpicture} %}
    \\
    (a) & (b) \\
    \\
%    \hspace{.5cm}
%    \subfigure[]{
    \begin{tikzpicture}[node distance=1.8cm, every node/.style={rectangle, draw}]
      \node(s1){$\max$}; \node(s2)[right of=s1]{$\max$}; \node(dts)[draw = none, right of=s2]{$\dots$}; \node(snm1)[right of=dts]{$\max$};
      \node(f1)[below of=s1]{$f_1$}; \node(f2)[below of=s2]{$f_2$}; \node(fnm1)[below of=snm1]{$f_{n-1}$}; \node(fn)[right of=snm1]{$f_n$};
      \node(d)[draw=none, node distance=1.2cm, above of=s1]{$x$};
      \path[every node/.style={}]
      (s1)edge[->-out=.5] (d)
      (s1)edge node[right]{$t_1$} (f1) (s2)edge node[right]{$t_2$}(f2) (snm1)edge node[right]{$t_{n-1}$}(fnm1) edge node[above]{$e_{n-1}$}(fn)
      (s2)edge[->-out=.5] node[above]{$e_1$} (s1) (dts)edge[->-out=.5] node[above]{$e_2$} (s2) (snm1)edge[->-out=.5] node[above]{$e_{n-2}$} (dts); 
    \end{tikzpicture} %}
   & 
%    \subfigure[]{
    \begin{tikzpicture}[node distance=1.8cm, every node/.style={rectangle, draw}]
      \node(s1){$f'_1$}; \node(s2)[right of=s1]{$f'_2$}; \node(dts)[draw = none, right of=s2]{$\dots$}; \node(snm1)[right of=dts]{$f'_{n-1}$};
      \node(sn)[right of=snm1]{$f_n$};
%      \node(f1)[below of=s1]{$f_1$}; \node(f2)[below of=s2]{$f_2$}; \node(fnm1)[below of=snm1]{$f_{n-1}$}; \node(fn)[right of=snm1]{$f_n$};
      \node(d)[draw=none, node distance=1.2cm, above of=s1]{$x$};
      \node(dump)[node distance=1cm, draw=none, below of=s1]{};
      \path[every node/.style={}]
      (s1)edge[->-out=.5] (d)
%      (s1)edge node[right]{$t_1$} (f1) (s2)edge node[right]{$t_2$}(f2) (snm1)edge node[right]{$t_{n-1}$}(fnm1) edge node[above]{$e_{n-1}$}(fn)
      (s2)edge[->-out=.5] node[above]{$e_1$} (s1) (dts)edge[->-out=.5] node[above]{$e_2$} (s2) (snm1)edge[->-out=.5] node[above]{$e_{n-2}$} (dts) 
      (snm1)edge node[above]{$e_{n-1}$}(fn); 
    \end{tikzpicture} %}
    \\ (c) & (d)
  \end{tabular}
    \caption{Example~\ref{ex:elim-complexity}: The complexity of eliminating a sum or max indicator function. }
   \label{fig:elim-complexity}
  \end{figure}
  \label{ex:elim-complexity}
\end{Example}

It is not hard to see that one may impose an elimination ordering such that the elimination algorithm
may be viewed as one that merges a node with its neighbors, picks another node and merges it with its neighbors, and so on, until
there is only one node left. 
We refer to such elimination as ``block elimination,'' and to the elimination algorithm at the block level as the ``block elimination algorithm,''
i.e., in each step, the block elimination algorithm replaces a function $f_i$ and its neighbors with the function 
$\big\langle f_{i},f_{j}:j \in {\rm ne}(i) \big\rangle$.
%
%If we refer to such elimination as ``block elimination,'' then it is clear that in each block
%elimination, one replaces a function $f_{i}$ and its neighbors with
%the function $\langle f_{i}, f_{j}:j\in {\rm ne}(i) \rangle$. 
%
In the special case where each block eliminated function $f_i$ is the equality indicator and
none of the neighbors of $f_i$ share any edges, cf. Fig.~\ref{fig:block-elim} and Example~\ref{ex:block-elim}, then the block elimination becomes the classical elimination
algorithm of undirected graphical models. (An undirected graphical model can be converted to a FG which by
Proposition~\ref{prop:mul-NFG} is equivalent to a constrained NFG whose interface functions are equality indicators--- Note that none of the
neighbors of such equality nodes shares any edges due to the bipartite nature of an NFG model.)
%, i.e. a bipartite NFG with equality indicators in one partition, and so the two requirements above are satisfied.)
We remark that whereas the elimination algorithm does not preserve the bipartite structure of an NFG in each edge elimination
step, on the block elimination level, the algorithm respects such bipartite structure. 
%(i.e. if the original NFG is a bipartite, then the resulting NFG from merging any node and its neighbours is a bipartite NFG). 
(Let $I$ and $J$ be the two independent vertex sets, and let $f_{ {\rm ne}(i)}$ be the node resulting from merging node $i \in I$
and its neighbors. Let node $j \in J$ be connected to $f_{ {\rm ne}(i)}$ by some edge $e$, then $e \in E(i)$ or $e \in
E(j')$ for some $j' \in {\rm ne}(i)$. Both such cases are impossible, since $e \in E(i)$ implies $j \in {\rm ne}(i)$ and so $j$ would have been merged with
$i$ in the block elimination, and $e \in E(j')$ violates the bipartite assumption.)
\begin{figure}[ht]
  \centering
%  \begin{tabular}{ccc}
  \subfigure[]{
  \begin{tikzpicture}[scale=.9, every node/.style={draw, rectangle, node distance=\dist, transform shape}]
    \def\dist{2cm}; \def\angle{63.43}; \def\radius{\dist/2*sqrt(5)};
    \node(q1)[label=right:$q_1$]{$=$}; \node(q2)[right of=q1, label=right:$q_2$]{=}; 
    \node(q3)[right of=q2, label=right:$q_3$]{=}; \node(q4)[right of=q3, label=right:$q_4$]{=};
    \node(f1) at ([shift={(180+\angle:\radius)}]q1){$f_1$}; \node(f2) at ([shift={(-\angle:\radius)}]q1){$f_2$};
    \node(f3) at ([shift={(-\angle:\radius)}]q2){$f_3$}; \node(f4) at ([shift={(-\angle:\radius)}]q3){$f_4$};
    \node(g1) at ([shift={(90:\dist/2)}]q1){$g_1$}; \node(g2) at ([shift={(90:\dist/2)}]q2){$g_2$}; 
    \node(g3) at ([shift={(90:\dist/2)}]q3){$g_3$}; \node(x4)[draw=none] at ([shift={(90:\dist/2)}]q4){$x_4$};
    \path[every node/.style={transform shape}]
    (q1)edge node[left]{$x_1$}(g1)edge(f1)edge(f2) (q2)edge node[left]{$x_2$}(g2)edge(f1)edge(f3) (q3)edge node[left]{$x_3$}(g3)edge(f2)edge(f4) (q4)edge(x4)edge(f2)edge(f3);
  \end{tikzpicture} }
%  & &
  \subfigure[]{
  \begin{tikzpicture}[scale=.9, every node/.style={draw, rectangle, node distance=\dist, transform shape}]
    \def\dist{2cm}; \def\angle{63.43}; \def\radius{\dist/2*sqrt(5)};
    \node(q1)[label=right:$q_1$]{$=$}; \node(q2)[right of=q1, label=right:$q_2$]{=};
    \node(q3)[right of=q2, label=right:$q_3$]{=}; \node(q4)[right of=q3, label=right:$q_4$]{=};
    \node(f1) at ([shift={(180+\angle:\radius)}]q1){$f_1$}; \node(f2) at ([shift={(-\angle:\radius)}]q1){$f_2$};
    \node(f3) at ([shift={(-\angle:\radius)}]q2){$f_3$}; \node(f4) at ([shift={(-\angle:\radius)}]q3){$f_4$};
    \node(g1)[label=above left:$f'_1$] at ([shift={(90:\dist/2)}]q1){$g_1$}; \node(g2) at ([shift={(90:\dist/2)}]q2){$g_2$};
    \node(g3) at ([shift={(90:\dist/2)}]q3){$g_3$}; \node(x4)[draw=none] at ([shift={(90:\dist/2)}]q4){$x_4$};
    \path[every node/.style={transform shape}]
    (q1)edge node[left]{$x_1$}(g1)edge(f1)edge(f2) (q2)edge node[left]{$x_2$}(g2)edge(f1)edge(f3) 
    (q3)edge node[left]{$x_3$}(g3)edge(f2)edge(f4)  (q4)edge(x4)edge(f2)edge(f3);
    \draw[dashed] (q1) \clbox{-1.4}{-2.4}{1.4}{1.3};
  \end{tikzpicture} }
%  \\ (a) & & (b) \\
  \subfigure[]{
  \begin{tikzpicture}[scale=.9, every node/.style={draw, rectangle, node distance=\dist, transform shape}]
    \def\dist{2cm}; \def\angle{63.43}; \def\radius{\dist/2*sqrt(5)};
    \node(q2)[label=right:$q_2$]{=}; \node(q3)[right of=q2, label=right:$q_3$]{=}; \node(q4)[right of=q3, label=right:$q_4$]{=}; 
    \node(f1) at ([shift={(180+\angle:\radius)}]q2){$f'_1$}; \node(f3) at ([shift={(-\angle:\radius)}]q2){$f_3$}; \node(f4) at ([shift={(-\angle:\radius)}]q3){$f_4$};
    \node(g2)[label=above left:$f'_2$] at ([shift={(90:\dist/2)}]q2){$g_2$}; \node(g3) at ([shift={(90:\dist/2)}]q3){$g_3$}; \node(x4)[draw=none] at ([shift={(90:\dist/2)}]q4){$x_4$};
    \path[every node/.style={transform shape}]
    (q2)edge node[left]{$x_2$}(g2)edge(f1)edge(f3) 
    (q3)edge node[left]{$x_3$}(g3)edge(f1)edge(f4) (q4)edge(x4)edge(f1)edge(f3) ;
    \draw[dashed] (q2) \clbox{-1.4}{-2.4}{1.4}{1.3};
  \end{tikzpicture} }
%  &
  \subfigure[]{
  \begin{tikzpicture}[scale=.9, every node/.style={draw, rectangle, node distance=\dist, transform shape}]
    \def\dist{2cm}; \def\angle{63.43}; \def\radius{\dist/2*sqrt(5)};
    \node(q3)[label=right:$q_3$]{=}; \node(q4)[right of=q3, label=right:$q_4$]{=}; 
    \node(f3) at ([shift={(180+\angle:\radius)}]q3){$f'_2$}; \node(f4) at ([shift={(-\angle:\radius)}]q3){$f_4$};
    \node(g3)[label=above left:$f'_3$] at ([shift={(90:\dist/2)}]q3){$g_3$}; \node(x4)[draw=none] at ([shift={(90:\dist/2)}]q4){$x_4$};
    \path[every node/.style={transform shape}]
    (q3)edge node[left]{$x_3$}(g3)edge(f1)edge(f4) (q4)edge(x4)edge[bend right=10](f3)edge[bend left=10](f3) ;
    \draw[dashed] (q3) \clbox{-1.4}{-2.4}{1.4}{1.3};
  \end{tikzpicture} }
%  &
  \subfigure[]{
  \begin{tikzpicture}[scale=.9, every node/.style={draw, rectangle, node distance=\dist, transform shape}]
    \def\dist{2cm}; \def\angle{63.43}; \def\radius{\dist/2*sqrt(5)};
    \node(q4)[label=right:$q_4$, label=above left:$f'_4$]{=}; \node(x4)[draw=none] at ([shift={(90:\dist/2)}]q4){$x_4$};
    \node(f3)[below of=q4]{$f'_3$};
    \path[every node/.style={transform shape}]
    (q3)edge(x4)edge[bend right=10](f3)edge[bend left=10](f3) ;
    \draw[dashed] (q4) \clbox{-.8}{-2.4}{.8}{.3};
  \end{tikzpicture} }
%  \\ (c) & (d) & (e)
%  \end{tabular}
  \caption{Block elimination algorithm: (a) an example NFG $\G$,  (b), (c), (d), and (e) illustrate the block elimination of nodes
    $q_1$, $q_2$, $q_3$, and $q_4$, respectively, and it is not hard to see that $Z_{\G}(x_4) = f'_{4}(x_4) = f'_{3}(x_4,x_4)$.
    Note that if $\G$ is obtained from a constrained NFG by gluing functions $g_1$, $g_2$, and $g_3$ into the external edges $x_1$,
    $x_2$, and
    $x_3$ of the constrained NFG, then in the special case where every $g_i$ is the constant-one indicator, it follows that
    $Z_{\G}$ is the marginal probability distribution $p_{X_{4}}$, cf. Section~\ref{sec:inference}.}
  \label{fig:block-elim}
\end{figure}
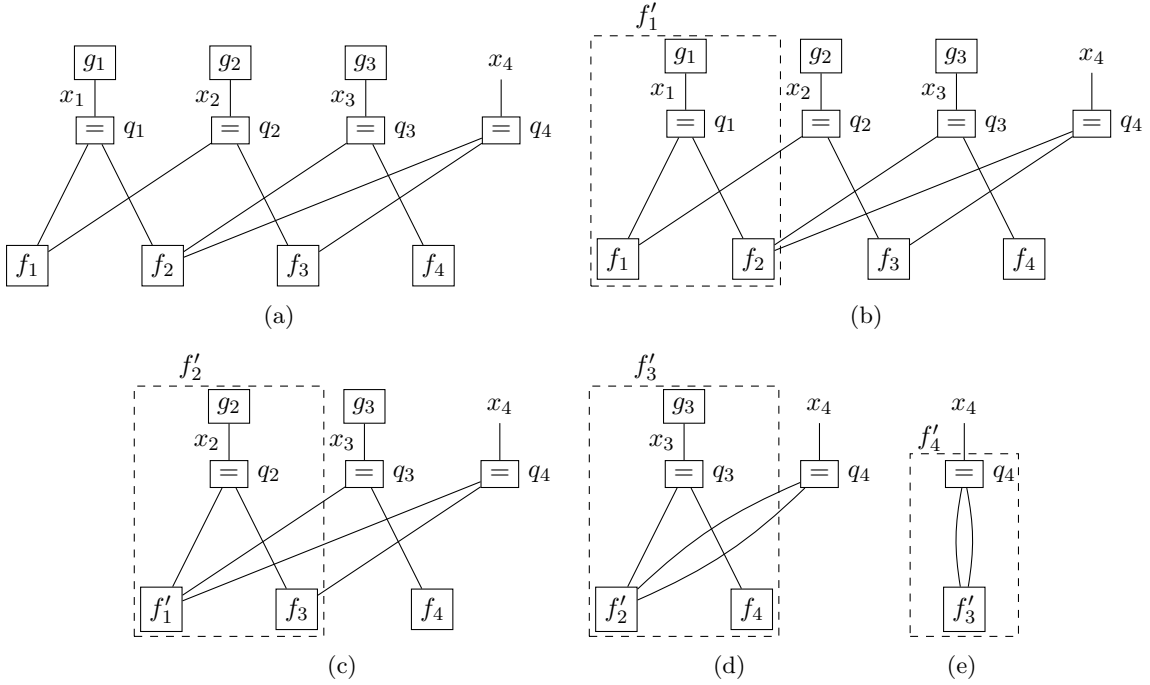

The following example addresses the complexity of a block elimination. %transforming a function.
\begin{Example}
  Let $\G$ be as in Fig.~\ref{fig:elim2-complexity}, where $E_{1}, \dots, E_{n}$ are disjoint. 
  \begin{figure}[ht]
    \center
%    \setcounter{subfigure}{0}
%    \subfigure[]{
      \begin{tikzpicture}[v/.style={rectangle, draw, node distance=\vdistance, minimum size=6mm}, d/.style={node distance=\ddistance}]
	\def\vdistance{2.5cm} \def\ddistance{1.2cm} 
	\node(g)[v]{$f$}; 
        \node(f1)[v, below left of=g]{$f_1$}; \node(fn)[v, below right of=g]{$f_n$}; 
	\node(d1l)[d, below left of=f1]{}; \node(d1r)[d, below right of=f1]{};
	\node(dnl)[d, below left of=fn]{}; \node(dnr)[d, below right of=fn]{};
        \path[every node/.style={}]
	(g)edge node[left]{$x_{1}$} (f1) node[node distance=\vdistance/2, below of=g]{$\dots$}
	(g)edge node[right]{$x_{n}$}(fn)
	(f1)edge(d1l)edge(d1r) node[node distance=\ddistance/2, below of=f1]{$\dots$} node[node distance=\ddistance/1.2, below of=f1]{$E_1$}
	(fn)edge(dnl)edge(dnr) node[node distance=\ddistance/2, below of=fn]{$\dots$} node[node distance=\ddistance/1.2, below of=fn]{$E_n$};
    \end{tikzpicture} %}
    %
%    \subfigure[]{
%      \begin{tikzpicture}[v/.style={rectangle, draw, node distance=\vdistance, minimum size=6mm}, d/.style={node distance=\ddistance}]
%	\def\vdistance{2.5cm} \def\ddistance{1.2cm} 
%	\node(g)[v]{$f$}; 
%        \node(f1)[v, below left of=g]{$f_1$}; \node(fn)[v, below right of=g]{$f_n$}; 
%	\node(d1)[d, below of=f1]{}; % Merely used for alignment of the two subfigures
%	\node(dn)[d, below of=fn]{};
%        \path[every node/.style={}]
%	(g)edge node[left]{$x_{1}$} (f1) node[node distance=\vdistance/2, below of=g]{$\dots$}
%	(g)edge node[right]{$x_{n}$}(fn)
%	(f1)edge node[below right]{$x'_{1}$}(d1) 
%	(fn)edge node[below right]{$x'_{n}$}(dn);
%    \end{tikzpicture} }
    \caption{Example~\ref{ex:block-elim}.}
    \label{fig:elim2-complexity}
  \end{figure}
%  This NFG may be understood as a sub-NFG, the computation of its exterior function corresponds to the elimination of a node and its neighbors, in a general NFG.
  This NFG may be understood as a sub-NFG, the computation of its exterior function corresponds to a block elimination of a node
  and its neighbors in a bigger NFG. Note that if the bigger NFG is a bipartite, then the requirement that $E_1, \ldots, E_n$ be
  disjoint is automatically satisfied.
  Assuming the variable of each edge incident on $f$ takes its values from $\X$, then 
  by recursive application of Lemma~\ref{lemma:elem-complexity}, the complexity 
  of computing the exterior function 
  $Z_{\G}(x_{E_{1}}, \dots, x_{E_{n}})$ is given by (assuming the elimination order $x_1, x_{2}, \dots, x_n$)
  \[\sum_{i=1}^{n} |\X|^{1+n-i+\sum_{j=1}^{i}|E_{j}|}.\]
  That is, if $E_i$ is non-empty for all $i$, then the complexity is of order $|\X|^{1+|E_1|+\cdots+|E_n|}$.
  As an example, consider for instance $n=3$, then the exterior function of $\G$
  is given by
  \[
   Z_{\G}(x_{E_{1}}, x_{E_{2}}, x_{E_{3}}) = 
    \overbrace{\Big\langle f_{3},
    \underbrace{\big\langle f_{2}, 
      \underbrace{\langle f_{1}, f \rangle}_{\rm SP1}
    \big\rangle}_{\rm SP2}
    \Big\rangle}^{\rm SP3}.
    \]
    Computing the first sum of products (SP1), i.e., eliminating $x_{1}$, involves
    $|\X|^{|E_{1}|+3}$ computations by Lemma~\ref{lemma:elem-complexity},
    % since for each $x_{E_{1}}, x_{2}, x_{3}$, one needs $|\X|$ to perform the sum over $x_{1}$. The resulting SP1 is a function of $x_{E_{1}}, x_{2}, x_{3}$, 
    and similarly, one needs $|\X|^{|E_{1}|+|E_{2}|+2}$
    operations to compute $SP2$. Finally, SP3 requires $|\X|^{|E_{1}|+|E_{2}|+|E_{3}|+1}$ operations.
    \label{ex:block-elim}
\end{Example}

Next we consider the transformation of a function. 
Suppose we are interested in the exterior function of the NFG in Fig.~\ref{fig:transform-complexity}~(a), where 
$x$ and $x'$ take their values from a finite alphabet $\X^{n}$ for some positive integer $n$. Clearly the exterior function
represents a matrix-vector multiplication, and hence may be computed in $|\X|^{2n}$ operations.
In the case where the bivariate function (on $\X^{n} \times \X^{n}$) $f'$ factors as the product of bivariate functions
(on $\X \times \X$) $f'_1, \cdots, f'_n$, then from the previous example, it follows that the complexity of transforming the function $f$ via
transformers $f'_{1}, \cdots, f'_{n}$ is $n|\X|^{n+1}$. (This is the special case
with $|E_1|=\dots=|E_{n}|=1$, compare Figs~\ref{fig:elim2-complexity} and \ref{fig:transform-complexity}~(b).) 
Another way of viewing this is that the elimination
of each $x_{i}$ accounts to a matrix-vector multiplication for a given configuration
$x_{N\backslash \{i\}}$, where $N:=\left\{ 1,\cdots, n \right\}$, and hence requires $|\X|^{2}$ computations. That
is, eliminating $x_{i}$ requires $|\X|^{(n-1)+2}$ operations, and eliminating $x_1, \cdots, x_n$ requires in 
total $n|\X|^{n+1}$ operations, as observed.
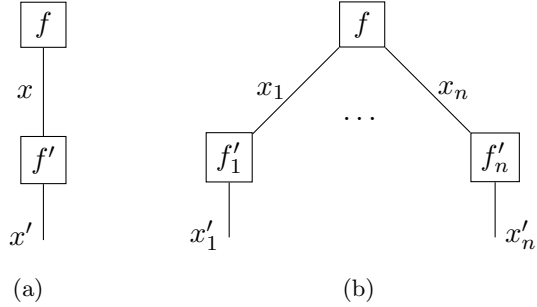
\begin{figure}[ht]
  \center
  \setcounter{subfigure}{0}
  \subfigure[]{
    \begin{tikzpicture}[v/.style={rectangle, draw, node distance=\vdistance, minimum size=6mm}, d/.style={node distance=\ddistance}]
      \def\vdistance{1.8cm} \def\ddistance{1.2cm} 
      \node(u)[v]{$f$}; \node(v)[v, below of=u]{$f'$}; \node(d)[d, below of=v]{};
      \path[every node/.style={}]
      (v)edge node[left]{$x$}(u)edge node[below left]{$x'$}(d);
  \end{tikzpicture} }
  \hspace{1cm}
  \subfigure[]{
    \begin{tikzpicture}[v/.style={rectangle, draw, node distance=\vdistance, minimum size=6mm}, d/.style={node distance=\ddistance}]
      \def\vdistance{2.5cm} \def\ddistance{1.2cm} 
      \node(g)[v]{$f$}; 
      \node(f1)[v, below left of=g]{$f'_1$}; \node(fn)[v, below right of=g]{$f'_n$}; 
      \node(d1)[d, below of=f1]{}; % Merely used for alignment of the two subfigures
      \node(dn)[d, below of=fn]{};
      \path[every node/.style={}]
      (g)edge node[left]{$x_{1}$} (f1) node[node distance=\vdistance/2, below of=g]{$\dots$}
      (g)edge node[right]{$x_{n}$}(fn)
      (f1)edge node[below left]{$x'_{1}$}(d1) 
      (fn)edge node[below right]{$x'_{n}$}(dn);
  \end{tikzpicture} }
  \caption{Transformation of a function.}
  \label{fig:transform-complexity}
\end{figure}

    We emphasize that if the transformers exhibit a special form, then more efficient
    computations may be achieved. For instance, if each $f'_{i}$ is a Fourier
    kernel, then the fast Fourier transform may be used in the elimination of each edge using $\log(|\X|)|\X|^{n}$ computations,
    giving rise to a total $n\log(|\X|)|\X|^{n}$ operations for computing the transformation of $f$. %(That is, $\log(|\X_{N}|)|\X_{N}|$ computations.)
    Another important case where such savings are possible is when each $f'_{i}$ is the
    cumulus or the difference transform. In such case, the matrix-vector multiplication associated with the elimination of each
    edge may be performed using $|\X|$ operations, giving rise to a total complexity of
    $n|\X|^{n}$. To see this, consider our
    example with $n=3$ and assume $f'_1, f'_2, f'_3$ are cumulus transforms. To eliminate $x_1$, we start with the initiation
    $s := f$, and for $x'_{1} = 2, \dots, |\X|$, we update $s$ as,
    \[s(x'_1,x_2,x_3)=f(x'_1,x_2,x_3)+s(x'_1-1,x_2,x_3).\]
    Hence, computing SP1 in $|\X|^{3}$ operations instead of $|\X|^{4}$. Similarly,
    we compute SP2 and SP3, giving rise to a total number of operations $3|\X|^{3}$.
    This clearly extends to any $n$, and below we provide two algorithms for fast computation of the cumulus and difference
    transformations, where to facilitate notation, 
    we use $J^{-}$ to denote the set $\{1, \dots, j-1\}$ and
    $J^{+}$ to denote $\{j+1, \dots, n\}$ for any $j \in \{1, \dots, n\}$.
%    \\
%    \newpage
    \begin{algorithm}[Fast cumulus transform] 
    \emph{\texttt{
      Initialize: $s_{0} = f$; \\
      For $j = 1, \dots, n\{$ \\ 
	\indent For each $(x'_{J^{-}}, x_{J^{+}}) \in \X^{n-1} \{$ \\
	  \indent \indent Initialize: $s_{j}(x'_{J^{-}}, x'_{j}=1, x_{J^{+}})
	  = s_{j-1}(x'_{J^{-}}, x_{j}=1, x_{J^{+}})$; \\
	  \indent\indent For $x'_{j} = 2,\dots,|\X| \{$ \\
	    \indent\indent $s_{j}(x'_{J^{-}}, x'_{j}, x_{J^{+}})=
	    s_{j-1}(x'_{J^{-}}, x'_{j}, x_{J^{+}}) %|_{x_{j}=x'_{j}}
 	    + s_{j}(x'_{J^{-}}, x'_{j}-1, x_{J^{+}})$; \\
	    \indent\indent $\}$ \\
	  \indent $\}$ \\
      $\}$ \\
      Return $s_n$;
      }} 
    \end{algorithm}
\noindent The difference transformation is performed in exactly the same manner, except for the updating rule: 
    \begin{algorithm}[Fast difference transform] 
     \emph{\texttt{
      Initialize: $s_{0} = f$; \\
      For $j = 1, \dots, n\{$ \\ 
	\indent For each $(x'_{J^{-}}, x_{J^{+}}) \in \X^{n-1} \{$ \\
	  \indent \indent Initialize: $s_{j}(x'_{J^{-}}, x'_{j}=1, x_{J^{+}})
	  = s_{j-1}(x'_{J^{-}}, x_{j}=1, x_{J^{+}})$; \\
	  \indent\indent For $x'_{j} = 2,\dots,|\X| \{$ \\
	    \indent\indent $s_{j}(x'_{J^{-}}, x'_{j}, x_{J^{+}})=
	    s_{j-1}(x'_{J^{-}}, x'_{j}, x_{J^{+}}) %|_{x_{j}=x'_{j}}
	    - s_{j-1}(x'_{J^{-}}, x'_{j}-1, x_{J^{+}})$; \\ %|_{x_{j}=x'_{j}}$; \\
	    \indent\indent $\}$ \\
	  \indent $\}$ \\
    $\}$ \\
    Return $s_n$;
}}
% \\
    \end{algorithm}

%In Appendix~?%\ref{appendix:mobius}
%, we give alternative algorithms, based on the fast Mobius and the fast
%Mobius inverse transforms, for computing
%the cumulus and differnce transformations with the same complexity as above.

%  \begin{itemize}
%    \item $(n-1)|\X|$ if $g = \delta_{=}$.
%    \item $(n-1)|\X|^{2}$ if $g = \delta_{\sum}$ or $\delta_{\max}$.
%  \end{itemize}

%The previous two examples suggest that inference over NFGs with equality indicator functions is more favouriable

\subsection{Sum-product algorithm}

Given an NFG $\G$ with no external edges, in many circumstances, for each edge $e$ in the NFG, one may be interested in computing 
the \emph{marginal exterior function} \cite{Forney2011:PartitionFunction} defined as 
\[Z_{\G}(x_{e}):= \sum_{x_{E\backslash \left\{ e \right\}}} \prod_{v \in V} f_{v}(x_{E(v)}).\]
It is possible to compute such marginals using the elimination algorithm by imposing an elimination ordering such that
$x_e$ appears last. (More precisely, the elimination algorithm is slightly modified here such that when it reaches $x_e$ it
multiplies the functions of the two nodes incident with $e$ without summing out $x_e$.)
%
%
%The elimination algorithm is exact for any NFG, however, 
Although the elimination algorithm is exact for any NFG, it may be expensive to perform, as it is likely to produce
nodes with large degrees. Further, in addressing the marginals problem above, the elimination algorithm is repeated for 
each marginal, giving rise to some redundant computations since most of the computations used for evaluating one of the
marginals can be used in determining some other marginals.
The \emph{sum-product algorithm} (SPA) is an efficient alternative in the case of NFGs with no cycles. 
We refer the reader to \cite{frank:factor} for an excellent exposure to the SPA on FGs,
and to \cite{Forney2001:Normal, Forney2011:PartitionFunction} for its formulation on NFGs.

Let $\G$ be a tree NFG (i.e. an NFG whose underlying graph is a tree) with a vertex set $V$ and an edge set $E$ comprised entirely of internal edges, i.e. $\G$ has no
external edges. %($E$ denotes both half and regular edges.)
%To facilitate notation, recall that, for any $v \in V$, we use $L(v) \subseteq E(v)$ 
%to denote the set of half edges incident on $v$, where $E(v)$ is the set of all (regular and half) edges
%incident on $v$. 
To facilitate notation, if nodes $u$ and $v$ are neighbors, we use $e_{uv}$ to denote the edge $\left\{ u,v \right\}$---
  Note that $e_{uv} = e_{vu}$.   
  The SPA defines a set of \emph{messages} that are passed between adjacent nodes, where we use $\mu_{u \rightarrow v}$ to denote the
  message passed from node $u$ to node $v$. The messages are governed by the following \emph{update rule} (Fig.~\ref{fig:spa}): \\
  \[\mu_{u \rightarrow v}(x_{e_{uv}}) = \sum_{x \in \X_{E(u) \backslash \{ e_{uv} \}}} f_u(x,x_{e_{uv}}) 
    \prod_{v' \in \N(u) \backslash \{v\}} \mu_{v' \rightarrow u}(x_{e_{v'u}}),\]
    where a node $u$ sends its message to adjacent node $v$ only if it has received the messages of all its other neighbors, and
    the algorithm terminates when every node has sent a message to all its neighbors.
  That is, the message $\mu_{u \rightarrow v}$ is the sum-of-products
  $\big\langle f_{u}, \mu_{v' \rightarrow u}:v' \in {\rm ne}(u)\backslash \left\{ v \right\} \big\rangle,$
  and the complexity of such computation is $(\deg(u)-1)|\X|^{\deg(u)-1}$.
  In the special case when $f_{u}$ is the indicator function $\delta_{\sum}$ or $\delta_{\max}$, it is not hard to see that the
  complexity becomes $(\deg(u)-2)|\X|^{2}$.
  Further, it is clear that when $f_u$ is the equality indicator, the message sent to node $v$ 
  is simply the multiplication of the incoming
  messages from all other neighbors of $u$, i.e.,
  \[\mu_{u \rightarrow v}(x_{e_{uv}}) = \prod_{v' \in {\rm ne}(u) \backslash \{v\}} \mu_{v' \rightarrow u}(x_{e_{uv}}),\]
  and the complexity of computing such message is $(\deg(u)-2)|\X|$.

  Note that upon termination, the SPA would have computed $2|E|$ messages with two messages per edge, and it is possible to show, due
  to the tree structure, that each marginal $Z_{\G}(x_{e})$ is given as the product of the two messages carried by the edge $e$. That
  is, for any $e_{uv} \in E$, we have 
  \[Z_{\G}(x_{e_{uv}}) = \mu_{u \rightarrow v}(x_{e_{uv}}) \mu_{v \rightarrow u}(x_{e_{uv}}).\]

  We remark that although the SPA was formulated on NFGs with no external edges, this does not present a serious limitation, as 
  in many applications one converts the external edges, if present, into regular ones by gluing the constant-one or the
  evaluation indicators. (For each external edge the choice of the indicator function depends on interest and whether such edge is
  observed as ``evidence'' or not, cf. Section~\ref{sec:inference}.)
%  many applications although one may start with an NFG with external edges, such external edges may be converted into internal
%  ones by glueing the constant-one or the evaluation indicators. (For each external edge the choice of the indicator function
%  depends on interest and whether such edge is observed as ``evidence'' or not, cf. Section~\ref{sec:inference}.) 
  However, if one insists on NFGs with external edges, then the SPA algorithm still works for such NFGs, and it is possible to show that in this case
  if $L$ is the set of external edges, then for any internal edge $e_{uv} \in E$, the product 
  $\mu_{u \rightarrow v} \mu_{v \rightarrow u}$ is equal to $Z_{\G}(x_{e_{uv}}, x_{L})$ defined as $Z_{\G}(x_{e_{uv}},
  x_{L}):= \sum_{x_{E \backslash (L \cup \left\{e_{uv} \right\})}} \prod_{v\in V} f_{v}(x_{E(v)})$, and hence, the exterior
  function is given as $Z_{\G}(x_L) = \sum_{x_{e_{uv}}} Z_{\G}(x_{e_{uv}},x_{L})$. However, in this case, the SPA might be
  expensive to perform since the size of each message increases every time it is passed by a node with a dangling edge, as the
  message accumulates the variable of such edge as an argument.

%  For each node $u$ with a non-empty set of dangling edges $L(u)$, 
%  the \emph{summary} is defined as the multiplication of $f_u$ 
%  and all the incoming messages, summed over the 
%  regular edges incident on $u$. It is not hard to show that when the underlying graph
%  is a tree, and the exterior function of the NFG is the joint distribution of RVs $X_L$, then the summary at node $u$ is equal to the probability
%  distribution of the RVs $X_{L(u)}$, see the example below.
%  That is, when the NFG is a tree, we have
%  \[p(x_{L(u)}) = \sum_{x \in \X_{E(u)\backslash L(u)}} f_u(x_{L(u)},x) \prod_{v \in {\rm ne}(u)} \mu_{v \rightarrow u}(x_{e_{uv}}).\]

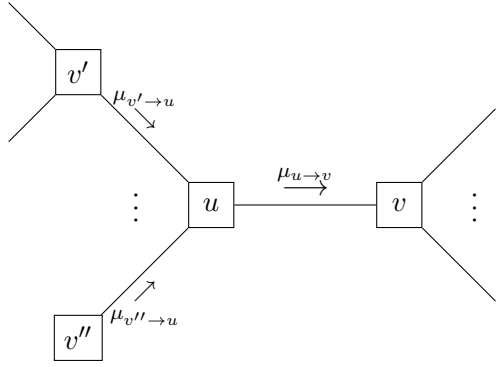
\begin{figure}[ht]
\centering
\begin{tikzpicture}[v/.style={node distance=2.5cm,rectangle,fill=none,draw, minimum size=6mm}, 
  d/.style={node distance=2cm}, every node/.append style={transform shape}]
  \node(u)[v]{$u$}; \node(nu1)[v, above left of=u]{$v'$}; \node(nun)[v, below left of=u]{$v''$};
  \node(v)[v,right of=u]{$v$};
  \node(vd1)[d, above right of=v]{}; \node(vd2)[d, below right of=v]{}; 
  \node(nud1)[d, node distance=1.5cm, above left of=nu1]{}; \node(nud2)[d, node distance=1.5cm, below left of=nu1]{};
  \path[every node/.style={transform shape}]
    (u)edge node[above]{$\stackrel{\mu_{u\rightarrow v}}{\longrightarrow}$} (v)
%	      node[below]{$\stackrel{\mu_{v\rightarrow u}}{\longleftarrow}$} (v)
    (u)edge node[above]{${\mu_{v'\rightarrow u} \atop \searrow}$}  (nu1)
    node[left of=u, rotate=90]{\dots} 
    (u)edge node[below]{${\nearrow \atop \mu_{v'' \rightarrow u}}$} (nun)
    (v)edge(vd1) edge(vd2)
    node[right of=v, rotate=90]{\dots}
    (nu1)edge(nud1)edge(nud2)
    ;
\end{tikzpicture}
\caption{SPA update rule, where the message sent from node $u$ to its neighbor $v$ is computed using all the messages arriving to $u$
  from all its other neighbors, and its local function $f_u$, namely, $\mu_{u\rightarrow v}=\big\langle f_u, \mu_{v'\rightarrow u},
  \ldots, \mu_{v''\rightarrow u}\big\rangle$.}
\label{fig:spa}
\end{figure}

Below, we give a simple example illustrating the SPA.
\begin{Example}[SPA example]
Given the NFG in Fig.~\ref{fig:ex-spa}, we have

\begin{figure}[ht]
 \centering
 \begin{tikzpicture}[node distance=2.5cm,every node/.style={rectangle,fill=none,draw}]
	\node(p1) {$f_1$}; 
	\node(p2) [right of=p1] {$f_2$};
	\node(p3) [right of=p2] {$f_3$};
	\node(p4) [right of=p3] {$f_4$}; 
	\node(p5) [above of=p3] {$f_5$};
%	\node[draw=none](d1)[above of=p2]{$x_1$}; \node[draw=none](d2)[above of=p4]{$x_2$};
	\path[every node/.style={}]
%	(p2)edge(d1) (p4)edge(d2)
	(p1)edge node[below]{$y_1$} node[above]{$\stackrel{\mu_{1\rightarrow 2}}{\longrightarrow} \ \stackrel{\mu_{2\rightarrow 1}}{\longleftarrow}$} (p2)
	(p2)edge node[below]{$y_2$} node[above]{$\stackrel{\mu_{2\rightarrow 3}}{\longrightarrow} \ \stackrel{\mu_{3\rightarrow 2}}{\longleftarrow}$}(p3)
	(p3)edge node[below]{$y_3$} node[above]{$\stackrel{\mu_{3\rightarrow 4}}{\longrightarrow} \ \stackrel{\mu_{4\rightarrow 3}}{\longleftarrow}$}(p4)
	(p3)edge node[left]{$y_4$} node[above, rotate=-90]{$\stackrel{\mu_{5\rightarrow 3}}{\longrightarrow}
	 \ \stackrel{\mu_{3\rightarrow 5}}{\longleftarrow}$}
%	  \above 0pt \uparrow \mu_{3\rightarrow 5}}$} (p5) % one may also simple use \atop 
	(p5);
	%\ \ \stackrel{\mu_{5\rightarrow 4}}{\longleftarrow}$}(p5);
    \end{tikzpicture}
    \caption{SPA example.}
    \label{fig:ex-spa}
\end{figure}
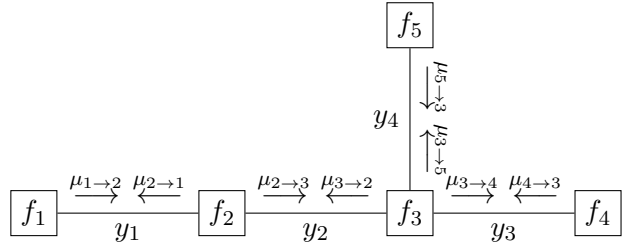

\begin{flalign*}
&  \mu_{1\rightarrow 2}(y_1) = f_{1}(y_{1}), & &
\mu_{4\rightarrow 3}(y_3) =  f_4(y_3),  \\
&  \mu_{2\rightarrow 3}(y_2) = \sum_{y_1} f_2(y_1,y_2) \mu_{1\rightarrow 2}(y_1), & &
\mu_{3\rightarrow 5}(y_4) = \sum_{y_2,y_3} f_3(y_2, y_3,y_4) \mu_{2\rightarrow 3}(y_2)\mu_{4\rightarrow 3}(y_3), \\
&  \mu_{5\rightarrow 3}(y_4) = f_{5}(y_{4}), & &
\mu_{3\rightarrow 2}(y_2) = \sum_{y_3, y_{4}} f_3(y_2,y_3, y_{4}) \mu_{4\rightarrow 3}(y_3) \mu_{5\rightarrow 3}(y_{4}), \\
&  \mu_{3\rightarrow 4}(y_3) = \sum_{y_2,y_4} f_3(y_2,y_3,y_4) \mu_{2\rightarrow 3}(y_2)\mu_{5\rightarrow 3}(y_4), & &
\mu_{2\rightarrow 1}(y_1) = \sum_{y_2} f_2( y_1,y_2) \mu_{3\rightarrow 2}(y_2), \\
\end{flalign*}
and it is clear by direct substitution that, for instance,
\[Z_{\G}(y_2) = \mu_{2 \rightarrow 3}(y_2) \mu_{3 \rightarrow 2}(y_2).\]   %\sum_{y_1,y_2} f_2(x_1,y_1,y_2) \mu_{1\rightarrow 2}(y_1) \mu_{3\rightarrow 2}(y_2). \]
\end{Example}

%$\rotatebox[origin=c]{180}{$\Gamma$}$

Another example is provided below in relation to the difference transform and the ``derivative sum-product algorithm'' of Huang
and Frey \cite{Frey:CDN2011}.

\begin{Example}
  \begin{figure}[ht]
    \centering
    \begin{tikzpicture}[every node/.style={node distance=\dist, draw, rectangle, transform shape}]
      \def\dist{2cm};
      \node(q1)[label=left:$q_1$]{$=$}; 
      \node(f1)[below of=q1]{$f_1$}; 
      \node(f2) at ([shift={(-45:\dist*sqrt(2))}]q1) {$f_2$}; 
      \node(f3)[right of=f2]{$f_3$};
      \node(q2)[above of=f3, label=right:$q_2$]{$=$}; 
      \node(d1)[node distance=\dist*.7, above of=q1, label=left:$d_1$]{$D$}; 
      \node(d2)[node distance=\dist*.7, above of=q2, label=right:$d_2$]{$D$}; 
      \node(u1)[node distance=\dist*.7, above of=d1, label=left:$u_1$]{$\delta_{\overline{x}_1}$}; 
      \node(u2)[node distance=\dist*.7, above of=d2, label=right:$u_2$]{$\delta_{\overline{x}_2}$}; 
      \path[every node/.style={transform shape}]
      (u1)edge node[right]{$x_1$}(d1)
      (u2)edge node[left]{$x_2$}(d2)
      (q1)edge node[right]{$y_1$}(d1)edge node[right]{$y'_1$}(f1)edge node[right]{$y''_1$}(f2)
      (q2)edge node[left]{$y_2$}(d2)edge node[left]{$y''_2$}(f2)edge node[left]{$y'_2$}(f3) ;
    \end{tikzpicture}
    \caption{}
    \label{fig:spa-diff}
  \end{figure}
Let $\G$ be as in Fig.~\ref{fig:spa-diff}, the SPA computes the following messages: \\
\begin{tabular}{ll}
$\mu_{u_{1} \rightarrow d_{1}}(x_1) = \delta_{\overline{x}_{1}}(x_1)$; &
$\mu_{u_{2} \rightarrow d_{2}}(x_2) = \delta_{\overline{x}_{2}}(x_2)$;  \\
$\mu_{d_1 \rightarrow q_1}(y_1) = \sum_{x_1}D(x_1,y_1) \mu_{u_1 \rightarrow d_1}(x_1)$; \hspace{1cm}  &
$\mu_{d_2 \rightarrow q_2}(y_2) = \sum_{x_2}D(x_2,y_2) \mu_{u_2 \rightarrow d_2}(x_2)$;  \\
$\mu_{f_1 \rightarrow q_1}(y'_1) = f_1(y'_1)$; & 
$\mu_{q_2 \rightarrow f_3}(y'_2) = \mu_{f_2 \rightarrow q_2}(y'_2) \mu_{d_2 \rightarrow q_2}(y'_2)$; \\
$\mu_{q_1 \rightarrow f_2}(y''_1) = \mu_{d_1 \rightarrow q_1}(y''_1) \mu_{f_1 \rightarrow q_1}(y''_1)$; &
$\mu_{q_2 \rightarrow f_2}(y''_2) = \mu_{d_2 \rightarrow q_2}(y''_2) \mu_{f_3 \rightarrow q_2}(y''_2)$; \\
$\mu_{f_2 \rightarrow q_2}(y''_2)=\sum_{y''_1}f_2(y''_1,y''_2) \mu_{q_1 \rightarrow f_2}(y''_1)$; &
$\mu_{f_2 \rightarrow q_1}(y''_1)=\sum_{y''_2}f_2(y''_1,y''_2) \mu_{q_2 \rightarrow f_2}(y''_2)$; \\
$\mu_{f_3 \rightarrow q_2}(y'_2) = f_3(y'_2)$; &
$\mu_{q_1 \rightarrow f_1}(y'_1) = \mu_{f_2 \rightarrow q_1}(y'_1) \mu_{d_1 \rightarrow q_1}(y'_1)$; \\
$\mu_{q_2 \rightarrow d_2}(y_2) = \mu_{f_2 \rightarrow q_2}(y_2) \mu_{f_3 \rightarrow q_2}(y_2)$; &
$\mu_{q_1 \rightarrow d_1}(y_1) = \mu_{f_1 \rightarrow q_1}(y_1) \mu_{f_2 \rightarrow q_1}(y_1)$; \\
$\mu_{d_2 \rightarrow u_2}(x_2) = \sum_{y_2} D(x_2,y_2) \mu_{q_2 \rightarrow d_2}(y_2)$; &
$\mu_{d_1 \rightarrow u_1}(x_1) = \sum_{y_1} D(x_1,y_1) \mu_{q_1 \rightarrow d_1}(y_1)$.
\end{tabular}

By noting that %$\mu_{d_1 \rightarrow q_1}(y_1) = \delta_{|\X|}(y_1)$ (follows from the definition of $D$) and 
$\mu_{d_1 \rightarrow q_1}(y_1) = D(\overline{x}_1, y_1)$ and $\mu_{d_2 \rightarrow q_2}(y_2) = D(\overline{x}_2, y_2)$, one can see by direct substitution that
\[\mu_{d_1 \rightarrow u_1}(x_1) = \sum_{y_1}D(x_1,y_1)f_1(y_1) \sum_{y''_2}D(\overline{x}_2, y''_2)f_2(y_1,y''_2)f_3(y''_2),\] and
\[\mu_{d_2 \rightarrow u_2}(x_2) = \sum_{y_2}D(x_2,y_2)f_3(y_2)\sum_{y''_1}D(\overline{x}_1,y''_1)f_2(y''_1,y_2)f_1(y''_1).\]
If we use the notation $\partial_{x}f(x,y):=\sum_{z}D(x,z)f(x,y)$ to denote the difference transform of a function $f$ (with respect to
$x$), then the above two messages can equivalently be written as:
\[\mu_{d_1 \rightarrow u_1}(x_1) = \partial_{x_1}\big[f_1(x_1) \partial_{x_2}[f_2(x_1,x_2)f_3(x_2)]|_{x_2=\overline{x}_2} \big], \]
which is the difference transform of the product of all hidden functions with respect to $x_1, x_2$ evaluated at $x_2 =
\overline{x}_2$, and similarly
\[\mu_{d_2 \rightarrow u_2}(x_2) = \partial_{x_2}\big[f_3(x_2) \partial_{x_1}[f_2(x_1,x_2)f_1(x_1)]|_{x_1=\overline{x}_1}\big] \] 
is the difference transform evaluated at $x_1 = \overline{x}_1$. 
\label{ex:diff-spa}
\end{Example}

Let $\G$ be a tree NFG that is also a constrained NFG with a set $I$ of interface nodes consisting of equality indicators. 
%in which each interface function $i \in I$ is an equality indicator. 
Let $\G'$ be the
NFG obtained from $\G$ by gluing on each dangling edge a difference function, and $\G''$ be the NFG resulting from
$\G'$ by gluing an evaluation function on the other end of each such difference function, cf. Fig.~\ref{fig:spa-diff}.
%that has an evaluation function on its other end, cf. Fig.~\ref{fig:spa-diff}. 
Performing the SPA on $\G''$, it is clear that the discussion in the previous example
extends to any such $\G''$. That is, the message $\mu_{d_i \rightarrow u_i}(x_i)$ is the difference transform of the product
of the hidden functions with respect to $x_I$ evaluated at $\overline{x}_{I\backslash \left\{ i \right\}}$, where $d_i$ is the
difference node adjacent to interface node $i$ and $u_i$ is the evaluation node adjacent to $d_i$. This is the
``derivative-sum-product'' algorithm of \cite{Frey:CDN2011} in the case of finite alphabets. %, and we invite the reader to contrast the two. 

We close with two remarks: First, we emphasize that in performing the SPA over $\G''$, one may utilize the fast difference algorithm in computing the messages
emitted by the difference nodes, making the complexity of computing such messages linear in the alphabet size.
Second, as one may expect, marginalization by summation on $\G'$ is marginalization by evaluation on $\G$, and the difference
function guarantees the conversion between such notions of marginalization, as demonstrated in the example below.   
%Our second remark is illustrated in the following example.
\begin{Example}[continues=ex:diff-spa]
  Assume the function of node $u_1$ in the NFG in Fig.~\ref{fig:spa-diff} is replaced with the constant-one indicator (this accounts to marginalizing
  $x_1$ by summing it out), and assume all variables take their values from a set $\X$. Then we have $\mu_{u_1 \rightarrow d_1}(x_1) = 1$, and
  from the definition of the difference function, it is clear that $\mu_{d_1 \rightarrow q_1}(y_1) = \delta_{|\X|}(y_1)$. By direct
  substitution, it follows that
  \[\mu_{d_2 \rightarrow u_2}(x_2) = \sum_{y_2}D(x_2,y_2)f_1(|\X|)f_2(|\X|,y_2)f_3(y_2).\]
  Or equivalently,
  \[\mu_{d_2 \rightarrow u_2}(x_2) = \partial_{x_2}\big[ f_1(|\X|)f_2(|\X|,x_2)f_3(x_2) \big] \] 
  That is, the message $\mu_{d_{2}\rightarrow u_{2}}$ is the difference transform with respect to $x_2$ of the multiplication of the hidden functions, with
  $x_1$ marginalized by evaluation at $|\X|$. One may verify that $\mu_{d_1 \rightarrow u_1}$ remains unchanged. 
\end{Example}

\subsection{Indirect evaluation of the exterior function}
%Given a bipartite NFG $\G(I \cup J, E, f_{I \cup J})$ that encodes an inference problem, i.e.,
%the exterior function of $\G$ is the desired inference computation.
%Note that we can no longer assume that each node $i \in I$ has a half edge, however, if there is a half edge 
%incident with $i$, then there is exactly one such
%half edge, which we identify by $i$. Let $I' \subseteq I$ be the set of interface nodes having
%half edges incident on them.

Below, we discuss an indirect method for finding the exterior function of a bipartite NFG $\G$, where a holographic
transformation, a one that preserves the exterior function, is applied to the NFG \emph{a priori} in hope of facilitating the
computation.  The idea is to perform a transformation, if it exists, that replaces each function $f_i$ with an equality indicator,
and hence benefit from the low computational complexity of such nodes in the elimination or the SPA.  Of course
one might not be able to find a holographic transformation that converts each function $f_i$ into an equality indicator, in which
case, one may settle with one that produces such effect for a subset $I' \subseteq I$.  For this approach to work, it
is necessary that:
  1) There exists a set of transformers $\left\{g_{e}: e\in E \right\}$ such that
    $\big\langle f_{i}, g_{e}: e \in E(i) \big\rangle = \delta_{=}$, for all $i \in I'$ for some non-empty $I'$, and
    2) There exists efficient algorithms for computing the transformations induced by $g_{e}$.
The following example demonstrates this approach.
\begin{Example}
%  Let $I$ be a finite set and let $\X_i$ be a finite ordered set for all $i \in I$. 
  Let $\X$ be a finite ordered set and consider the NFG in
  Fig.~\ref{fig:indirect-infer}~(a) where each edge variable assumes its values from the partially-ordered set 
  $\X^{n}$ for some integer $n$. Directly computing the exterior function requires $|\X|^{2n}$ operations. However, the
  exterior function may be computed using a number of operations of order $n|\X|^{n}$ by first computing the fast cumulus
  transforms of $f_1$ and $f_2$, multiply the resulting two functions, and then invert the result using the fast difference
  transform, Fig.~(c). The exterior function is invariant under such procedure since it simply accounts to the holographic 
  transformation in Fig.~(b), which is equivalent to the NFG in (a) by the GHT and to the one in (c) by Lemma~\ref{lem:cumulus}.

  The discussion above parallels the well-known fast Fourier transform approach to finding the convolution of functions. In this
  case, the fast Fourier transform is used to reduce the complexity of computing the convolution of two functions defined on
  $\X^{n}$ (assuming $\X$ is a finite abelian group) from order $|\X|^{2n}$ to $n\log(|\X|)|\X|^{n}$ by first computing the fast
  Fourier transforms of $f_1$ and $f_2$, multiply the resulting two functions, and then invert the result using the fast Fourier
  inverse transform. This is justified by the relation between the sum and parity indicators, and the duality of the equality and
  parity indicators under the Fourier transform, Figs~\ref{fig:indirect-infer}~(d)--(f). 
\end{Example}
\begin{figure}[ht]
  \center
  \setcounter{subfigure}{0}
  \subfigure[]{
  \begin{tikzpicture}[v/.style={node distance=2cm, draw, rectangle, minimum size=6mm}, 
    d/.style={node distance=1.5cm}]
    \node(max)[v]{$\max$}; \node(d1)[d, above of=max]{};
    \node(f)[v, below left of=max]{$f_1$}; \node(g)[v, below right of=max]{$f_2$};
    \path[every node/.style={}]
    (f)edge(max) (max)edge(g) (max)edge[->-out=.5](d1);
  \end{tikzpicture} }
  \hspace{.5cm}
  \subfigure[]{
  \begin{tikzpicture}[v/.style={node distance=3cm, draw, rectangle, minimum size=6mm}, 
     t/.style={node distance=1cm, draw, rectangle, minimum size=4mm},
     d/.style={node distance=3cm}]
    \node(max)[v]{$\max$}; 
    \node(d1)[d, above of=max]{};
    \node(Am)[t, above of=max]{$A$}; \node(Dm)[t, above of=Am]{$D$};
    \node(mDl)[t, below left of=max]{$D$}; \node(mDr)[t, below right of=max]{$D$};
    \node(f)[v, below left of=max]{$f_1$}; \node(g)[v, below right of=max]{$f_2$};
    \node(Af)[t, above right of=f]{$A$}; \node(Ag)[t, above left of=g]{$A$};
    \path[every node/.style={}]
    (f)edge(Af) (Af)edge[o--](mDl) (mDl)edge[o--](max)
    (max)edge[--o](mDr) (mDr)edge[--o](Ag) (Ag)edge(g) 
    (max)edge[->-out=.7](Am) (Am)edge[o--](Dm) (Dm)edge[o--](d1);
  \end{tikzpicture} }
  \hspace{.5cm}
  \subfigure[]{
  \begin{tikzpicture}[v/.style={node distance=2cm, draw, rectangle, minimum size=6mm}, 
    d/.style={node distance=2cm}]
    \node(max)[v]{$=$}; \node(d1)[d, above of=max]{}; \node(Dm)[v, node distance=1cm, above of=max]{$D$};
    \node(f)[v, below left of=max]{$F_1$}; \node(g)[v, below right of=max]{$F_2$};
    \path[every node/.style={}]
    (f)edge(max) (max)edge(g) (max)edge(Dm) (Dm)edge[o--](d1);
  \end{tikzpicture}	}
  \subfigure[]{
  \begin{tikzpicture}[v/.style={node distance=2cm, draw, rectangle, minimum size=6mm}, 
    d/.style={node distance=1.5cm}]
    \node(max)[v]{$\sum$}; \node(d1)[d, above of=max]{};
    \node(f)[v, below left of=max]{$f_1$}; \node(g)[v, below right of=max]{$f_2$};
    \path[every node/.style={}]
    (f)edge(max) (max)edge(g) (max)edge[->-out=.5](d1);
  \end{tikzpicture} }
  \hspace{.5cm}
  \subfigure[]{
  \begin{tikzpicture}[v/.style={node distance=2.2cm, draw, rectangle, minimum size=4mm}, 
     t/.style={node distance=.7cm, draw, rectangle, inner sep=.5mm, minimum size=4mm},
     d/.style={node distance=.8cm}]
    \node(max)[v]{$+$}; 
    \node(Am)[t, above of=max]{$\hat{\kappa}$};  
    \node(Dm)[t, above of=Am]{$\kappa$}; \node(s2)[t, above of=Dm]{$+$};
    \node(s2T1)[t, above of=s2]{$\kappa$}; \node(s2T2)[t, above of=s2T1]{$\hat{\kappa}$};
    \node(d1)[d, above of=s2T2]{};
    \node(mDl)[t, node distance=.75cm, below left of=max]{$\hat{\kappa}$}; \node(mDr)[t, node distance=.75cm, below right of=max]{$\hat{\kappa}$};
    \node(f)[v, below left of=max]{$f_1$}; \node(g)[v, below right of=max]{$f_2$};
    \node(Af)[t, node distance=.75cm, above right of=f]{$\kappa$}; \node(Ag)[t, node distance=.75cm, above left of=g]{$\kappa$};
    \path[every node/.style={}]
      (f)edge(Af) (Af)edge(mDl) (mDl)edge(max)
      (max)edge(mDr) (mDr)edge(Ag) (Ag)edge(g) 
      (max)edge(Am) (Am)edge(Dm) (Dm)edge(s2) (s2)edge(s2T1) (s2T1)edge(s2T2) (s2T2)edge(d1);
      \draw[dashed] (f) \clbox{-.4}{-.4}{.8}{.8};
      \draw[dashed] (g) \clbox{-.8}{-.4}{.4}{.8};
      \draw[dashed] (max) \clbox{-.8}{-.8}{.8}{1};
      \draw[dashed] (s2) \clbox{-.4}{-1}{.4}{1};
  \end{tikzpicture} }
  \hspace{.5cm}
  \subfigure[]{
  \begin{tikzpicture}[v/.style={node distance=2cm, draw, rectangle, minimum size=6mm}, 
    d/.style={node distance=2cm}]
    \node(max)[v]{$=$}; \node(d1)[d, above of=max]{}; \node(Dm)[v, node distance=1cm, above of=max]{$\hat{\kappa}$};
    \node(f)[v, below left of=max]{$\widehat{f}_1$}; \node(g)[v, below right of=max]{$\widehat{f}_2$};
    \path[every node/.style={}]
      (f)edge(max) (max)edge(g) (max)edge(Dm) (Dm)edge(d1);
  \end{tikzpicture}	}
  \caption{Indirect computation of the exterior function, where $F_j$, and $\widehat{f}_j$ are the cumulus and Fourier transforms of $f_j$, respectively.}
  \label{fig:indirect-infer}
\end{figure}
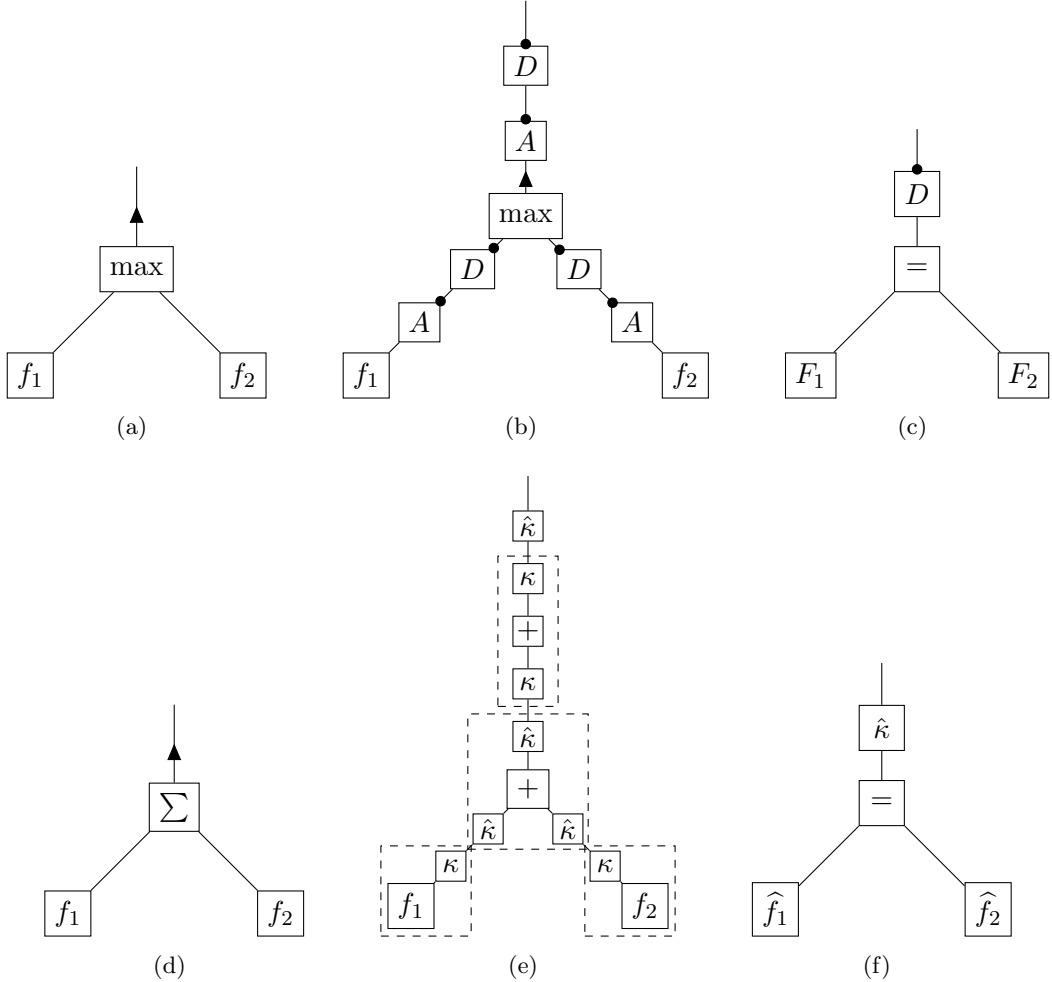

\section{The inference problem}
\label{sec:inference}
%Given a joint probability distribution $p_{X_A}$ of a set of RVs $X_A$, the inference problem is 
%to marginalize a set $M \subseteq A$ of the RVs and to evaluate $p_{X_A}$ at some observed values (evidence) of RVs $N \subseteq A$, where
%$M$ and $N$ are disjoint. That is, the problem is to find
%\[  p_{X_R}(x_R,\overline{x}_N) = \sum_{x_M} p_{X_A}(x_R,x_M,x_N) |_{x_N = \overline{x}_N},  \]
%where $R = A \backslash (M \cup N)$.
%The complexity of such a computation depends primarily on the factorization structure of $p_{X_A}$.

Given an NFG $\G(V,E, f_V)$ representing a set of RVs $X_{L}$.
(That is, the exterior function of $\G$ is the probability distribution $p_{X_{L}}$.)
The inference problem is 
to marginalize a set $M \subseteq L$ of the RVs and to evaluate $p_{X_{L}}$ at some observed values (evidence) of RVs $N \subseteq L$, where
$M$ and $N$ are disjoint. That is, the problem is to find
\[  p_{R}(x_R,\overline{x}_N) = \sum_{x_M} p_{X_{L}}(x_R,x_M,x_N) |_{x_N = \overline{x}_N},  \]
where $R = L \backslash (M \cup N)$.
We remark that, in general, $p_{R}$ is not a probability distribution over the RVs $X_{R}$.
It is rather an \emph{up to scale} distribution over $X_{R}$, namely, 
it is the conditional distribution $p_{X_{R}|x_{N}}(x_{R}|x_{N}=\overline{x}_{N})$ 
up to the scaling constant $p_{X_{N}}(\overline{x}_{N})$.
Evidently, the complexity of the inference problem depends primarily on the factorization structure of $p_{X_{L}}$, which
is reflected by the graphical structure of the NFG. In order to perform the desired inference, we define $\G^{*}$ as the NFG 
whose exterior function is the desired function $p_{R}$. That is, we define $\G^{*}$ as the NFG obtained from $\G$ by:
1) Converting each dangling edge $e \in M$ into a regular edge by gluing a new vertex $u_{e}$ to $e$, where $u_e$ is
associated the constant-one function. 
2) Convert each dangling edge $e \in N$ into a regular edge by gluing a new vertex $v_{e}$ to $e$, where $v_e$ is
associated the evaluation indicator $\delta_{\overline{x}_{e}}$.
%2) Each function $f_v \in f_V$ in $\G$ having an incident edge from $N$ is replaced with
%$f_{v}^{*}(x_{E(v)\backslash N}) = f_v(x_{E(v)\backslash N}, \overline{x}_{E(v) \cap N})$,
%and all the half edges in $N$ are deleted.
%3) All remaining functions are copied to $\G^{*}$ unchanged.
%Alternatively, step 2) above may be replaced with the following step: $2')$ Convert each dangling edge
%$e \in N$ into a regular edge by glueing a new vertex $u_{e}$ to $e$, where $u_e$ is
%associated the evaluation indicator $\delta_{\overline{x}_{e}}$. It is clear that the resulting
%NFGs are equivalent, and hence, we will use $\G^{*}$ to refer to either one.
An example is shown in Fig.~\ref{fig:gstar} where the original NFG is as in (a).
Assuming we are interested in 
$\sum_{x_3} p_{X_{1}X_2X_3}(x_1,x_2,x_3)|_{x_2 = \overline{x}_2}$, 
then $\G^{*}$ is as in (b). % or (c).

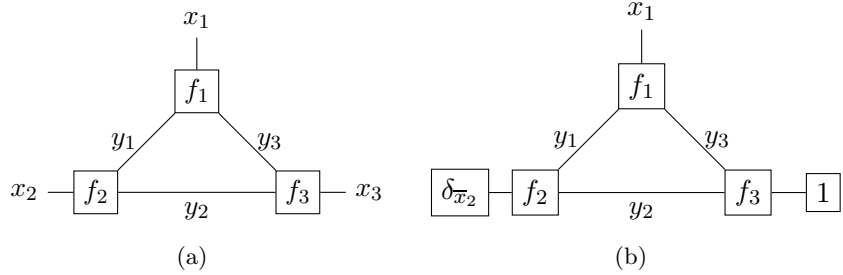
\begin{figure}[ht]
  \begin{center}
    \setcounter{subfigure}{0}
  \subfigure[]{
    \begin{tikzpicture}
      [scale=.95, v/.style={node distance=2cm, rectangle, draw}, d/.style={node distance=1cm}, 
       every node/.append style={transform shape}]
      \node(p1)[v]{$f_1$}; \node(p2)[v, below left of=p1]{$f_2$}; \node(p3)[v, below right of=p1]{$f_3$};
      \node(d1)[d, above of=p1]{$x_1$}; \node(d2)[d, left of=p2]{$x_2$};
      \node(d3)[d, right of=p3]{$x_3$}; 
      \path[every node/.style={transform shape}]
      (p1)edge node[left]{$y_1$}(p2) (p2)edge node[below]{$y_2$}(p3)edge(d2) (p3)edge node[right]{$y_3$}(p1)
      (p1)edge(d1) (p3)edge(d3);
  \end{tikzpicture} }
  \hspace{0cm}
%  \subfigure[]{
%  \begin{tikzpicture}
%      [scale=.95, every node/.style={node distance=2cm, rectangle, draw, transform shape}]
%      \node(p1){$f_1$};
%      \node(p2)[below left of=p1]{$f^{*}_2$};
%      \node(p3)[below right of=p1]{$f_3$};
%      \node(d1)[draw=none, node distance=1cm, above of=p1]{$x_1$}; 
%      \node(d3)[node distance=1cm, right of=p3]{$1$}; 
%      \path[every node/.style={transform shape}]
%      (p1)edge node[left]{$y_1$}(p2) (p2)edge node[below]{$y_2$}(p3) (p3)edge node[right]{$y_3$}(p1)
%      (p1)edge(d1) (p3)edge(d3);
%  \end{tikzpicture} }
%  \hspace{0cm}
  \subfigure[]{
  \begin{tikzpicture}
    [scale=.95, v/.style={node distance=2cm, rectangle, draw}, d/.style={node distance=1cm}
    every node/.append style={transform shape}]
      \node(p1)[v]{$f_1$}; \node(p2)[v, below left of=p1]{$f_2$}; \node(p3)[v, below right of=p1]{$f_3$};
      \node(d1)[d, above of=p1]{$x_1$}; \node(d2)[d, draw, left of=p2]{$\delta_{\overline{x}_{2}}$};
      \node(d3)[d, draw, right of=p3]{$1$}; 
      \path[every node/.style={transform shape}]
      (p1)edge node[left]{$y_1$}(p2) (p2)edge node[below]{$y_2$}(p3)edge(d2) (p3)edge node[right]{$y_3$}(p1)
      (p1)edge(d1) (p3)edge(d3);
  \end{tikzpicture} }
\end{center}
\caption{Inference: (a) An example NFG $\G$ representing $p_{X_{1}X_2X_3}(x_1,x_2,x_3)$, (b) the resulting $\G^{*}$ assuming we are interested
in $\sum_{x_3} p_{X_{1}X_2X_3}(x_1,x_2,x_3)|_{x_2 = \overline{x}_2}$.}
%Clearly, the two NFGs (b) and (c) are equivalent.}
  \label{fig:gstar}
\end{figure}

Clearly the inference problem is encoded in $\G^{*}$, and hence reduces to computing the exterior
function of $\G^{*}$, which can be performed by invoking the elimination algorithm on $\G^{*}$. From this
equivalence between inference and the computation of the exterior function, one can always assume
that the given NFG represents the desired computation, i.e., one can assume that the NFG is already reduced to the
desired inference $(.)^{*}$ form.

%\subsection{The constrained NFG model}
%\noindent \textbf{The constrained model:}
We remark that
in a constrained model, by Proposition~\ref{prop:Const_Eq}, we may assume that each interface node is an equality indicator.
Hence, for each evaluated RV, i.e., for each
$i \in N \subseteq I$, we may 1) for each neighbor $j \in {\rm ne}(i)$ of $i$, connected to
$i$ by edge $e$, replace $f_{j}$ with $f_{j}(x_{E(j)})|_{x_e=\overline{x}_{e}}$ and delete $e$, and 2) delete node $i$.
Hence, the inference problem over constrained models is simply a marginalization one, and one may always assume
$N$ is empty.

%\subsection{The generative model}
On the other hand, for a conditional function of $x$ given $y$, we have $\sum_{x}f(x,y)$ is a constant $c$ independent of $y$.
Hence, in a generative model $\G$ with vertex set $I \cup J$, for each marginalized RV, i.e., for each $i \in M \subseteq I$, 
we may 1) absorb the constant $c_i$ into one of the neighbors of $i$ 
by replacing $f_j$ with $c_i f_j$ for some $j \in {\rm ne}(i)$,
2) for each neighbor $j \in \N(i)$ of $i$, connected to $i$ by edge $e$, replace $f_j$ with
$\sum_{x_e} f_{j}(x_{E(j)})$ and delete $e$, and 3) delete node $i$.
Hence, the inference problem over generative models is simply an evaluation one, and one may always assume $M$ is empty.

%% file: conc.tex
\section{Concluding Remarks}
\label{section:conc}
In this paper we presented NFGs as a new class of probabilistic graphical models. We showed that 
this framework and the transformation technique herein unify various previous models, 
including multiplicative models, convolutional factor graphs, and the known transform-domain models.
 We focused on two dual categories of NFG models, constrained and generative NFG models, and revealed an interesting duality in their implied dependence structure. 

We feel that approaching learning and inference problems in a transform domain is methodologically appealing. The generic and flexible transformation technique introduced in this framework may potentially demonstrate great power along that direction.

%Although we did not touch on the algorithm aspect of the framework, it is in fact straight-forward to extend the existing algorithms (such as belief propagation) to this context. It is also possible to extend this framework to modelling continuous random variables, which is merely a matter of technicality.

%In this paper,  we have been primarily concerned with functions over finite alphabets and with finitely valued random variables. We remark however that there is in fact
%little difficulty in extending this framework to accommodate continuous random variables,
%where appropriate Dirac measures will replace the indicator functions, and one may consider generalized functions or operators in addition to standard functions. 

 It is our hope that the idea and modelling framework presented in this paper find new applications beyond the reach of conventional models.